\documentclass[
UKenglish,
 cleveref,
 autoref, 
thm-restate
]{lipics-v2021}
\nolinenumbers
\newcommand{\typeof}{1} %
\newcommand{\longversion}[1]{\ifthenelse{\equal{\typeof}{0}}{}{#1}\ignorespaces}
\newcommand{\shortversion}[1]{\ifthenelse{\equal{\typeof}{0}}{#1}{}\ignorespaces}
\newcommand{\longshortversion}[2]{\ifthenelse{\equal{\typeof}{0}}{#2}{#1}}

\usepackage[T1]{fontenc}
\usepackage{graphicx}
\usepackage{amssymb}
\usepackage{amsmath}
\usepackage{amsthm}
\usepackage{thmtools}
\usepackage{bm}
\usepackage{prftree}
\usepackage{stmaryrd}
\usepackage{cmll}
\usepackage{caption}
\usepackage{subcaption}

\usepackage{tikz}
\usetikzlibrary{cd}
\usetikzlibrary{decorations.markings}%
\tikzset{
  ->-/.style = { decoration = { markings, mark=at
      position #1 with {\arrow{>}} }, postaction =
    {decorate} }, ->-/.default = 0.5 }

\tikzset{ -<-/.style = { decoration = { markings,
      mark=at position #1 with {\arrow{<}}},
    postaction = {decorate} }, -<-/.default = 0.5
}

\tikzset{rect/.style = {
    draw,
    minimum width=15,
    minimum height=15
  }  
}

\longversion{
  \newtheorem{notation}{\textbf{Notation}}[section]
}

\newcommand{\metcppo}{\mathbf{MetCppo}}
\newcommand{\sem}[1]{\llbracket#1\rrbracket}
\newcommand{\psem}[1]{\llparenthesis#1\rrparenthesis}
\newcommand{\letin}[3]{\mathbf{let} \ #1 \
  \mathbf{be} \ #2 \ \mathbf{in} \ #3}
\newcommand{\fix}[2]{\mathbf{fix}_{#1}(#2)}
\newcommand{\LL}[1]{\Lambda_{#1}}
\newcommand{\LLL}[1]{\Lambda_{#1}^{!}}
\newcommand{\lrec}[1]{\Lambda_{\mathrm{rec}}}
\newcommand{\bbR}{\mathbb{R}}
\newcommand{\bbRR}{\mathbb{R}_{\geq 0}^{\infty}}

\newcommand{\dobs}{d^{\mathrm{obs}}}
\newcommand{\dext}{d^{\mathrm{log}}}
\newcommand{\deq}{d^{\mathrm{equ}}}
\newcommand{\dden}{d^{\mathrm{den}}}
\newcommand{\dint}{d^{\mathrm{int}}}
\newcommand{\term}{\mathbf{Term}}
\newcommand{\fterm}{\mathbf{Term}^{!}}
\newcommand{\type}{\mathbf{Ty}}
\newcommand{\env}{\mathbf{Env}}

\newcommand{\val}{\mathbf{Value}}
\newcommand{\fval}{\mathbf{Value}^{!}}
\newcommand{\tr}{\mathrm{tr}}

\newcommand{\const}[1]{\overline{#1}}
\newcommand{\abs}[1]{|#1|}

\newcounter{numberone}
\newenvironment{varenumerate}
{
\begin{list}{\arabic{numberone}.}
{
  \usecounter{numberone}
  \setlength{\itemsep}{0pt}
  \setlength{\topsep}{0pt}
  \setlength{\parsep}{0pt}
  \setlength{\partopsep}{0pt}
  \setlength{\leftmargin}{15pt}
  \setlength{\rightmargin}{0pt}
  \setlength{\itemindent}{0pt}
  \setlength{\labelsep}{5pt}
  \setlength{\labelwidth}{15pt}
}}
{
\end{list} 
}
\declaretheorem[style=definition,name=Example,qed=$\square$]{myexample}

\newcommand{\XX}{x}

\title{On the Lattice of Program 
  Metrics\footnote{The first
    and third authors are
    partially supported by the
    ERC CoG DIAPASoN, GA
    818616.}}


\author{Ugo Dal Lago}{University of Bologna \& INRIA Sophia Antipolis}{ugodallago@unibo.it}{}{}
\author{Naohiko Hoshino}{Sojo University}{nhoshino@cis.sojo-u.ac.jp}{}{}
\author{Paolo Pistone}{University Roma Tre}{paolo.pistone@uniroma3.it}{}{}

\hideLIPIcs

\begin{document}
%
%
%
%
\authorrunning{U. Dal Lago et al.}
%
%
\maketitle 


\begin{abstract}
  In this paper we are concerned with understanding the nature
  of program metrics for calculi with higher-order
  types, seen as natural generalizations of
  program equivalences. Some of the metrics we are
  interested in are well-known, such as those based
  on the interpretation of terms in metric spaces
  and those obtained by generalizing observational
  equivalence. We also introduce a new one, called
  the interactive metric, built by applying the
  well-known Int-Construction to the category of
  metric complete partial orders. Our aim is then
  to understand how these metrics relate to each
  other, i.e., whether and in which cases one such
  metric refines another, in analogy
  with corresponding well-studied problems
  about program equivalences. The results we obtain
  are twofold. We first show that the metrics of
  semantic origin, i.e., the denotational and
  interactive ones, lie \emph{in between} the
  observational and equational metrics and that in
  some cases, these inclusions are strict. Then, we
  give a result about the relationship between the
  denotational and interactive metrics, revealing that the former is less discriminating
  than the latter. All our results are given for a
  linear lambda-calculus, and some of them can be
  generalized to calculi with graded comonads, in
  the style of Fuzz.
\end{abstract}

\section{Introduction}
Program equivalence is one of the most important concepts in the semantics of 
programming languages: every way of giving semantics to programs induces a 
notion of equivalence, and the various notions of equivalence available for the 
same language, even when very different from each other, help us understanding 
the deep nature of the language itself. Indeed, there is not \emph{one} single, 
preferred way to construct a notion of equivalence for programs. The latter is 
especially true in presence of higher-order types or in scenarios in which 
programs have a fundamentally interactive behavior, e.g. in process algebras. 
For example, the relationship between  
observational equivalence,  the most 
coarse-grained congruence relation among those which are coherent 
with the underlying notion of observation, and denotational semantics has led in some cases to so-called full-abstraction results (e.g.~\cite{Hyland:2000aa, 10.1145/3164540}), which are known to hold only for 
\emph{some} denotational models and in \emph{some} programming languages. A 
similar argument applies to applicative bisimularity, which, e.g., is indeed 
fully abstract in presence of \emph{probabilistic} effects \cite{DBLP:conf/esop/CrubilleL14, DBLP:conf/birthday/CrubilleLSV15} but not so in 
presence of \emph{nondeterministic} effects \cite{LassenPhD}.

Equivalences, although central to the theory of programming languages, do not 
allow us to say anything about all those pairs of programs which, while 
qualitatively exhibiting different behaviors, behave \emph{similarly} in a 
quantitative sense. This has led to the study of notions of \emph{distance} 
between programs, which often take the form of (pseudo-)metrics on the 
space of programs or their denotations. In this sense we can distinguish at 
least three defining styles:
\begin{itemize}
\item
	First, observational equivalence can be generalized to a metric, 
	maintaining the intrinsic quantification across all contexts, but observing 
	a difference rather than an equality~\cite{DLC2015,DLC2016}.
\item
	There is also an approach obtained by generalizing equational logic, 
	recently introduced by Mardare et al.~\cite{Mardare2016}, which has been adapted to 
	higher-order computations with both linear~\cite{DBLP:conf/csl/DahlqvistN22} and non-linear~\cite{DLHLP2022} types.
\item	
Finally, linear calculi admit a denotational interpretation in the 
	category of metric complete partial orders~\cite{10.1145/3009837.3009890}, and this is well-known 
	to work well in presence of graded comonads.
\end{itemize}
In other words, various definitional styles for program equivalences for 
higher-order calculi have been proved to have a meaningful metric counterpart, at least when the underlying type system is based on linear or graded 
types. There is a missing tale in this picture, however, namely the one 
provided by interactive semantic models akin to game semantics and the geometry 
of interaction \cite{GoI1}, which were key ingredients towards the 
aforementioned
full-abstraction results. Moreover, the relationship between the various 
notions of distance in the literature has been studied only superficially, and 
the overall situation is currently less clear than for program equivalences.

The aim of this work is to shed light on the landscape about metrics in 
higher-programs. Notably, a new metric between programs inspired by Girard's 
geometry of interaction \cite{GoI1} is defined, being obtained by applying the so-called 
Int-construction \cite{jsv, nfg} to the category of metric complete partial orders. The result is a denotational model, which, while fundamentally 
different from existing metric models, provides a natural way to measure the 
distance between programs, which we will call the \emph{interactive metric}. In 
the interactive metric, differences between two programs can be observed 
incrementally, by interacting with the underlying denotational interpretation 
in the question-answer protocol typical of game semantics and the geometry of 
interaction.

Technically, the main part of the work is an in-depth study of the 
relationships between the various metrics existing in the literature, including 
the interactive metric. Overall, the result of this analysis is the one in 
Figure \ref{fig1}. The observational metric remains the least discriminating, while the 
equational metric is proved to be the one assigning the greatest distances 
to (pairs of) programs. The two metrics of a semantic nature, namely the 
denotational one and the interactive one, stand in between the two metrics 
mentioned above, with the interactive metric being more discriminating than the denotational one.

\begin{figure}
\begin{center}
\includegraphics[scale=0.8]{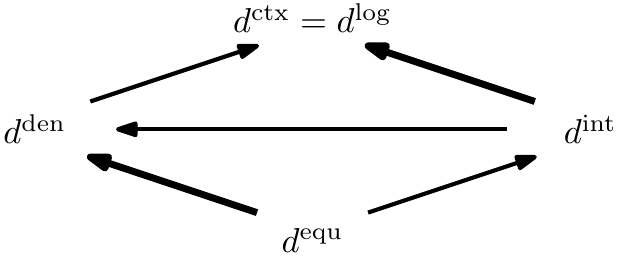}
\end{center}
\caption{Illustration of our comparison results for program metrics: an arrow $d^{\mathbf a} \to d^{\mathbf b}$ indicates that 
$d^{\mathbf b}$ is \emph{coarser} (i.e.~less discriminating) than $d^{\mathbf a}$. Thick arrows indicate \emph{strict} domination.}
\label{fig1}
\end{figure}


The remainder of this manuscript is structured as follows.
After recalling some basic facts about metric spaces in Section 2, in Section 3 
we introduce a basic linear programming language over the reals and its 
associated notion of program metrics; in Section 4 we discuss the logical 
relation metric and the observational metric; in Section 5 we discuss the 
equational metric; in Section 6 we introduce the two denotational metrics; 
Sections 7 and 8 contain our main comparison results, and in Section 9 we 
\shortversion{shortly } discuss the case of graded exponentials.
\shortversion{Some proofs are omitted and can be found in an extended version
of this paper \cite{longversion}.}

\section{Preliminaries}
\label{sec:preliminaries}

In this section, we recall the notions of extended
pseudo-metric spaces and non-expansive functions.
Let $\bbRR$ be the set
$\{a \in \bbR \mid a \geq 0\} \cup \{\infty\}$ of
non-negative real numbers and infinity. An
\emph{extended pseudo-metric space} $X$ consists
of a set $|X|$ and a function
$d_{X} \colon |X| \times |X| \to \bbRR$ satisfying
the following conditions:
\begin{itemize}
\item For all $x \in |X|$, we have
  $d_{X}(x,x) = 0$;
\item For all $x,y \in |X|$, we have
  $d_{X}(x,y) = d_{X}(y,x)$;
\item For all $x,y,z \in |X|$, we have
  $d_{X}(x,z) \leq d_{X}(x,y) + d_{X}(y,z)$.
\end{itemize}
In the sequel, we simply refer to extended
pseudo-metric spaces as metric spaces, and we
denote the underlying set $|X|$ by $X$.

For metric spaces $X$ and $Y$, a function
$f \colon X \to Y$ is said to be
\emph{non-expansive} when for all $x,x' \in X$, we
have $d_{Y}(fx,fx') \leq d_{X}(x,x')$. We write
$\mathbf{Met}$ for the category of metric spaces
and non-expansive functions. The category
$\mathbf{Met}$ has a symmetric monoidal closed
structure $(1,\otimes,\multimap)$ where the metric
of the tensor product $X \otimes Y$ is given by
\begin{equation*}
  d_{X \otimes Y}((x,y),(x',y'))
  = d_{X}(x,x') + d_{Y}(y,y').
\end{equation*}
We suppose that the monoidal product is left
associative, and we denote the $n$-fold monoidal
product of $X$ by $X^{\otimes n}$. In the sequel,
$\bbR$ denotes the metric space of real numbers
equipped with the absolute distance
$d_{\bbR}(a,b) = |a - b|$.

\section{A Linear Programming Language}
\label{sec:line-progr-lang}

\subsection{Syntax and Operational Semantics}
\label{sec:syn}


We introduce our target language that is a linear
lambda calculus equipped with constant symbols for
real numbers and non-expansive functions. We fix a
set $S$ of non-expansive functions
$f \colon \bbR^{\otimes n} \to \bbR$ with
$n \geq 1$. We call $n$ the \emph{arity} of $f$.
For example, $S$ may include addition
$+ \colon \mathbb{R} \otimes \mathbb{R} \to
\mathbb{R}$ and trigonometric functions such as
$\sin, \cos \colon \mathbb{R} \to \mathbb{R}$. We
assume function symbols $\const{f}$ for $f \in S$
and constant symbols $\const{a}$ for real numbers
$a \in \bbR$.

Our language, denoted by $\LL{S}$, is given as
follows. Types and environments are given by
\begin{equation*}
  \mathrm{Types}
  \quad
  \tau,\sigma := \mathbf{R} \mid
  \mathbf{I} \mid
  \tau \multimap \sigma \mid
  \tau \otimes \sigma,
  \qquad
  \mathrm{Environments}
  \quad
  \Gamma,\Delta :=
  \varnothing \mid
  \Gamma , x : \tau.
\end{equation*}
We denote the set of types by $\type$ and denote
the set of environments by $\env$. We always
suppose that every variable appears at most once
in any environment. For environments $\Gamma$ and
$\Delta$ that do not share any variable, we write
$\Gamma \# \Delta$ for a \emph{merge}
\cite{dill,10.5555/519939} of $\Gamma$ and
$\Delta$, that is an environment obtained by
shuffling variables in $\Gamma$ and $\Delta$
preserving the order of variables in $\Gamma$ and
the order of variables in $\Delta$. For example,
$(x:\tau,y:\sigma,y':\sigma',x':\tau')$ is a merge
of $(x:\tau,x':\tau')$ and
$(y:\sigma,y':\sigma')$. \longversion{Formally, an
  environment $\Xi$ is said to be a merge of
  $\Gamma$ and $\Delta$ when
  \begin{itemize}
  \item $\Xi$, $\Gamma$ and $\Delta$ are equal to
    the empty environment; or
  \item $\Gamma = \Gamma',x:\tau$ and there is a
    merge $\Xi'$ of $\Gamma'$ and $\Delta$ such
    that $\Xi = \Xi',x:\tau$; or
  \item $\Delta = \Delta',x:\tau$ and there is a
    merge $\Xi'$ of $\Gamma$ and $\Delta'$ such
    that and $\Xi = \Xi',x:\tau$.
  \end{itemize}} When we write $\Gamma \# \Delta$,
we implicitly suppose that no variable is shared by
$\Gamma$ and $\Delta$. Terms, values and contexts
are given by the following BNF.
\begin{equation*}
  \begin{array}{ll}
    \!
    \mathrm{Terms}
    &
      M,N := x \in \mathbf{Var}
      \mid
      \const{a}  \mid
      \ast \mid
      \const{f}(M_{1},
      \ldots,M_{\mathrm{ar}(f)}) \mid
      M\,N \mid
      \lambda x:\tau.\,M \mid
    \\
    & \phantom{M,N :=}
      M \otimes N \mid 
      \letin{\ast}{M}{N} \mid
      \letin{x \otimes y}{M}{N} \\
    \!
    \mathrm{Values}
    &
      V,U := \const{a}
      \mid \ast 
      \mid \lambda x:\tau.\, M
      \mid V \otimes U
    \\
    \!
    \mathrm{Contexts}
    &
      C[-] :=
      [-] \mid \const{f}(M,\ldots,M',
      C[-],N',\ldots,N) \mid
      C[-]\,M \mid M\,C[-] \mid
      \lambda x:\tau.\, C[-] \mid\\
    & \phantom{C[-]::=}\hspace{-0.8pt}
      C[-] \otimes M \mid
      M \otimes C[-] \mid
      \letin{\ast}{C[-]}{M} \mid 
      \letin{\ast}{M}{C[-]} \mid \\
    &  \phantom{C[-]::=}
      \letin{x \otimes y}{C[-]}{M} \mid
      \letin{x \otimes y}{M}{C[-]}
  \end{array}
\end{equation*}
Here, $a$ ranges over $\bbR$, $f$ ranges over $S$,
and $x$ ranges over a countably infinite set
$\mathbf{Var}$ of variables. We write
$\Gamma \vdash M :\tau$ when the typing judgement
is derived from the rules given in
Figure~\ref{fig:typing_rules}. Evaluation rules
are given in Figure~\ref{fig:evaluation_rules}.
Since $\LL{S}$ is a purely linear programming
language, for any closed term $\vdash M:\tau$,
there is a value $\vdash V:\tau$ such that
$M \hookrightarrow V$. For an environment $\Gamma$
and a type $\tau$, we define $\term(\Gamma,\tau)$
to be the set of all terms $M$ such that
$\Gamma \vdash M:\tau$, and we define $\val(\tau)$
to be the set of closed values of type $\tau$. We
simply write $\term(\tau)$ for
$\term(\varnothing,\tau)$, that is the set of
closed terms of type $\tau$. For a context $C[-]$,
we write
$C[-] \colon (\Gamma,\tau) \to (\Delta,\sigma)$
when for all terms $\Gamma \vdash M:\tau$, we have
$\Delta \vdash C[M]:\sigma$.

We adopt Church-style lambda abstraction so that
every type judgement $\Gamma \vdash M : \tau$ has
a unique derivation, which makes it easier to
define denotational semantics for $\LL{S}$. Except
for this point, our language can be understood as
a fragment of Fuzz
\cite{10.1145/1932681.1863568}--- the typing
judgment $x:\sigma , \ldots, y:\rho \vdash M:\tau$
corresponds to
$x:_{1}\sigma , \ldots, y:_{1}\rho \vdash M:\tau$
in Fuzz. In Section~\ref{sec:grade}, we discuss
extending our results in this paper to a richer
language, closer to the one from
\cite{10.1145/1932681.1863568}.


  \begin{figure}[t]
  \centering
  \begin{equation*}
    \prftree{
    }{
      x : \tau
      \vdash x:\tau
    }
    \qquad
    \prftree{
      a \in \bbR
    }{
      \vdash  \const{a} : \mathbf{R}
    }
    \qquad
    \prftree{
    }{
      \vdash \ast: \mathbf{I}
    }
    \qquad
    \prftree{
      f \in S
    }{
      \Gamma_{1} \vdash M_{1}:\mathbf{R}
    }{
      \ldots
    }{
      \Gamma_{\mathrm{ar}(f)} \vdash M_{\mathrm{ar}(f)}:\mathbf{R}
    }{
      \Gamma_{1} \# \cdots \# \Gamma_{\mathrm{ar}(f)}
      \vdash \const{f} (M_{1},\ldots,M_{\mathrm{ar}(f)}):\mathbf{R}
    }
  \end{equation*}
  \begin{equation*}
    \prftree{
      \Gamma , x : \sigma  \vdash M:\tau
    }{
      \Gamma
      \vdash
      \lambda x:\sigma.\,M : \sigma \multimap \tau
    }
    \qquad
    \prftree{
      \Gamma \vdash M : \sigma \multimap \tau
    }{
      \Delta \vdash N : \sigma
    }{
      \Gamma \# \Delta \vdash
      M\,N : \tau
    }
    \qquad
    \prftree{
      \Gamma \vdash M : \tau
    }{
      \Delta \vdash N : \sigma
    }{
      \Gamma \# \Delta \vdash
      M \otimes N : \tau \otimes \sigma
    }
  \end{equation*}
  \begin{equation*}
    \prftree{
      \Gamma \vdash M : \mathbf{I}
    }{
      \Delta \vdash N : \tau
    }{
      \Gamma \# \Delta
      \vdash
      \letin{\ast}{M}{N} : \tau
    }
    \qquad
    \prftree{
      \Gamma \vdash M : \sigma_{1} \otimes \sigma_{2}
    }{
      \Delta,x:\sigma_{1},y:\sigma_{2} \vdash N:\tau
    }{
      \Gamma \# \Delta
      \vdash
      \letin{x \otimes y}{M}{N}:\tau
    }
  \end{equation*}
  \caption{Typing Rules}
  \label{fig:typing_rules}
\end{figure}

\begin{figure}[t]
  \centering
  \begin{equation*}
    \prftree{
      V
      \hookrightarrow
      V
    }
    \qquad
    \prftree{
      M_{1}
      \hookrightarrow
      \const{a_{1}}
    }{
      \ldots
    }{
      M_{n}
      \hookrightarrow
      \const{a_{n}}
    }{
      \const{f}(M_{1},\ldots,M_{n})
      \hookrightarrow
      \const{f(a_{1},\ldots,a_{n})}
    }
    \qquad
    \prftree{
      M
      \hookrightarrow
      \lambda x:\tau.\,L
    }{
      N \hookrightarrow V
    }{
      L[V/x]
      \hookrightarrow U
    }{
      M\,N
      \hookrightarrow
      U
    }
  \end{equation*}
  \begin{equation*}
    \prftree{
      M \hookrightarrow V
    }{
      N \hookrightarrow U
    }{
      M \otimes N
      \hookrightarrow
      V \otimes U
    }
    \qquad
    \prftree{
      M \hookrightarrow \ast
    }{
      N \hookrightarrow V
    }{
      \letin{\ast}{M}{N}
      \hookrightarrow
      V
    }
    \qquad
    \prftree{
      M
      \hookrightarrow
      V \otimes U
    }{
      N[V/x,U/y]
      \hookrightarrow
      W
    }{
      \letin{x \otimes y}{M}{N}
      \hookrightarrow
      W
    }
  \end{equation*}
  \caption{Evaluation Rules}
  \label{fig:evaluation_rules}
\end{figure}

\subsection{Equational Theory}

In this paper we consider an equational theory
for $\LL{S}$, which will turn out to be instrumental to define a
notion of well-behaving family of metrics for
$\LL{S}$ called admissibility
(Section~\ref{sec:adm}) and to give a quantitative
equational theory for $\LL{S}$
(Section~\ref{sec:eq}). In both cases, if two terms
are to be considered equal, then the distance
between them is required to be $0$. Here, we adopt
the standard equational theory for the linear
lambda calculus \cite{Mackie1993AnIL} extended
with the following axiom
\begin{equation*}
  \prftree{
    f \in S
  }{
    f(a_{1},\ldots,a_{\mathrm{ar}(f)}) = b
  }{
    \vdash \const{f}(
    \const{a_{1}},
    \ldots,
    \const{a_{\mathrm{ar}(f)}}
    ) = \const{b} : \tau
  }.
\end{equation*}
For terms $\Gamma \vdash M : \tau$ and
$\Gamma \vdash N : \tau$, we write
$\Gamma \vdash M = N : \tau$ when the equality is
derivable.

We may add some other axioms to the equational
theory as long as the axioms are valid when we
interpret function symbols $\const{f}$ as $f$ and
constant symbols $\const{a}$ as $a$. For example,
when
$\mathrm{add} \colon \mathbb{R} \otimes \mathbb{R}
\to \mathbb{R}$ is in $S$, we may add the
commutativity law
\begin{math}
  x:\mathbf{R},y:\mathbf{R} \vdash
  \overline{\mathrm{add}}(x,y) =
  \overline{\mathrm{add}}(y,x) : \mathbf{R}
\end{math}
to the equational theory. The rest of this paper
is not affected by such extensions to the
equational theory.

\longversion{
  \begin{figure}
    \centering
    \begin{equation*}
      \prftree{
        \Gamma \vdash M : \tau      
      }{
        \Gamma \vdash M = M :\tau      
      }
      \qquad
      \prftree{
        \Gamma \vdash M = N: \tau
      }{
        \Gamma \vdash N = M :\tau      
      }
      \qquad
      \prftree{
        \Gamma \vdash M = N: \tau
      }{
        \Gamma \vdash N = L: \tau
      }{
        \Gamma \vdash M = L :\tau      
      }
    \end{equation*}
    \begin{equation*}
          \prftree{
        f(a_{1},\ldots,a_{n}) = b
      }{
        \vdash \const{f}(
        \const{a_{1}},
        \ldots,
        \const{a_{\mathrm{ar}(f)}}
        ) = \const{b} : \tau
      }
      \qquad
      \prftree{
        \Gamma \vdash M = N:\tau
      }{
        \Delta \vdash C[M] : \sigma
      }{
        \Delta \vdash C[N] : \sigma
      }{
        \Delta \vdash C[M] = C[N] :\sigma
      }
    \end{equation*}
    \begin{equation*}
      \prftree{
        \Gamma,x:\tau \vdash M:\sigma
      }{
        \Delta \vdash N:\tau
      }{
        \Gamma \# \Delta \vdash (\lambda x:\tau.\,M)\,N
        = M[N/x]:\sigma
      }
      \qquad
      \prftree{
        \Gamma \vdash M:\tau \multimap \sigma
      }{
        \Gamma \vdash \lambda x:\tau.\,M\,x
        = M:\tau \multimap \sigma
      }
    \end{equation*}
    \begin{equation*}
      \prftree{
        \Gamma \vdash M:\tau
      }{
        \Gamma \vdash
        \letin{\ast}{\ast}{M}
        = M : \tau
      }
      \qquad
      \prftree{
        \Gamma \vdash
        \letin{x \otimes y}{M \otimes N}{L}
        : \tau 
      }{
        \Gamma
        \vdash
        \letin{x \otimes y}{M \otimes N}{L}
        =
        L[M/x,N/y] : \tau
      }
    \end{equation*}
    \begin{equation*}
      \prftree{
        \Gamma \vdash M:\mathbf{I}
      }{
        \Gamma \vdash
        \letin{\ast}{M}{\ast}
        = M : \mathbf{I}
      }
      \qquad
      \prftree{
        \Gamma \vdash M : \tau \otimes \sigma
      }{
        \Gamma
        \vdash
        \letin{x \otimes y}{M}{x \otimes y}
        =
        M : \tau \otimes \sigma
      }
    \end{equation*}
    \begin{equation*}
      \prftree{
        \Gamma \vdash \letin{\ast}{M}{C[N]}:\tau
      }{
        \Gamma \vdash \letin{\ast}{M}{C[N]}
        =
        C[\letin{\ast}{M}{N}]
        :\tau
      }
    \end{equation*}
    \begin{equation*}
      \prftree{
        \Gamma \vdash \letin{x \otimes y}{M}{C[N]}:\tau
      }{
        \Gamma \vdash \letin{x \otimes y}{M}{C[N]}
        =
        C[\letin{x \otimes y}{M}{N}]
        :\tau
      }
    \end{equation*}
    \caption{Derivation Rules of Equational Theory
      for $\LL{S}$}
    \label{fig:equational_theory}
  \end{figure}
}

\subsection{Admissibility}
\label{sec:adm}

Let us call a family
$\{d_{\Gamma,\tau}\}_{\Gamma \in \env,\tau \in
  \type}$ in which $d_{\Gamma,\tau}$ is a metric
on $\term(\Gamma,\tau)$ a \emph{metric on
  $\LL{S}$}. We introduce a class of metrics on
$\LL{S}$, which is the object of study of this
paper.

\begin{definition}[Admissible Metric]\label{def:adm}
  Let
  $\{d_{\Gamma,\tau}\}_{\Gamma \in \env, \tau \in
    \type}$ be a metric on $\LL{S}$. We say that
  $\{d_{\Gamma,\tau}\}_{\Gamma \in
    \env,\tau:\type}$ is \emph{admissible} when
  the following conditions hold:
  \begin{varenumerate}
  \item[(A1)] For any environment $\Gamma$, any
    type $\tau$, any pair of terms
    $\Gamma \vdash M : \tau$,
    $\Gamma \vdash N : \tau$ and any context
    $C[-] \colon (\Gamma,\tau) \to
    (\Delta,\sigma)$, we have
    \begin{math}
      d_{\Delta,\sigma}(C[M],C[N])
      \leq
      d_{\Gamma,\tau}(M,N).  
    \end{math}
  \item[(A2)] For all $a,b \in \bbR$,
    we have
    \begin{math}
      d_{\varnothing,\mathbf{R}}(a,b) = |a - b|.
    \end{math}
  \item[(A3)] For all
    $a_{1},\ldots,a_{n}, b_{1},\ldots,b_{n} \in
    \bbR$ and all closed values $\vdash V : \tau$
    and $\vdash U : \tau$, we have
    \begin{equation*}
      d_{\varnothing, \mathbf{R}^{\otimes n} \otimes \tau}
      \left(
        \const{a_{1}} \otimes \cdots \otimes
        \const{a_{n}}
        \otimes V,
        \const{b_{1}} \otimes \cdots \otimes
        \const{b_{n}}
        \otimes U
      \right) 
      \geq
      |a_{1}- b_{1}| + \cdots +
      |a_{n} - b_{n}|
      .
    \end{equation*}
  \item[(A4)] If $\Gamma \vdash M = N: \tau$, then
    $d_{\Gamma,\tau}(M,N) = 0$.
  \end{varenumerate}
\end{definition}

The first condition (A1) states that all contexts
are non-expansive, and the second condition (A2)
states that the metric on $\mathbf{R}$ coincides
with the absolute metric on $\mathbb{R}$. (A3)
states that the distance between two terms
$\const{a_{1}} \otimes \cdots \otimes
\const{a_{n}} \otimes V$ and
$\const{b_{1}} \otimes \cdots \otimes
\const{b_{n}} \otimes U$ is bounded (from below)
by the distance between their ``observable
fragments''
$d_{\bbR^{\otimes n}}
((a_{1},\ldots,a_{n}),(b_{1},\ldots,b_{n}))$. The
last condition (A4) states that $d_{\Gamma,\tau}$
subsumes the equational theory for $\LL{S}$.

The definition of admissibility is motivated by
the study of Fuzz \cite{10.1145/1932681.1863568},
which is a linear type system for verifying
differential privacy
\cite{10.1145/1065167.1065184}. There, Reed and
Pierce introduce a syntactically defined metrics
on Fuzz using a family of relations called metric
relations, and they prove that all programs are
non-expansive with respect to the syntactic metric
(Theorem~6.4 in \cite{10.1145/1932681.1863568}).
(A1) is motivated by this result. Furthermore, in
the definition of the metric relation, the tensor
product of types is interpreted as the monoidal
product of metric spaces, and the type of real
numbers is interpreted as $\mathbb{R}$ with the
absolute distance. (A2) and (A3) are motivated by
these definitions. In fact, given an admissible
metric
$\{d_{\Gamma,\tau}\}_{\Gamma \in \env,\tau \in
  \type}$ on $\LL{S}$, we can show that
$d_{\varnothing,\mathbf{R}^{\otimes n}}$ coincides
with the metric of $\bbR^{\otimes n}$.
\begin{lemma}
  If a metric
  $\{d_{\Gamma,\tau}\}_{\Gamma \in \env,\,\tau \in
    \type}$ is admissible, then for all
  $a_{1},b_{1},\ldots,a_{n},b_{n} \in \bbR$,
  \begin{equation}\label{eq:L1}
    d_{\varnothing,\mathbf{R}^{\otimes n}}
    (
    \const{a_{1}} \otimes \cdots \otimes \const{a_{n}},
    \const{b_{1}} \otimes \cdots \otimes \const{b_{n}}
    )
    = |a_{1} - b_{1}| + \cdots +
    |a_{n} - b_{n}|.
  \end{equation}
\end{lemma}
\longversion{
  \begin{proof}
    By (A1) and (A3),
    \begin{align*}
      \sum_{1 \leq i \leq n}
      |a_{i} - b_{i}|
      &\leq 
      d_{\varnothing,\mathbf{R}^{\otimes n} \otimes \mathbf{I}}
      \left(
        \const{a_{1}} \otimes \cdots \otimes \const{a_{n}}
        \otimes \ast,
        \const{b_{1}} \otimes \cdots \otimes \const{b_{n}}
        \otimes \ast
      \right)
      \\
      &\leq
      d_{\varnothing,\mathbf{R}^{\otimes n}}
      \left(
        \const{a_{1}} \otimes \cdots \otimes \const{a_{n}},
        \const{b_{1}} \otimes \cdots \otimes \const{b_{n}}
      \right).
    \end{align*}
    The other inequality follows from (A1), (A2)
    and triangle inequalities:
    \begin{align*}
      \sum_{1 \leq i \leq n}
      |\const{a_{i}} - \const{b_{i}}|
      &\geq
      d_{\varnothing,\mathbf{R}^{\otimes n}}
      \left(
        \const{a_{1}} \otimes \const{a_{2}} \otimes \cdots \otimes \const{a_{n}},
        \const{b_{1}} \otimes \const{a_{2}} \otimes \cdots \otimes \const{a_{n}}
      \right) 
      +
      \sum_{2 \leq i \leq n}
      |\const{a_{i}} - \const{b_{i}}| \\
      &\geq
      d_{\varnothing,\mathbf{R}^{\otimes n}}
      \left(
        \const{a_{1}} \otimes \const{a_{2}}
        \otimes \cdots \otimes \const{a_{n}},
        \const{b_{1}} \otimes \const{b_{2}}
        \otimes \const{a_{3}} \otimes \cdots \otimes \const{a_{n}}
      \right)
      +
      \sum_{3 \leq i \leq n}
      | \const{a_{i}} - \const{b_{i}}| \\
      &\geq \cdots \\
      &\geq
      d_{\varnothing,\mathbf{R}^{\otimes n}}
      \left(
        \const{a_{1}} \otimes \const{a_{2}}
        \otimes \cdots \otimes \const{a_{n}},
        \const{b_{1}} \otimes \const{b_{2}}
        \otimes \cdots \otimes \const{b_{n}}
      \right).
    \end{align*}
  \end{proof}
} The reason that we do not take \eqref{eq:L1} as
the third condition of admissibility and instead
rely on the stronger condition (A3) above is that
requiring \eqref{eq:L1} would not allow us to
characterize the observational metric
(Section~\ref{sec:observational-metric}) as the
least admissible metric on $\LL{S}$.

\section{Logical Metric and Observational Metric}
\label{sec:dext_dobs}

We give two syntactically defined metrics on
$\LL{S}$: one is based on logical relations, and
the other is given in the style of Morris
observational equivalence \cite{morrisPhd}. We
then show that the two metrics coincide. This can
be seen as a metric variant of Milner's context
lemma \cite{MILNER19771}.

\subsection{Logical Metric}
\label{sec:dext}

The first metric on $\LL{S}$ is given by means of
a quantitative form of logical relations
\cite{10.1145/1932681.1863568} called \emph{metric
  logical relations}. Here, we directly define
metric logical relations, and then, we define the
induced metric on $\LL{S}$. The metric logical
relations
\begin{equation*}
  \left\{
    (-) \simeq_{r}^{\tau} (-)
    \subseteq
    \term(\tau)
    \times
    \term(\tau)
  \right\}_{\tau \in \type,\, r \in \bbRR}
\end{equation*}
are given by induction on $\tau$ as follows.
\begin{align*}
  M \simeq_{r}^{\mathbf{R}} N
  &
  \iff
  M \hookrightarrow \const{a}
  \text{ and }
  N \hookrightarrow \const{b}
  \text{ and }
  |a - b| \leq r \\
  M \simeq_{r}^{\mathbf{I}} N
  &
  \iff
  M \hookrightarrow \ast
  \text{ and }
  N \hookrightarrow \ast \\  
  M \simeq_{r}^{\tau \otimes \sigma} N
  &\iff
  M \hookrightarrow V \otimes V'
  \text{ and }
  N \hookrightarrow U \otimes U'
  \text{ and } \\
  &\mathrel{\phantom{\iff}}
  \exists s,s' \in \bbRR,\,
  V \simeq_{s}^{\tau} U
  \text{ and }
  V' \simeq_{s'}^{\sigma} U'
  \text{ and }
  s + s' \leq r
  \\
  M \simeq_{r}^{\tau \multimap \sigma} N
  &\iff
  M \hookrightarrow \lambda x:\tau.\,M'
  \text{ and }
  N \hookrightarrow \lambda x:\tau.\,N'
  \text{ and } \\
  &\mathrel{\phantom{\iff}}
  \forall V,U \in \val(\tau),\,
  \text{if} \
  V \simeq^{\tau}_{s} U,
  \ \text{then} \
  M'[V/x] \simeq_{r + s}^{\sigma} N'[U/x]
\end{align*}
Then for an environment
$\Gamma = (x:\sigma,\ldots,y:\rho)$ and a pair of
terms $\Gamma \vdash M:\tau$ and
$\Gamma \vdash N:\tau$, we define
$\dext_{\Gamma,\tau}(M,N) \in \bbRR$ by
\begin{equation*}
  \dext_{\Gamma,\tau}(M,N) =
  \inf \{r \in \bbRR \mid
  \lambda x:\sigma.\,\cdots
  \lambda y:\rho.\,M
  \simeq_{r}^{\sigma \multimap \cdots
  \multimap \rho \multimap \tau}
  \lambda x:\sigma.\,\cdots
  \lambda y:\rho.\,N
  \}.
\end{equation*}

\shortversion{We give a consequence of our results
  in this paper, namely, Theorem~\ref{thm:ctx-ext}
  and Theorem~\ref{thm:main}. For the detail of
  the proof of Proposition~\ref{prop:ext-adm}, see
  \cite{longversion}.}
\begin{proposition}\label{prop:ext-adm}
  For any environment $\Gamma$ and any type
  $\tau$, the function $\dext_{\Gamma,\tau}$ is a
  metric on $\term(\Gamma,\tau)$. Furthermore, 
  $\{\dext_{\Gamma, \tau}\}_{\Gamma \in \env, \tau
    \in \type}$ is admissible.
\end{proposition}
\longversion{\begin{proof} It is straightforward
    to show that $\dobs$ given in the next section
    is a metric on $\LL{S}$ and satisfies (A1).
    Hence, it follows from
    Theorem~\ref{thm:ctx-ext} that
    $\dext_{\Gamma,\tau}$ is a metric on
    $\term(\Gamma,\tau)$ and satisfies (A1). (A2)
    and (A3) follow from the definition of
    $\dext$. The proof of (A4) is given in
    Corollary~\ref{cor:adm}.
  \end{proof}} We call $\dext$ \emph{logical
  metric}. \shortversion{ We note that we can
  directly check that $\dext$ satisfies (A2) and
  (A3), and we need Theorem~\ref{thm:ctx-ext} and
  Theorem~\ref{thm:main} to show that $\dext$ is a
  metric on $\LL{S}$ and satisfies (A1) and (A4).}

\begin{myexample}\label{eg:ma}
  For $a \in \bbR$, we define a term $M_{a}$ to be
  \begin{equation*}
    \vdash \const{a} \otimes \const{a} \otimes
    V
    :
    \mathbf{R} \otimes \mathbf{R} \otimes
    ((\mathbf{R} \otimes \mathbf{R} \multimap
    \mathbf{R}) \multimap \mathbf{R})
    \qquad
    \text{  where
      $V= \lambda k:\mathbf{R} \otimes \mathbf{R}
      \multimap \mathbf{R}.\,
      k\,\const{0}\,\const{0}$.}
  \end{equation*}
  Since
  $\dext_{\varnothing,\mathbf{R}}(\const{0},\const{1})
  = 1$, we obtain
  $\dext_{\varnothing, \mathbf{R} \otimes
    \mathbf{R} \otimes ((\mathbf{R} \otimes
    \mathbf{R} \multimap \mathbf{R}) \multimap
    \mathbf{R})} (M_{0},M_{1}) = 1 + 1 + 0 = 2$.
\end{myexample}

\subsection{Observational Metric}
\label{sec:observational-metric}

We next give a metric, which we call the 
\emph{observational metric}, that measures distances
between terms by observing concrete values
produced by any possible context. For terms
$\Gamma \vdash M :\tau$ and
$\Gamma \vdash N :\tau$, we define
$\dobs_{\Gamma,\tau}(M,N) \in \bbRR$ by
\begin{equation*}
  \dobs_{\Gamma,\tau}(M,N)
  =
  \sup_{(n,\sigma,C[-]) \in \mathcal{K}(\Gamma,\tau)}
  \left\{
    |a_{1}-b_{1}| + \cdots + |a_{n}-b_{n}|
    \,\middle|\hspace{-3pt}
    \begin{array}{l}
      C[M] \hookrightarrow
      \const{a_{1}} \otimes
      \cdots \otimes \const{a_{n}}
      \otimes V
      \\
      \text{and }
      C[N] \hookrightarrow
      \const{b_{1}} \otimes
      \cdots \otimes \const{b_{n}}
      \otimes U      
    \end{array}
    \hspace{-4.5pt}
  \right\}
\end{equation*}
where
$(n,\sigma,C[-]) \in \mathcal{K}(\Gamma,\tau)$ if
and only if $C[-]$ is a context from
$(\Gamma,\tau)$ to
$(\varnothing,\mathbf{R}^{\otimes n} \otimes
\sigma)$.

\shortversion{
  \begin{theorem}\label{thm:ctx-ext}
    For any environment $\Gamma$ and any type
    $\tau$, we have
    $\dobs_{\Gamma,\tau}= \dext_{\Gamma,\tau}$.
  \end{theorem}
  This theorem follows from
  coincidence of the metric logical relations
  with the metric relations and
  the fundamental lemma for metric
  logical relations.}
\begin{myexample}
  We consider the term $\vdash M_{a} :\tau$ given
  in Example~\ref{eg:ma} again. By observing
  $M_{0}$ and $M_{1}$ by the trivial context
  $[-]$, we can directly check that
  $\dobs_{\varnothing, \mathbf{R} \otimes
    \mathbf{R} \otimes ((\mathbf{R} \otimes
    \mathbf{R} \multimap \mathbf{R}) \multimap
    \mathbf{R}) } (M_{0},M_{1}) \geq 2$. (In fact,
  it follows from Theorem~\ref{thm:ctx-ext} that
  the distance is equal to $2$.) The purpose of
  the auxiliary type $\sigma$ in the definition of
  $\mathcal{K}(\Gamma,\tau)$ is to enable
  observations of this type. In this case, while
  the logical metric distinguishes $M_{0}$ from
  $M_{1}$, we can not observationally distinguish
  $M_{0}$ from $M_{1}$ by means of observations at
  types $\mathbf{R}^{\otimes n}$ when
  $S = \emptyset$. \longversion{See
    Proposition~\ref{aprop:two_to_one_arg} for
    impossibility of observational distinction of
    these terms at $\mathbf{R}^{\otimes n}$.}
\end{myexample}

\longversion{\begin{proposition}\label{aprop:two_to_one_arg}
    If $S = \emptyset$, then for any
    $n \in \mathbb{N}$, there is no context
    \begin{equation*}
      C[-] \colon (\varnothing, \mathbf{R} \otimes
      ((\mathbf{R}^{\otimes 2} \multimap \mathbf{R})
      \multimap \mathbf{R})) \to
      (\varnothing,\mathbf{R}^{\otimes n}). 
    \end{equation*}
  \end{proposition}
  \begin{proof}
    We first show that there
    is no closed term of type
    $\mathbf{R}^{\otimes 2} \multimap \mathbf{R}$.
    To see this,
    for each type $\tau$, we inductively define
    $|\tau| \in \mathbf{Z}$ by
    \begin{align*}
      |\mathbf{R}| &= 1,&
      |\mathbf{I}| &= 0, &
      |\tau \otimes \sigma| &= |\tau| + |\sigma|, &
      |\tau \multimap \sigma| &= -|\tau| + |\sigma|.
    \end{align*}
    We extend the definition of $|-|$ to
    environments
    $\Gamma = (x:\tau,\ldots,y:\sigma)$ by letting
    $|\Gamma|$ to be $|\tau| + \cdots + |\sigma|$.
    Then by induction on the derivation of
    $\Gamma \vdash M : \tau$, we can show that if
    $S = \emptyset$, then $|\Gamma| \leq |\tau|$.
    Since
    $|\mathbf{R}^{\otimes 2} \multimap \mathbf{R}|
    = -1$, we see that there is no closed term of
    type
    $\mathbf{R}^{\otimes 2} \multimap \mathbf{R}$.
    We next show the statement. Let us suppose
    that there is a context
    $C[-] : (\varnothing,\mathbf{R} \otimes (
    (\mathbf{R}^{\otimes 2} \multimap \mathbf{R})
    \multimap \mathbf{R})) \to
    (\varnothing,\mathbf{R}^{\otimes n})$ for some
    $n \in \mathbb{N}$, and we derive
    contradiction. Because $\LL{S}$ is
    normalizing, there is a value $V$ such that
    $C[\const{0} \otimes (\lambda
    f:\mathbf{R}^{\otimes 2} \multimap
    \mathbf{R}.\, f\,\const{0}\,\const{0})]
    \hookrightarrow V$. As we have observed, there
    is no closed value
    $U \in \val(\mathbf{R}^{\otimes 2} \multimap
    \mathbf{R})$. Therefor, there is no
    $\beta$-reduction of the form
    $(\lambda f:(\mathbf{R}^{\otimes 2} \multimap
    \mathbf{R}).\, f\,\const{0}\,\const{0})\,U
    \hookrightarrow U\,\const{0}\,\const{0}$
    during the reduction
    $C[\const{0} \otimes (\lambda
    f:(\mathbf{R}^{\otimes 2} \multimap
    \mathbf{R}).\, f\,\const{0}\,\const{0})]
    \hookrightarrow V$. Hence,
    $\lambda f:(\mathbf{R}^{\otimes 2} \multimap
    \mathbf{R}).\, f\,\const{0}\,\const{0}$ must
    be a subterm of $V$, contradicting
    $V \in \val(\mathbf{R}^{\otimes n})$.
  \end{proof}}

\longversion{
  \subsection{Coincidence of the Logical 
    Metric and the Observational Metric}
  \label{sec:coinc-logic-relat}

  This section is devoted to prove that the
  logical metric coincides with the observational
  metric. For the proof, we introduce another
  family of quantitative relations, called
  \emph{metric relations}
  \cite{10.1145/1932681.1863568}. We define the
  metric relations
  \begin{equation*}
    \{
    (-) \cong_{r}^{\tau} (-)
    \subseteq
    \term(\tau)
    \times
    \term(\tau)
    \}_{\tau \in \type, r \in \bbRR}
  \end{equation*}
  by induction on $\tau$ as follows.
  \begin{align*}
    M \cong_{r}^{\mathbf{R}} N
    &
    \iff
    M \hookrightarrow \const{a}
    \text{ and }
    N \hookrightarrow \const{b}
    \text{ and }
    |a - b| \leq r \\
    M \cong_{r}^{\mathbf{I}} N
    &
    \iff
    M \hookrightarrow \ast
    \text{ and }
    N \hookrightarrow \ast \\  
    M \cong_{r}^{\tau \otimes \sigma} N
    &\iff
    M \hookrightarrow V \otimes V'
    \text{ and }
    N \hookrightarrow U \otimes U'
    \text{ and } \\
    &\mathrel{\phantom{\iff}}
    \exists s,s' \in \bbRR,\,
    V \cong_{s}^{\tau} U
    \text{ and }
    V' \cong_{s'}^{\sigma} U'
    \text{ and }
    s + s' \leq r
    \\
    M \cong_{r}^{\tau \multimap \sigma} N
    &\iff
    M \hookrightarrow \lambda x:\tau.\,M'
    \text{ and }
    N \hookrightarrow \lambda x:\tau.\,N'
    \text{ and } \\
    &\mathrel{\phantom{\iff}}
    \forall V \in \val(\tau),\,   
    M'[V/x] \cong_{r}^{\sigma} N'[V/x]
  \end{align*}
  The only difference between the definition of
  $\simeq$ and $\cong$ is in the case of the
  linear function type.

  Let us introduce some notations. For an
  environment
  $\Gamma=(x_{1}:\tau_{1},\ldots,
  x_{n}:\tau_{n})$, we define $\val(\Gamma)$ to be
  $\val(\tau_{1}) \times \cdots \times
  \val(\tau_{n})$. Given $\gamma \in \val(\Gamma)$
  and $\Gamma \vdash M : \tau$, we define
  $\vdash M\gamma:\tau$ in the obvious way. For
  $\gamma = (V_{1},\ldots,V_{n}),\delta
  =(U_{1},\ldots,U_{n}) \in \val(\Gamma)$, we
  write $\gamma \simeq_{r}^{\Gamma} \delta$ when
  there are $s_{1},\ldots,s_{n} \in \bbRR$ such
  that $r \geq s_{1} + \cdots + s_{n}$ and
  $V_{1} \simeq_{s_{1}}^{\tau_{1}} U_{1}, \ldots,
  V_{n} \simeq_{s_{n}}^{\tau_{n}} U_{n}$ hold.

  \begin{lemma}\label{alem:ext->ctx}
    For any environment
    $\Gamma=
    (x_{1}:\tau_{1},\ldots,x_{n}:\tau_{n})$ and
    any pair of terms $\Gamma \vdash M : \tau$ and
    $\Gamma \vdash N : \tau$, if
    $\gamma \in \val(\Gamma)$, then
    $M\gamma
    \cong_{\dobs_{\Gamma,\tau}(M,N)}^{\tau}
    N\gamma$.
  \end{lemma}
  \begin{proof}
    We prove the statement by induction on $\tau$.
    (When $\tau=\mathbf{R}$) Let
    $\gamma=(V_{1},\ldots,V_{n})$ be an element of
    $\val(\Gamma)$. For $a,b \in \mathbb{R}$ such
    that $M\gamma \hookrightarrow \const{a}$ and
    $N\gamma \hookrightarrow \const{b}$, we show
    that
    $|a-b| \leq \dobs_{\Gamma,\mathbf{R}}(M,N)$.
    Let a context $C[-]$ be
    \begin{equation*}
      (\lambda x_{1}:\tau_{1}.\,\cdots\lambda
      x_{n}:\tau_{n}.\,[-] \otimes \ast
      )\,V_{1}\,\cdots\,V_{n}. 
    \end{equation*}
    Then, since we have
    $C[M] \hookrightarrow \const{a} \otimes \ast$
    and
    $C[N] \hookrightarrow \const{b} \otimes \ast$,
    we obtain
    $|a-b| \leq \dobs_{\Gamma,\mathbf{R}}(M,N)$.
    (When $\tau= \sigma_{1} \otimes \sigma_{2}$)
    Let $\gamma=(V_{1},\ldots,V_{n})$ be an
    element of $\val(\Gamma)$. For
    $U_{1},V_{1} \in \val(\sigma_{1})$ and
    $U_{2},V_{2} \in \val(\sigma_{2})$ such that
    $M \gamma \hookrightarrow U_{1} \otimes U_{2}$
    and
    $N \gamma \hookrightarrow V_{1} \otimes
    V_{2}$, we show that there are
    $s,s' \in \bbRR$ such that
    $U_{1} \cong_{s}^{\sigma_{1}} W_{1}$ and
    $U_{2} \cong_{s'}^{\sigma_{2}} W_{2}$ and
    $s + s' \leq \dobs_{\Gamma,\sigma_{1} \otimes
      \sigma_{2}}(M,N)$. By the induction
    hypothesis on $\sigma_{1}$ and $\sigma_{2}$,
    we obtain
    $U_{1}
    \cong^{\sigma_{1}}_{\dobs_{\varnothing,\sigma_{1}}
      (U_{1},W_{1})} W_{1}$ and
    $U_{2}
    \cong^{\sigma_{2}}_{\dobs_{\varnothing,\sigma_{2}}
      (U_{2},W_{2})} W_{2}$. Hence, by the
    definition of $\cong$, we have
    $M\gamma
    \cong_{\dobs_{\varnothing,\sigma_{1}}(U_{1},W_{1})
      +
      \dobs_{\varnothing,\sigma_{2}}(U_{2},W_{2})}^{\sigma_{1}
      \otimes \sigma_{2}} N\gamma$. It remains to
    check that
    $\dobs_{\varnothing,\sigma_{1}}(U_{1},W_{1}) +
    \dobs_{\varnothing,\sigma_{2}}(U_{2},W_{2})
    \leq \dobs_{\Gamma,\sigma_{1} \otimes
      \sigma_{2}}(M,N)$. To see this, we show that
    for any
    $t_{1} <
    \dobs_{\varnothing,\sigma_{1}}(U_{1},W_{1})$
    and
    $t_{2} <
    \dobs_{\varnothing,\sigma_{2}}(U_{2},W_{2})$,
    we have
    $t_{1} + t_{2} \leq \dobs_{\Gamma,\sigma_{1}
      \otimes \sigma_{2}}(M,N)$. Given such
    $t_{1}$ and $t_{2}$, we can find contexts
    $C_{1}[-] \colon (\varnothing,\sigma_{1}) \to
    (\varnothing,\mathbf{R}^{\otimes n} \otimes
    \rho_{1})$ and
    $C_{2}[-] \colon (\varnothing,\sigma_{2}) \to
    (\varnothing,\mathbf{R}^{\otimes m} \otimes
    \rho_{2})$ such that
    \begin{align*}
      C_{1}[U_{1}] &\hookrightarrow
      \const{a_{1}} \otimes \cdots \otimes
      \const{a_{n}} \otimes W_{1}, \\
      C_{1}[V_{1}] &\hookrightarrow
      \const{b_{1}} \otimes \cdots \otimes
      \const{b_{n}} \otimes W'_{1}, \\
      C_{2}[U_{2}] &\hookrightarrow
      \const{c_{1}} \otimes \cdots \otimes
      \const{c_{m}} \otimes W_{2}, \\
      C_{2}[V_{2}] &\hookrightarrow
      \const{d_{1}} \otimes \cdots \otimes
      \const{d_{m}} \otimes W'_{2},
    \end{align*}
    and
    \begin{equation*}
      t_{1} \leq
      |a_{1} - b_{1}| +
      \cdots + |a_{n} - b_{n}|,
      \qquad
      t_{2} \leq
      |c_{1} - d_{1}| +
      \cdots + |c_{m} - d_{m}|.
    \end{equation*}
    We define a context $D[-]$ by
    \begin{multline*}
      D[-] =
      \letin{y \otimes z}{
        (\lambda x_{1}:\tau_{1}.\,\cdots\lambda
        x_{n}:\tau_{n}.\,[-])\,V_{1}\,\cdots\,V_{n}
      }{} \\
      \letin{v \otimes v'}{C_{1}[y]}{
        \letin{u \otimes u'}{C_{2}[z]}{
          (H\,(v \otimes u)) \otimes (v' \otimes u')
        }
      }
    \end{multline*}
    where $y$ and $z$ are fresh variables that do
    not appear in $C[-]$ nor $D[-]$, and
    $\vdash H \colon \mathbf{R}^{\otimes n}\otimes
    \mathbf{R}^{\otimes m} \to \mathbf{R}^{\otimes
      (n + m)}$ is a value that changes bracketing
    by using let-bindings. Then, we obtain
    \begin{align*}
      B[M] &\hookrightarrow \const{a_{1}} \otimes \cdots \otimes
      \const{a_{n}} \otimes \const{c_{1}} \otimes \cdots \otimes
      \const{c_{m}} \otimes (W_{1} \otimes W_{2}), \\
      B[N] &\hookrightarrow \const{b_{1}} \otimes \cdots \otimes
      \const{b_{n}} \otimes \const{d_{1}} \otimes \cdots \otimes
      \const{d_{m}} \otimes (W'_{1} \otimes W'_{2}).
    \end{align*}
    Hence,
    $t_{1} + t_{2} \leq \dobs_{\Gamma,\sigma_{1}
      \otimes \sigma_{2}}(M,N)$.
    (When $\tau=\sigma_{1} \multimap \sigma_{2}$)
    Let $\gamma=(V_{1},\ldots,V_{n})$ be an
    element of $\val(\Gamma)$. Given
    $\lambda x:\sigma_{1}.\, M'$ and
    $\lambda x:\sigma_{1}.\, N'$ in
    $\val(\sigma_{1} \multimap \sigma_{2})$ such
    that
    $M \gamma \hookrightarrow \lambda
    x:\sigma_{1}.\, M'$ and
    $N \gamma \hookrightarrow \lambda
    x:\sigma_{1}.\, N'$, we show that for any
    $U \in \val(\sigma_{1})$, we have
    $M'[U/x] \cong_{\dobs_{\Gamma,\sigma_{1}
        \multimap \sigma_{2}}(M,N)}^{\sigma_{2}}
    N'[U/x]$. By the induction hypothesis on
    $\sigma_{2}$ and the definition of $\cong$,
    we have
    \begin{equation*}
      M'[U/x] \cong^{\sigma_{2}}_{\dobs_{\varnothing,\sigma_{2}}
        ((\lambda x:\sigma_{1}.\,M')\,U,
        (\lambda x:\sigma_{1}.\,N')\,U)} N'[U/x].
    \end{equation*}
    Hence, it remains to check that
    $\dobs_{\varnothing,\sigma_{2}} ((\lambda
    x:\sigma_{1}.\,M')\,U, (\lambda
    x:\sigma_{1}.\,N')\,U) \leq
    \dobs_{\Gamma,\sigma_{1} \multimap
      \sigma_{2}}(M,N)$. To see this, we show that
    for any
    $r < \dobs_{\varnothing,\sigma_{2}} ((\lambda
    x:\sigma_{1}.\,M')\,U, (\lambda
    x:\sigma_{1}.\,N')\,U)$, we have
    $r \leq \dobs_{\Gamma,\sigma_{1} \multimap
      \sigma_{2}}(M,N)$. Since
    $r < \dobs_{\varnothing,\sigma_{2}} ((\lambda
    x:\sigma_{1}.\,M')\,U, (\lambda
    x:\sigma_{1}.\,N')\,U)$, there is a context
    \begin{equation*}
      C[-] \colon (\varnothing,\sigma_{2}) \to
      (\varnothing,\mathbf{R}^{\otimes m} \otimes
      \upsilon)
    \end{equation*}
    such that
    \begin{align*}
      C[(\lambda x:\sigma_{1}.\,M')\,U] &\hookrightarrow
      \const{a_{1}} \otimes \cdots \otimes \const{a_{m}}
      \otimes V \\
      C[(\lambda
      x:\sigma_{1}.\,N')\,U] &\hookrightarrow
      \const{b_{1}} \otimes \cdots \otimes \const{b_{m}}
      \otimes W
    \end{align*}
    and
    \begin{math}
      r \leq |a_{1} - b_{1}| + \cdots + |a_{m} - b_{m}|.
    \end{math}
    We define $D[-]$ by
    \begin{equation*}
      D[-] =
      (\lambda y:\sigma_{1} \multimap
      \sigma_{2}.\, C[y\,U])\,
      ((\lambda x_{1}:\tau_{1}.\,\cdots\lambda
      x_{n}:\tau_{n}.\,[-])\,V_{1}\,\cdots\,V_{n}).
    \end{equation*}
    Since
    \begin{equation*}
      D[M] \hookrightarrow
      \const{a_{1}} \otimes \cdots \otimes \const{a_{m}}
      \otimes V,
      \qquad
      D[N] \hookrightarrow
      \const{b_{1}} \otimes \cdots \otimes \const{b_{m}}
      \otimes W,
    \end{equation*}
    we see that
    $r \leq \dobs_{\Gamma,\sigma_{1} \multimap
      \sigma_{2}}(M,N)$.
  \end{proof}

  \begin{lemma}\label{lem:simeq-basic-lemma}
    Let
    $\Gamma =
    (x_{1}:\tau_{1},\ldots,x_{n}:\tau_{n})$ be an
    environment, and let
    $\gamma,\gamma' \in \val(\Gamma)$ be
    substitutions. If
    $\gamma \simeq_{r}^{\Gamma} \gamma'$, then
    for any term $\Gamma \vdash M:\tau$, we have
    $M\gamma \simeq_{r}^{\tau} M\gamma'$.
  \end{lemma}
  \begin{proof}
    By induction on the derivation of
    $\Gamma \vdash M: \tau$.
  \end{proof}

  \begin{lemma}\label{lem:simeq-cong-metric}
    Let $\tau$ be a type.
    \begin{enumerate}
    \item For any $\vdash M : \tau$, we have
      $M \simeq_{0}^{\tau} M$.
    \item For any $\vdash M,N,L:\tau$, if
      $M \simeq_{r}^{\tau} N$ and
      $N \simeq_{s}^{\tau} L$, then
      $M \simeq_{r+s}^{\tau} L$.
    \item For any $\vdash M:\tau$, we have
      $M \cong_{0}^{\tau} M$.
    \item For any $\vdash M,N,L:\tau$, if
      $M \cong_{r}^{\tau} N$ and
      $N \cong_{s}^{\tau} L$, then
      $M \cong_{r+s}^{\tau} L$.
    \end{enumerate}
  \end{lemma}
  \begin{proof}
    (1) follows from
    Lemma~\ref{lem:simeq-basic-lemma}. (2) By
    induction on $\tau$. For the case of
    $\tau = \tau_{1} \multimap \tau_{2}$, we use
    (1). (3 and 4) By induction on $\tau$.
  \end{proof}

  \begin{lemma}\label{lem:simeq=cong}
    For any type $\tau$ and any $r \in \bbRR$,
    $M \simeq_{r}^{\tau} N$ if and only if
    $M \cong_{r}^{\tau} N$.
  \end{lemma}
  \begin{proof}
    By induction on $\tau$. The only non-trivial
    case is $\tau=\tau_{1} \multimap \tau_{2}$. We
    first show that $M \simeq_{r}^{\tau} N$
    implies $M \cong_{r}^{\tau} N$. Let
    $\lambda x:\tau_{1}.\, M'$ and
    $\lambda x:\tau_{1}.\, N'$ be values such that
    $M \hookrightarrow \lambda x:\tau_{1}.\, M'$
    and
    $N \hookrightarrow \lambda x:\tau_{1}.\, N'$.
    We show that for any $V \in \val(\tau_{1})$,
    we have
    $M'[V/x] \cong_{r}^{\tau_{2}} N'[V/x]$. Given
    $V \in \val(\tau_{1})$, by
    Lemma~\ref{lem:simeq-cong-metric}, we have
    $M'[V/x] \simeq_{r}^{\tau_{2}} N'[V/x]$. By
    the induction hypothesis on $\tau_{2}$, we
    obtain the conclusion
    $M'[V/x] \cong_{r}^{\tau_{2}} N'[V/x]$. We
    next suppose that $M \cong_{r}^{\tau} N$ and
    $M \hookrightarrow \lambda x:\tau_{1}.\, M'$
    and
    $N \hookrightarrow \lambda x:\tau_{1}.\, N'$.
    For all $V,U \in \val(\tau_{1})$, we show that
    if $V \simeq_{s}^{\tau_{1}} U$, then
    $M'[V/x] \simeq_{r+s}^{\tau_{2}} N'[U/x]$. By
    Lemma~\ref{lem:simeq-basic-lemma}, we have
    $M'[V/x] \simeq_{s}^{\tau_{2}} M'[U/x]$. From
    the assumption $M \cong_{r}^{\tau} N$, we
    obtain $M'[U/x] \cong_{r}^{\tau_{2}} N'[U/x]$.
    Then it follows from the induction hypothesis
    on $\tau_{2}$ that
    $M'[U/x] \simeq_{r}^{\tau_{2}} N'[U/x]$.
    Finally, it follows from by
    Lemma~\ref{lem:simeq-cong-metric} that
    $M'[V/x] \simeq_{r+s}^{\tau_{2}} N'[U/x]$.
  \end{proof}

  \begin{theorem}\label{thm:ctx-ext}
    For any environment $\Gamma$ and any type
    $\tau$, we have
    $\dobs_{\Gamma,\tau}= \dext_{\Gamma,\tau}$.
  \end{theorem}
  \begin{proof}
    It follows from Lemma~\ref{alem:ext->ctx} and
    Lemma~\ref{lem:simeq=cong} that
    $\dobs_{\Gamma,\tau} \geq
    \dext_{\Gamma,\tau}$. For the other
    inequality, we show that if
    $\dext_{\Gamma,\tau}(M,N) \leq r$, then
    $\dobs_{\Gamma,\tau}(M,N) \leq r$. For
    simplicity, we suppose that
    $\Gamma = (x:\sigma)$, and we define $V$ for
    $\lambda x:\sigma.\,M$ and write $U$ for
    $\lambda x:\sigma.\,N$. When
    $\dext_{\Gamma,\tau}(M,N) \leq r$, we have
    $V \simeq_{r}^{\sigma \multimap \tau} U$.
    Then, by Lemma~\ref{lem:simeq-basic-lemma},
    for any context
    $C[-] \colon (\Gamma,\tau) \to
    (\varnothing,\mathbf{R}^{\otimes m} \otimes
    \upsilon)$, we have
    $C[V\,x] \simeq_{r}^{\mathbf{R}^{\otimes m}
      \otimes \upsilon} C[U\,x]$. From this, it is
    not difficult to derive
    $C[M] \simeq_{r}^{\mathbf{R}^{\otimes m}
      \otimes \upsilon} C[N]$. By the definition
    of $\simeq$, we obtain
    \begin{math}
      \dobs_{\Gamma,\tau}(M,N) \leq r.
    \end{math}
  \end{proof}
}

\section{Equational Metric}
\label{sec:eq}


We give another syntactic metric on $\LL{S}$,
which we call the \emph{equational metric}. This is
essentially the quantitative equational theory
from \cite{DBLP:conf/csl/DahlqvistN22} without the
rules called \textbf{weak}, \textbf{join} and
\textbf{Arch}. We exclude these rules since they
do not affect the equational metric $\deq$ given
below. \longversion{See Remark~\ref{rem:arch} for
  a proof.}

For terms $\Gamma \vdash M :\tau$
and $\Gamma \vdash N :\tau$,
and for $r \in \bbRR$,
we write
\begin{equation*}
  \Gamma \vdash M \approx_{r} N:\tau
\end{equation*}
when we can derive the judgement from the rules in
Figure~\ref{fig:qeq}. Then, for terms
$\Gamma \vdash M:\tau$ and $\Gamma \vdash N:\tau$
we define $\deq_{\Gamma,\tau}(M,N) \in \bbRR$ by
\begin{equation*}
  \deq_{\Gamma,\tau}(M,N)
  =
  \inf
  \{r \in \bbRR \mid
  \Gamma \vdash M \approx_{r} N:\tau\}.
\end{equation*}

\begin{figure}
  \centering
  \begin{equation*}
    \prftree{
      \Gamma \vdash M = N : \tau
    }{
      \Gamma \vdash M \approx_{0} N : \tau
    }
    \qquad
    \prftree{
      \Gamma \vdash M\approx_{r} N:\tau
    }{
      \Gamma \vdash N\approx_{r} M:\tau
    }
    \qquad
    \prftree{
      \Gamma \vdash M \approx_{r} N : \tau
    }{
      \Gamma \vdash N \approx_{s} L : \tau
    }{
      \Gamma \vdash M \approx_{r + s} L : \tau
    }
  \end{equation*}
  \begin{equation*}
    \prftree{
      |a - b| \leq r
    }{
      \vdash \const{a} \approx_{r} \const{b} : \mathbf{R}
    }
    \qquad
    \prftree{
      \Gamma \vdash M \approx_{r} N : \tau
    }{
      C[-] \colon (\Gamma,\tau)
      \to (\Delta,\sigma)
    }{
      \Delta \vdash C[M] \approx_{r} C[N]: \sigma
    }
  \end{equation*}  
  \caption{Derivation Rules for
    $\Gamma \vdash M \approx_{r} N :\tau$}
  \label{fig:qeq}
\end{figure}

\begin{proposition}\label{prop:deq}
  For any environment $\Gamma$ and any type
  $\tau$, the function $\deq_{\Gamma,\tau}$ is a metric
  on $\term(\Gamma,\tau)$. Furthermore,
  $\{\deq_{\Gamma,\tau}\}_{\Gamma \in \env,\tau \in \type}$ is
  admissible.
\end{proposition}
\longversion{
  \begin{proof}
    It is straightforward to check that
    $\deq_{\Gamma,\tau}$ is a metric on
    $\term(\Gamma,\tau)$. It is also straightforward
    to check (A1) and (A4). (A2) and (A3) follow
    from semantic observation (Corollary~\ref{cor:adm}).
  \end{proof}}
\begin{myexample}
  The equational metric measures differences between
  terms by comparing their subterms. For example,
  we have
  $\vdash \const{2} =_{1} \const{3} : \mathbf{R}$,
  and therefore, 
  \begin{math}
    k:\mathbf{R} \multimap \mathbf{R}
    \vdash k\,\const{2} =_{1} k\,\const{3} : \mathbf{R}
  \end{math}
  holds. From this, we see that
  $\deq_{(k:\mathbf{R} \multimap
    \mathbf{R}),\mathbf{R}} (k\,\const{2},
  k\,\const{3}) \leq 1$. In fact, this is an
  equality. This follows from
  $\dobs_{(k:\mathbf{R} \multimap
    \mathbf{R}),\mathbf{R}} (k\,\const{0},
  k\,\const{1}) \geq 1$, which is easy to check,
  and Theorem~\ref{thm:main} below.
\end{myexample}
In general, we have
$\dobs_{\Gamma,\tau}(M,N) <
\deq_{\Gamma,\tau}(M,N)$, i.e., the equational
metric is strictly more discriminating than the
observational metric (Theorem~\ref{thm:main}),
which is proved by semantically inspired metrics
in the next section.

\longversion{
  \begin{remark}\label{rem:arch}
    The following rules
    \begin{equation*}
      \prftree[r]{$(\mathbf{weak})$}{
        r \geq s
      }{
        \Gamma \vdash M \approx_{s} N:\tau
      }{
        \Gamma \vdash M \approx_{r} N:\tau
      }
      \qquad
      \prftree[r]{$(\mathbf{join})$}{
        \Gamma \vdash M \approx_{r} N:\tau
      }{
        \Gamma \vdash M \approx_{s} N:\tau
      }{
        \Gamma \vdash M \approx_{\min\{r,s\}} N:\tau
      }
    \end{equation*}
    \begin{equation*}
      \prftree[r]{$(\mathbf{Arch})$}{
        \forall r > s, \
        \Gamma \vdash M \approx_{r} N:\tau
      }{
        \Gamma \vdash M \approx_{s} N:\tau
      }
    \end{equation*}
    considered in
    \cite{DBLP:conf/csl/DahlqvistN22} is absent in
    Figure~\ref{fig:qeq} since they do not affect
    the equational metric. To see this, let us
    define
    $\Gamma \vdash M \approx_{r}^{+} N : \tau$ to
    be the family of binary relations generated by
    the rules in Figure~\ref{fig:qeq} with the
    rules \textbf{weak}, \textbf{join} and
    \textbf{Arch}. Then we have
    \begin{equation*}
      \deq_{\Gamma,\tau} (M,N) = \inf \{r \in \bbRR
      \mid \Gamma \vdash M \approx_{r}^{+} N :
      \tau\}. 
    \end{equation*}
    In fact, since
    $\Gamma \vdash M \approx_{r} N :\tau$ implies
    $\Gamma \vdash M \approx_{r}^{+} N :\tau$ for
    all $\Gamma \vdash M :\tau$ and
    $\Gamma \vdash N:\tau$, we have
    $\deq_{\Gamma,\tau}(M,N) \geq \inf \{r \in
    \bbRR \mid \Gamma \vdash M \approx_{r}^{+} N
    :\tau\}$. On the other hand, since the
    following family of binary relations
    \begin{equation*}
      \Gamma \vdash
      M \mathrel{\dot{\approx}_{r}}
      N :\tau
      \iff
      \deq_{\Gamma,\tau}(M,N) \leq r
    \end{equation*}
    satisfies the rules in Figure~\ref{fig:qeq}
    and the above three rules, if
    $\Gamma \vdash M \approx_{r}^{+} N:\tau$, then
    $\deq_{\Gamma,\tau}(M,N) \leq r$. Hence,
    $\deq_{\Gamma,\tau}(M,N) \leq \inf \{r \in
    \bbRR \mid \Gamma \vdash M \approx_{r}^{+} N
    :\tau\}$.
  \end{remark}}

\section{\texorpdfstring{Models of $\LL{S}$ and Associated Metrics}{Models and Associated Metrics}}
\label{sec:models}

Now, we move our attention to semantically derived
metrics on $\LL{S}$. We first give a notion of
models of $\LL{S}$ based on
$\mathbf{Met}$-enriched symmetric monoidal closed
categories. $\mathbf{Met}$-enriched symmetric
monoidal closed categories are studied in
\cite{DBLP:conf/csl/DahlqvistN22} as models of quantitative
equational theories for the linear lambda
calculus. Then, we give two examples of semantic
metrics on $\LL{S}$: one is based on domain
theory, and the other is based on Geometry of
Interaction.

\subsection{\texorpdfstring{$\mathbf{Met}$-enriched
    Symmetric Monoidal Closed
    Category}{Met-enriched Symmetric Monoidal
    Closed Category}}

We say that a symmetric monoidal closed category
$(\mathcal{C},I,\otimes,\multimap)$ is
\emph{$\mathbf{Met}$-enriched} when every hom-set
$\mathcal{C}(X,Y)$ has the structure of a metric
space subject to the following conditions:
\begin{itemize}
\item the composition is a morphism in
  $\mathbf{Met}$ from
  $\mathcal{C}(X,Y) \otimes \mathcal{C}(Z,X)$
  to $\mathcal{C}(Z,Y)$; and
\item the tensor is a morphism in
  $\mathbf{Met}$ from
  $\mathcal{C}(X,Y) \otimes \mathcal{C}(Z,W)$
  to $\mathcal{C}(X\otimes Z, Y \otimes W)$; and
\item the currying operation is an isomorphism in
  $\mathbf{Met}$ from $\mathcal{C}(X \otimes Y,Z)$
  to $\mathcal{C}(X,Y \multimap Z)$.
\end{itemize}
For morphisms $f,g \colon X \to Y$ in
$\mathcal{C}$, we write $d(f,g)$ for the distance
between $f$ and $g$.

\begin{definition}
  A \emph{pre-model
    $\mathcal{M}= (\mathcal{C}, I, \otimes,
    \multimap, \lfloor-\rfloor)$ of $\LL{S}$} is a
  $\mathbf{Met}$-enriched symmetric monoidal
  closed category
  $(\mathcal{C},I,\otimes,\multimap)$ equipped
  with an object
  $\lfloor \mathbf{R}\rfloor \in \mathcal{C}$ and
  families of morphisms
  $\{\lfloor a \rfloor \colon I \to \lfloor
  \mathbf{R} \rfloor\}_{a \in \mathbb{R}}$ and
  $\{\lfloor f\rfloor \colon \lfloor
  \mathbf{R}\rfloor^{\otimes \mathrm{ar}(f)} \to
  \lfloor \mathbf{R} \rfloor\}_{f \in S}$.
\end{definition}

For a pre-model
$\mathcal{M} =
(\mathcal{C},I,\otimes,\multimap,\lfloor -
\rfloor)$ of $\LL{S}$, we interpret types as
follows:
\begin{align*}
  \sem{\mathbf{R}}^{\mathcal{M}}
  &=
  \lfloor \mathbf{R} \rfloor,
  &
  \sem{\mathbf{I}}^{\mathcal{M}}
  &= I,
  &
  \sem{\tau \otimes \sigma}^{\mathcal{M}}
  &=
  \sem{\tau}^{\mathcal{M}} \otimes \sem{\sigma}^{\mathcal{M}},
  &
  \sem{\tau \multimap \sigma}^{\mathcal{M}}
  &=
  \sem{\tau}^{\mathcal{M}} \multimap \sem{\sigma}^{\mathcal{M}}.
\end{align*}
For an environment
$\Gamma=(x:\tau,\ldots,y:\sigma)$, we define
$\sem{\Gamma}^{\mathcal{M}}$ to be
$\sem{\tau}^{\mathcal{M}} \otimes \cdots \otimes
\sem{\sigma}^{\mathcal{M}}$. Then, the
interpretation
$\sem{\Gamma \vdash M:\tau}^{\mathcal{M}} \colon
\sem{\Gamma}^{\mathcal{M}} \to
\sem{\tau}^{\mathcal{M}}$ in $\mathcal{M}$ is
given in the standard manner following
\cite{Mackie1993AnIL}, and 
constants are interpreted as follows:
$\sem{\vdash \const{a}:\mathbf{R}}^{\mathcal{M}}
  = \lfloor a \rfloor$,
\begin{equation*}
  \sem{\Gamma \# \cdots \# \Delta
    \vdash \const{f}(M,\ldots,N):\mathbf{R}}^{\mathcal{M}}
  = \lfloor f \rfloor
  \circ (\sem{M}^{\mathcal{M}} \otimes \cdots \otimes
  \sem{N}^{\mathcal{M}}) \circ \theta
\end{equation*}
where
$\theta \colon \sem{\Gamma \#
  \Delta}^{\mathcal{M}} \xrightarrow{\cong}
\sem{\Gamma}^{\mathcal{M}} \otimes
\sem{\Delta}^{\mathcal{M}}$ swaps objects
following the merge $\Gamma \# \Delta$.

\begin{definition}
  We say that a pre-model
  $\mathcal{M} =
  (\mathcal{C},I,\otimes,\multimap,\lfloor -
  \rfloor)$ of $\LL{S}$ is a \emph{model} of
  $\LL{S}$ if $\mathcal{M}$ satisfies the
  following three conditions.
  \begin{itemize}
  \item (M1)
    For any $f \in S$,
    if $f(a_{1},\ldots,a_{\mathrm{ar}(f)}) = b$,
    then $\sem{\const{f}(\const{a_{1}},
      \ldots,\const{a_{n}})}^{\mathcal{M}}=
    \sem{\const{b}}^{\mathcal{M}}$.
  \item (M2) For all $a,b \in \bbR$,
    $d( \lfloor a \rfloor, \lfloor b \rfloor ) =
    |a-b|$.
  \item (M3) For all $x,y \colon I \to X$ in
    $\mathcal{C}$ and all finite sequences
    $a_{1},\ldots,a_{n},b_{1},\ldots,b_{n} \in
    \mathbb{R}$, we have
    \begin{equation*}
      d(\lfloor a_{1} \rfloor
      \otimes \cdots \otimes \lfloor a_{n} \rfloor
      \otimes x,\lfloor b_{1} \rfloor
      \otimes \cdots \otimes \lfloor b_{n} \rfloor
      \otimes y)
      \geq
      |a_{1} - b_{1}| + \cdots +
      |a_{n} - b_{n}|.
    \end{equation*}
  \end{itemize}
\end{definition}

The first condition corresponds to the reduction
rule for function symbols and is necessary to
prove soundness for models of $\LL{S}$. The
remaining conditions are for admissibility of the
metric derived from models of $\LL{S}$.

\begin{proposition}[Soundness]\label{prop:soundness}
  Let $\mathcal{M}$ be a model of $\LL{S}$. For
  any term $M \in \term(\tau)$ and any
  value $V \in \val(\tau)$, if
  $M \hookrightarrow V$, then
  $\sem{M}^{\mathcal{M}} = \sem{V}^{\mathcal{M}}$.
\end{proposition}
\longversion{
  \begin{proof}
    By induction on the derivation of
    $M \hookrightarrow V$. Except for the case
    $\const{f}(M_{1},\ldots, M_{\mathrm{ar}(f)})
    \hookrightarrow \const{b}$, we can check
    $\sem{M}^{\mathcal{M}} = \sem{V}^{\mathcal{M}}$
    by using soundness of symmetric monoidal closed
    categories with respect to the linear lambda
    calculus \cite{Mackie1993AnIL}. The case
    $\const{f}(M_{1},\ldots, M_{\mathrm{ar}(f)})
    \hookrightarrow \const{b}$ follows from (M1).
  \end{proof}}

Let
$\mathcal{M}=
(\mathcal{C},I,\otimes,\multimap,\lfloor-\rfloor)$
be a model of $\LL{S}$. For an environment
$\Gamma$ and a type $\tau$, we define
$d^{\mathcal{M}}_{\Gamma,\tau}$ to be the function
\begin{equation*}
  d(\sem{-}^{\mathcal{M}},
  \sem{-}^{\mathcal{M}}) \colon
  \term(\Gamma,\tau) \times \term(\Gamma,\tau)
  \to \bbRR.
\end{equation*}
\begin{proposition}\label{prop:adm_M}
  For any environment $\Gamma$ and any type
  $\tau$, the function $d_{\Gamma,\tau}^{\mathcal{M}}$
  is a metric on $\term(\Gamma,\tau)$.
  Furthermore, 
  $\{d_{\Gamma,\tau}^{\mathcal{M}}\}_{\Gamma \in \env,\tau \in \type}$
  is admissible.
\end{proposition}
\longversion{
  \begin{proof}
    It follows from $\mathbf{Met}$-enrichment that
    $d_{\Gamma,\tau}^{\mathcal{M}}$ is a metric
    and (A1) holds. (A2) and (A3) follow from (M2)
    and (M3). (A4) follows from soundness of
    symmetric monoidal closed categories with
    respect to the linear lambda calculus
    \cite{Mackie1993AnIL} and (M1).
  \end{proof}}

\begin{myexample}
  The symmetric monoidal closed category
  $\mathbf{Met}$ of metric spaces and
  non-expansive functions can be extended to a
  model
  $(\mathbf{Met},I,\otimes,
  \multimap,\lfloor-\rfloor)$ of $\LL{S}$ where we
  define
  $\lfloor \mathbf{R} \rfloor \in \mathbf{Met}$ to
  be $\bbR$, and for $f \in S$, we define
  $\lfloor f \rfloor \colon \bbR^{\otimes
    \mathrm{ar}(f)} \to \bbR$ to be $f$.
\end{myexample}

\subsection{Denotational Metric}
\label{sec:domain}


In this section, we recall the notion of metric
cpos introduced in \cite{10.1145/3009837.3009890}
as a denotational model of Fuzz, and we give a
model of $\LL{S}$ based on metric cpos. While we
do not need the domain-theoretic feature of metric
cpos to model $\LL{S}$, we believe that the
category of metric cpos is a good place to explore
how metrics from denotational models and metrics
from interactive semantic models are related. This
is because the domain theoretic structure of the
category of metric cpos directly gives rise to an
interactive semantic model via Int-construction as
we show in Section~\ref{sec:inter-semant-model}.

Let us recall the notion of (pointed)
metric cpos \cite{10.1145/3009837.3009890}.

\begin{definition}
  A \emph{(pointed) metric cpo} $X$ consists of a
  metric space $(|X|,d_{X})$ with a partial order
  $\leq_{X}$ on $|X|$ such that $(|X|,\leq_{X})$
  is a (pointed) cpo, and for all ascending chains
  $(x_{i})_{i \in \mathbb{N}}$ and
  $(x'_{i})_{i \in \mathbb{N}}$ on $X$, we have
  \begin{math}
    d_{X}\left(\bigvee_{i \in \mathbb{N}}x_{i},
      \bigvee_{i \in \mathbb{N}}x'_{i}\right) \leq
    \bigvee_{i \in \mathbb{N}}
    d_{X}(x_{i},x'_{i}).
  \end{math}
\end{definition}
For metric cpos $X$ and $Y$, a function
$f \colon |X| \to |Y|$ is said to be continuous
and non-expansive when $f$ is a continuous
function from $(|X|,\leq_{X})$ to $(|Y|,\leq_{Y})$
and is a non-expansive function from $(|X|,d_{X})$
to $(|Y|,d_{Y})$. Below, we simply write $X$
for the underlying set $|X|$.

Pointed metric cpos and continuous and
non-expansive functions form a category, which is
denoted by $\metcppo$. The unit object $I$
of $\metcppo$ is the unit object of
$\mathbf{Met}$ equipped with the trivial partial
order. The tensor product $X \otimes Y$ is given
by the tensor product of metric spaces with the
componentwise order. The hom-object
$X \multimap Y$ is given by the set of continuous
and non-expansive functions equipped with the
pointwise order and
\begin{equation*}
  d_{X \multimap Y}(f,g)
  = \sup_{x \in X} d_{Y}(fx,gx).
\end{equation*}
We associate $\metcppo$ with the structure
of a model of $\LL{S}$ as follows. We define
$\lfloor \mathbf{R} \rfloor$ to be
\begin{math}
  (\bbR \cup \{\bot\},d_{R},\leq_{R}) 
\end{math}
where $d_{R}$ is the extension of the metric on
$\bbR$ given by $d_{R}(a,\bot) = \infty$ for all
$a \in \bbR$, and $(\bbR \cup \{\bot\},\leq_{R})$
is the lifting of the discrete cpo $\bbR$. For
$f \in S$, we define
$\lfloor f \rfloor \colon R^{\otimes
  \mathrm{ar}(f)} \to R$ to be the function
satisfying
$\lfloor f \rfloor
(x_{1},\ldots,x_{\mathrm{ar}(f)}) = y \in \bbR$ if
and only if
$x_{1},\ldots,x_{\mathrm{ar}(f)} \in \bbR$ and
$f(x_{1},\ldots,x_{\mathrm{ar}(f)}) = y$. In the
sequel, we denote the metric on $\LL{S}$ induced
by this model by $\dden$, and we call the metric
$\dden$ the \emph{denotational metric}.

\subsection{Interactive Metric}
\label{sec:interaction-metric}


We describe another model of $\LL{S}$, which we
call the \emph{interactive semantic model}. In the
interactive semantic model, terms are interpreted
as strategies interacting with their evaluation
environments. Categorically speaking, the
construction is based on the notion of
\emph{trace operator} and on the related \emph{Int-construction}
\cite{jsv}. Below, we first explain how terms are
interpreted, and then, we formally describe the
construction of the interactive semantic model.

\subsubsection{How Terms are Interpreted, Informally}
\label{sec:how-terms-are}

We present the interpretation of terms in the 
interactive semantic model using string diagrams
without explaining their meaning precisely. We
first consider a simple term
$F = \lambda x:\mathbf{R}.\, \const{f}(x)$. Its
interpretation is given by the following diagram.
\begin{equation*}
  \begin{tikzpicture}[scale=0.4]
    \node[rect] at (0,0) (f) {$f$};%
    \draw[-<-] (f.west) -- ++ (-1,0) node[left]
    {$R$};%
    \draw[->-] (f.east) -- ++ (1,0) node[right]
    {$R$};%
  \end{tikzpicture}.
\end{equation*}
This interpretation means that given an argument
$a \in \bbR$, it returns the evaluation result of
$F(\const{a}) \hookrightarrow \const{f(a)}$ as
follows:
\begin{equation*}
  \begin{tikzpicture}[scale=0.4]
    \begin{scope}
      \fill[gray!15] (-1,-1) rectangle ++(2,2);
      \fill[gray!15] (2,-1) rectangle ++(2,2);
    \end{scope}
    \begin{scope}
      \node[rect] at (0,0) (a) {$a$}; %
      \node[rect] at (3,0) (f) {$f$}; %
      \draw[->-] (a.east) -- node[above] %
      {$R$} (f.west); %
      \draw[->-] (f.east) -- ++ (1,0) %
      node[right] {$R$}; %
    \end{scope}
    \begin{scope}[xshift=6.5cm]
      \node at (0,0) {$=$};%
    \end{scope}
    \begin{scope}[xshift=8.5cm]
      \node[rect] at (0,0) (fa)
      {$fa$}; \draw[->-] (fa.east) -- ++(1,0) %
      node[right] {$R$};%
    \end{scope}
  \end{tikzpicture}.
\end{equation*}
Here, the grey regions denote components
corresponding to the argument $\const{a}$ and the
term $F$. In this example, there is no interaction
between functions and their arguments, which
instead shows up in higher-order computation. Let
us consider
$M_{a} = \lambda k:\mathbf{R} \multimap
\mathbf{R}.\,k\,\const{a}$ for $a \in \bbR$. The
interpretation of this term is given in
Figure~\ref{sfig:a}, and the interpretation of the
application $M_{a}\,F$ is given in
Figure~\ref{sfig:b}, which can be understood as a
representation of the following interaction
process between the term $M_{a}$ and its argument
$F$: the term $M_{a}$ first sends the query
$\const{a}$ to the argument $F$, and then the
argument $F$ invokes $\const{f}(\const{a})$. The
evaluation result $\const{f(a)}$ is sent to
$M_{a}$, and $M_{a}$ outputs the value
$\const{f(a)}$.
\begin{figure}
  \centering
  \captionsetup[subfigure]{justification=centering}
  \begin{subfigure}[t]{0.3\textwidth}
    \centering
    \begin{tikzpicture}[scale=0.4]
      \draw[->-=0.15] (0,0) node[left] {$R$} %
      -- ++ (1,0) -- ++ (1,1.5) -- ++ (2,0) %
      node[right] {$R$};%
      \node[rect] at (2.5,0) (0) {$a$}; \draw[->-]
      (0.east) -- ++ (1,0) %
      node[right] {$R$};%
    \end{tikzpicture}
    \caption{The interpretation of $M_{a}$}
    \label{sfig:a}
  \end{subfigure}
  \begin{subfigure}[t]{0.6\textwidth}
    \centering
    \begin{tikzpicture}[scale=0.4]
      \begin{scope}
        \fill[gray!15] (-1,-1) rectangle ++ (2,2);
        \fill[gray!15] (2,-1) rectangle ++ (3,3);
      \end{scope}
      \begin{scope}
        \node[rect] at (0,0) (f) {$f$}; %
        \node[rect] at (4,0) (0) {$a$}; %
        \draw[->-=0.9] (f.east) -- ++ (2,0) -- ++
        (1,1.5) -- ++ (3,0) node [right] {$R$}; %
        \draw[->-] (0.east) -- ++ (1,0) arc
        (90:-90:0.75) %
        -- ++ (-7,0) arc (270:90:0.75) %
        -- node [above left] {$R$} (f.west);%
      \end{scope}
      \begin{scope}[xshift=8cm,yshift=-0.5cm]
        \node at (0,0) {$=$};
      \end{scope}
      \begin{scope}[xshift=10cm,yshift=-0.5cm]
        \node[rect] at (0,0) (f0) {$f(a)$}; %
        \draw[->-] (f0.east) -- ++(1,0) %
        node[right] {$R$};
      \end{scope}
    \end{tikzpicture}
    \caption{Interpretation of $M_{a}\,F$}
    \label{sfig:b}
  \end{subfigure}
  \\[10pt]
  \begin{subfigure}[t]{0.3\textwidth}
    \centering
    \begin{tikzpicture}[scale=0.4]
      \node[draw,rect] at (0,0) (f) {$f$};
      \node[draw,rect] at (0,1.5) (g) {$g$};
      \draw[-<-] (f.west) -- (-1.5,0) %
      node[left] {$R$}; %
      \draw[->-] (f.east) -- (1.5,0) %
      node[right] {$R$}; %
      \draw[-<-] (g.west) -- (-1.5,1.5) %
      node[left] {$R$}; %
      \draw[->-] (g.east) -- (1.5,1.5) %
      node[right] {$R$}; %
    \end{tikzpicture}
    \caption{Interpretation of $N$}
    \label{sfig:c}
  \end{subfigure}
  \begin{subfigure}[t]{0.6\textwidth}
    \centering
    \begin{tikzpicture}[scale=0.4]
      \begin{scope}
        \fill[gray!15] (-1,-1) rectangle ++
        (2,3.5); %
        \fill[gray!15] (-6.5,-1) rectangle ++
        (3.5,3); %
        \node[rect] at (0,0) (f) {$f$}; %
        \node[rect] at (0,1.5) (g) {$g$}; %
        \node[rect] at (-4,0) (0) {$a$}; %
        \draw[-<-] (f.west) -- node[above] {$R$}
        (0); %
        \draw[->-=0.4] (f.east) -- (1,0)
        arc(90:-90:0.75) -- ++ (-8,0) %
        arc(270:90:0.75) -- node[above left] {$R$}
        ++ (1,0) -- ++(1,1.5) -- (g); %
        \draw[->-=0.7] (g.east) -- ++ (1,0)
        node[right] {$R$};
      \end{scope}
      \begin{scope}[xshift=5cm]
        \node at (0,0) {$=$};
      \end{scope}
      \begin{scope}[xshift=7.5cm]
        \node[rect] at (0,0) (gf0) {$g(f(a))$};
        \draw[->-] (gf0) -- ++ (3,0) node [right]
        {$R$};
      \end{scope}
    \end{tikzpicture}
    \caption{Interpretation of $N\,M_{a}$}
    \label{sfig:d}
  \end{subfigure}
  \caption{Interpretation of Terms in the
    Interactive Semantic Model}
  \label{fig:ism}
\end{figure}
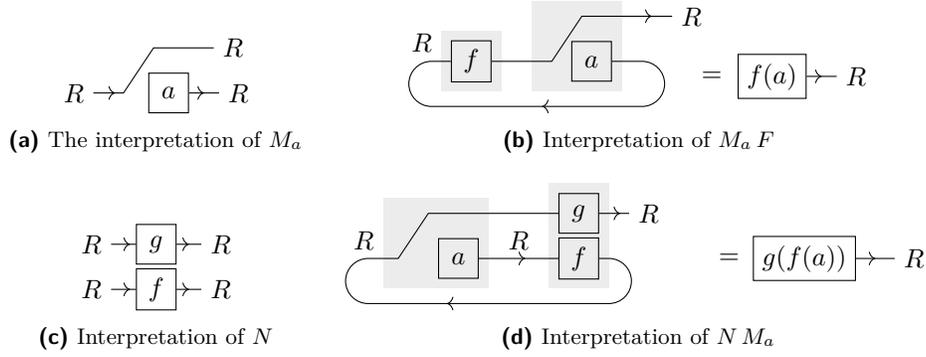
We consider another example
$N=\lambda k:(\mathbf{R} \multimap \mathbf{R})
\multimap \mathbf{R}.\, \const{g}(k\,(\lambda
x:\mathbf{R}.\,\const{f}(x)))$. Its interpretation
is given in Figure~\ref{sfig:c}, and the
interpretation of $N\,M_{a}$ is given in
Figure~\ref{sfig:d}. The interaction between $N$
and $M_{a}$ starts with the query $\const{a}$ from
$M_{a}$ to $N$. Then, $N$ invokes
$\const{f}(\const{a})$. The evaluation result
$\const{f(a)}$ is sent to $M_{a}$, and $M_{a}$
sends $\const{f(a)}$ to $N$. Finally, $N$ invokes
$\const{g}(\const{f(a)})$ and outputs the
evaluation result $\const{g(f(a))}$ of
$\const{g}(\const{f(a)})$.

In this way, in the interactive semantic model,
terms are interpreted as string diagrams that
represent ``strategies to interact with its
arguments''. The intuition of interactive metric
$\dint$ associated to the interactive semantic
model is to measure difference between these
strategies. For example, we have
\begin{math}
  \dint_{\varnothing,(\mathbf{R} \multimap
    \mathbf{R}) \multimap \mathbf{R}}
  (M_{0},M_{1}) = 1
\end{math}
since the difference between the two
interpretations
\begin{equation*}
  \begin{tikzpicture}[scale=0.4]
    \draw[->-=0.15] (0,0) node[left] {$R$} %
    -- ++ (1,0) -- ++ (1,1.5) -- ++ (2,0) %
    node[right] {$R$};%
    \node[rect] at (2.5,0) (0) {$0$};%
    \draw[->-] (0.east) -- ++ (1,0) %
    node[right] {$R$};%
  \end{tikzpicture}, \qquad
  \begin{tikzpicture}[scale=0.4]
    \draw[->-=0.15] (0,0) node[left] {$R$} %
    -- ++ (1,0) -- ++ (1,1.5) -- ++ (2,0) %
    node[right] {$R$};%
    \node[rect] at (2.5,0) (0) {$1$};%
    \draw[->-] (0.east) -- ++ (1,0) %
    node[right] {$R$};%
  \end{tikzpicture}
\end{equation*}
are the queries $0$ and $1$. We note that the
interactive semantic model provides an intentional
view, and therefore, interactive metric
distinguish some observationally equivalent terms.
For example, if $S$ has a constant function
$c \colon \mathbb{R} \to \mathbb{R}$, then for all
$a \in \mathbb{R}$, the terms
$L_{a}=\lambda k:\mathbf{R} \multimap \mathbf{R}
.\, \const{c}(k\,\const{a})$ are observationally
equivalent. On the other hand, we have
$\dint_{\varnothing,(\mathbf{R} \multimap
  \mathbf{R})\multimap \mathbf{R}}(L_{0},L_{1}) =
1$. This is because the interpretation of $L_{a}$
tells us that for any value
$V:\mathbf{R} \multimap \mathbf{R}$, the first
event in the evaluation of $L_{a}\,V$ is to invoke
$V\,\const{a}$.

\subsubsection{The Interactive Semantic Model, Formally}
\label{sec:inter-semant-model}

In order to formally describe the interactive
semantic model, we first observe that the category
$\metcppo$ has a trace operator, which is
necessary to apply the Int-construction to
$\metcppo$. For
$f \colon X \otimes Z \to Y \otimes Z$ in
$\metcppo$, we define
$\tr_{X,Y}^{Z}(f) \colon X \to Y$ by
\begin{equation*}
  \tr_{X,Y}^{Z}(f)(x)
  =
  \text{the first component of }
  f(x,z)
\end{equation*}
where $z$ is the least fixed point of the
continuous function $f(x,-) \colon Z \to Z$. When
we ignore the fragment of metric spaces, the
definition of $\tr_{X,Y}^{Z}(f)$ coincides with
the definition of the trace operator associated to
the least fixed point operator on the category of
pointed cpos and continuous function
\cite{10.5555/519939}. Hence, in order to show
that $\tr_{X,Y}^{Z}$ is a trace operator, it is
enough to check non-expansiveness of
$\tr_{X,Y}^{Z}(f)$.

\begin{proposition}\label{prop:trace}
  The symmetric monoidal category
  $(\metcppo,I,\otimes)$ equipped with the
  family of operators
  $\{\tr_{X,Y}^{Z}\}_{X,Y,Z \in \metcppo}$
  is a traced symmetric monoidal category.
\end{proposition}
\longversion{\begin{proof} We write
    $g \colon X \otimes Z \to Z$ and
    $h \colon X \otimes Z \to Y$ for the
    continuous and non-expansive functions such
    that $f(x,a) = (g(x,a),h(x,a))$. To prove
    non-expansiveness of
    $\mathrm{tr}_{X,Y}^{Z}(f)$, we suppose that
    there are $x,x' \in X$ such that
    \begin{equation*}
      d_{Y}(\mathrm{tr}_{X,Y}^{Z}(f)(x),
      \mathrm{tr}_{X,Y}^{Z}(f)(x'))
      > d_{X}(x,x')
    \end{equation*}
    and derive a contradiction. By the assumption,
    $d_{X}(x, x')$ is finite. We define
    $a_{n},a'_{n} \in Z$ by
    \begin{align*}
      a_{0} &= a'_{0} = \bot,
      &
      a_{n + 1} &= g(x,a_{n}),
      &
      a'_{n + 1} &= g(x', a'_{n}).
    \end{align*}
    We write $a_{\infty}$ for
    $\bigvee_{n \in \mathbb{N}}a_{n}$ and
    $a'_{\infty}$ for
    $\bigvee_{n \in \mathbb{N}}a'_{n}$. Below, we
    show that $d_{Z}(a_{\infty},a'_{\infty})$ is
    finite. We first check that
    $d_{Z}(a_{n},a'_{n})$ is finite. The base case
    is trivial. For the induction step $n > 0$, it
    follows from non-expansiveness of $g$ that we
    have
    \begin{equation*}
      d_{X}(x,x') + d_{Z}(a_{n},a'_{n})
      \geq d_{Z}(a_{n+1},a_{n+1}').
    \end{equation*}
    Hence, we conclude that $d_{Z}(a_{n},a'_{n})$
    is finite. We next check that
    the sequence $d_{Z}(a_{n},a'_{n})$ is bounded.
    Since we have
    \begin{equation*}
      \mathrm{tr}_{X,Y}^{Z}(f)(x)
      = h(x,a_{\infty})
      = \bigvee_{n \geq 0} h(x,a_{n}),
      \quad
      \mathrm{tr}_{X,Y}^{Z}(f)(x')
      = h(x,a'_{\infty})
      = \bigvee_{n \geq 0} h(x',a'_{n}),
    \end{equation*}
    by using Lemma 4.5 in
    \cite{10.1145/3009837.3009890},
    we obtain
    \begin{equation*}
      \liminf_{n \to \infty}
      d_{Y}(h(x,a_{n}), h(x',a'_{n}))
      \geq
      d_{Y}(\mathrm{tr}_{X,Y}^{Z}(f)(x),
      \mathrm{tr}_{X,Y}^{Z}(f)(x'))
      > d_{X}(x,x').
    \end{equation*}
    From this, we see that there exists $N \geq 0$
    such that for all $n \geq N$,
    \begin{equation*}
      d_{Y}(h(x,a_{n}), h(x',a'_{n})) > d_{X}(x,x').
    \end{equation*}
    Then, it follows from non-expansiveness of $f$
    that for all $n \geq N$, we have
    \begin{align*}
      d_{X}(x,x') + d_{Z}(a_{n}, a'_{n})
      &\geq
      d_{Y}(h(x,a_{n}), h(x',a'_{n}))
      + d_{Z}(a_{n + 1}, a'_{n + 1}) \\
      &\geq
      d_{X}(x,x')
      + d_{Z}(a_{n + 1}, a'_{n + 1}).
    \end{align*}
    Hence, since $d_{X}(x,x')$ is finite,
    we have
    \begin{equation*}
      d_{Z}(a_{n}, a'_{n})
      \geq
      d_{Z}(a_{n+1}, a'_{n+1})
    \end{equation*}
    for all $n \geq N$.
    Now, we obtain
    \begin{equation*}
      d_{Z}(
      a_{\infty},a'_{\infty}
      )
      = d_{Z}\left(\bigvee_{n \geq N} a_{n},
        \bigvee_{n \geq N} a'_{n}\right)
      \leq d_{X}(a_{N},a_{N}') < \infty.
    \end{equation*}
    Since
    \begin{equation*}
      d_{X}(x,x') + d_{Z}(a_{\infty},a'_{\infty})
      \geq
      d_{Y}
      (\mathrm{tr}_{X,Y}^{Z}(f)(x),
      \mathrm{tr}_{X,Y}^{Z}(f)(x'))
      +
      d_{Z}
      (a_{\infty},a'_{\infty}),
    \end{equation*}
    we have
    \begin{equation*}
      d_{X}(x,x') \geq d_{Y}
      (\mathrm{tr}_{X,Y}^{Z}(f)(x),
      \mathrm{tr}_{X,Y}^{Z}(f)(x')),
    \end{equation*}
    which contradicts the assumption.
  \end{proof}}

Now, we can apply the Int-construction to $\metcppo$
and obtain a symmetric monoidal closed category
$\mathbf{Int}(\metcppo)$. (In fact, what we obtain
is a compact closed category, and we only need its
symmetric monoidal closed structure to interpret
$\LL{S}$.) Objects in $\mathbf{Int}(\metcppo)$ are
pairs $X = (X_{+},X_{-})$ consisting of objects
$X_{+}$ and $X_{-}$ in $\metcppo$, and a morphism
from $X$ to $Y$ in $\mathbf{Int}(\metcppo)$ is a
morphism from $X_{+} \otimes Y_{-}$ to
$X_{-} \otimes Y_{+}$ in $\metcppo$. The identity
on $(X_{+},X_{-})$ is the symmetry
$X_{+} \otimes X_{-} \cong X_{-} \otimes X_{+}$,
and the composition of
$f \colon (X_{+},X_{-}) \to (Y_{+},Y_{-})$ is
given by
\begin{equation*}
  \tr_{X_{+} \otimes Z_{-},X_{-} \otimes Z_{+}}^{Y_{-} \otimes Y_{+}}
  \left(
    (X_{-} \otimes \theta)
    \circ
    (f \otimes g)
    \circ
    (X_{+} \otimes \theta')
  \right)
\end{equation*}
where
$\theta \colon Y_{+} \otimes Y_{-} \otimes Z_{+}
\to Z_{+} \otimes Y_{-} \otimes Y_{+}$ and
$\theta' \colon Y_{-} \otimes Y_{+} \otimes Z_{-}
\to Z_{-} \otimes Y_{-} \otimes Y_{+}$ are the
canonical isomorphisms, and we omit some coherence
isomorphisms. The symmetric monoidal closed
structure of $\mathbf{Int}(\metcppo)$ is
given as follows. The tensor unit is $(I,I)$, and
the tensor product $X \otimes Y$ is
$(X_{+} \otimes Y_{+},X_{-} \otimes Y_{-})$. The
hom-object $X \multimap Y$ is
$(X_{-} \otimes Y_{+},X_{+} \otimes Y_{-})$. For
more details on the categorical structure of
$\mathbf{Int}(\metcppo)$, see
\cite{jsv,selinger2011}.

We associate $\mathbf{Int}(\metcppo)$ with
the structure of a model of $\LL{S}$ as follows.
We define $\lfloor \mathbf{R} \rfloor$ to be
$(R,I)$, and for each $f \in S$, we define
$\lfloor f \rfloor \colon (R,I)^{\otimes
  \mathrm{ar}(f)} \to (R,I)$ by
\begin{equation*}
  R^{\otimes \mathrm{ar}(f)} \otimes I
  \xrightarrow{\cong}
  R^{\otimes \mathrm{ar}(f)}
  \xrightarrow{\text{the interpretation of $\const{f}$
      in $\metcppo$}}
  R
  \xrightarrow{\cong}
  I^{\otimes \mathrm{ar}(f)} \otimes R.
\end{equation*}
We write $\dint$ for the metric on $\LL{S}$
induced by the interactive semantic model, and we
call $\dint$ the \emph{interactive metric}.

In Figure~\ref{fig:goi}, we describe the
interpretation of $\LL{S}$ in
$\mathbf{Int}(\metcppo)$ in terms of string
diagrams. Here, we write $\psem{\tau}_{+}$ and
$\psem{\tau}_{-}$ for the positive part and the
negative part of the interpretation of $\tau$, and
we write $\psem{\Gamma \vdash M : \tau}$ for the
interpretation of a term $\Gamma \vdash M :\tau$.
See Figure~\ref{fig:smcc} (and
\cite{selinger2011}) for the meaning of string
diagrams. The interpretation
$\psem{\vdash \ast : \mathbf{I}}$ is not in
Figure~\ref{fig:goi} since
$\psem{\vdash \ast : \mathbf{I}}$ is the identity
on the unit object $I$, which is presented by zero
wires. In the interpretation of
$\const{f}(M_{1},\ldots,M_{\mathrm{ar}(f)})$, we
suppose that $\mathrm{ar}(f) = 2$ for legibility.

\begin{figure}[t]
  \centering
  \begin{tabular}{ccc}
    \rule{0pt}{12pt}
    $\mathrm{id}_{X} \colon X \to X$
    &
    $(g \colon Y \to Z) \circ (f \colon X \to Y)$
    &
    $f \colon X \otimes \cdots \otimes Z
    \to Y \otimes \cdots \otimes W$
    \\
    $\vcenter{\hbox{
        \begin{tikzpicture}[scale=0.5]
          \draw[->-] (0,0) node[left]
          {$\scriptstyle X$} -- ++ (2,0); %
        \end{tikzpicture}}}$ & $\vcenter{\hbox{
        \begin{tikzpicture}[scale=0.5]
          \draw (0,0) rectangle node {$f$} ++
          (1,1); %
          \draw (2,0) rectangle node {$g$} ++
          (1,1); %
          \draw[->-] (-1,0.5) node[left]
          {$\scriptstyle X$} -- ++ (1,0); %
          \draw[->-] (1,0.5) -- node[above]
          {$\scriptstyle Y$} ++ (1,0); %
          \draw[->-] (3,0.5) -- ++ (1,0)
          node[right] {$\scriptstyle Z$}; %
        \end{tikzpicture}}}$ &
    $\vcenter{\hbox{\begin{tikzpicture}[scale=0.5]
          \draw (0,0) rectangle node {$f$} ++
          (1,1.3); %
          \node[scale=0.7] at (-0.5, 0.8) {$\vdots$}; %
          \node[scale=0.7] at (1.5, 0.8) {$\vdots$}; %
          \draw[->-] (-1,0.2) node[left]
          {$\scriptstyle X$} -- ++ (1,0); %
          \draw[->-] (1,0.2) -- ++ (1,0)
          node[right] {$\scriptstyle Y$}; %
          \draw[->-] (-1,1.1) node[left]
          {$\scriptstyle Z$} -- ++ (1,0); %
          \draw[->-] (1,1.1) -- ++ (1,0)
          node[right] {$\scriptstyle W$}; %
        \end{tikzpicture}}}$
    \\[5pt]
    \rule{0pt}{12pt}
    $(f \colon X \to Y) \otimes (g \colon Z \to
    W)$ &
    $\mathrm{sym}_{X,Y} \colon X \otimes Y \to Y
    \otimes X$ &
    $\mathrm{tr}_{X,Y}^{Z}(f) \colon X \to Y$
    \\
    $\vcenter{\hbox{\begin{tikzpicture}[scale=0.4]
          \draw (0,0) rectangle node {$f$} ++
          (1,1); %
          \draw (0,1.25) rectangle node {$g$} ++
          (1,1); %
          \draw[->-] (-1,0.5) node[left]
          {$\scriptstyle X$} -- ++ (1,0); %
          \draw[->-] (1,0.5) -- ++ (1,0)
          node[right] {$\scriptstyle Y$}; %
          \draw[->-] (-1,1.75) node[left]
          {$\scriptstyle Z$} -- ++ (1,0); %
          \draw[->-] (1,1.75) -- ++ (1,0)
          node[right] {$\scriptstyle W$}; %
        \end{tikzpicture}}}$ &
    $\vcenter{\hbox{\begin{tikzpicture}[scale=0.4]
          \draw[->-] (0,0) node[left]
          {$\scriptstyle X$} -- ++ (1,0); %
          \draw[->-] (0,1) node[left]
          {$\scriptstyle Y$} -- ++ (1,0); %
          \draw (1,0) -- ++(1,1); %
          \draw (1,1) -- ++(1,-1); %
          \draw[->-] (2,0) -- ++ (1,0) node[right]
          {$\scriptstyle Y$}; %
          \draw[->-] (2,1) -- ++ (1,0) node[right]
          {$\scriptstyle X$}; %
        \end{tikzpicture}}}$ &
    $\vcenter{\hbox{\begin{tikzpicture}[scale=0.5]
          \draw (0,0) rectangle node
          {$\scriptstyle f$} ++ (1,1); %
          \draw[->-] (-2,0.25) node[left]
          {$\scriptstyle X$} -- ++ (2,0); %
          \draw[->-] (1,0.25) -- ++ (2,0)
          node[right] {$\scriptstyle Y$}; %
          \draw[->-] (-1,0.75) -- ++ (1,0); %
          \draw[->-] (1,0.75) -- ++ (1,0); %
          \draw (-1,0.75) arc(270:90:0.25); %
          \draw (2,0.75) arc(-90:90:0.25); %
          \draw[-<-] (-1,1.25) -- ++(3,0); %
        \end{tikzpicture}}}$
  \end{tabular}
  \caption{String Diagrams for the Traced
    Symmetric Monoidal Structure}
  \label{fig:smcc}
\end{figure}
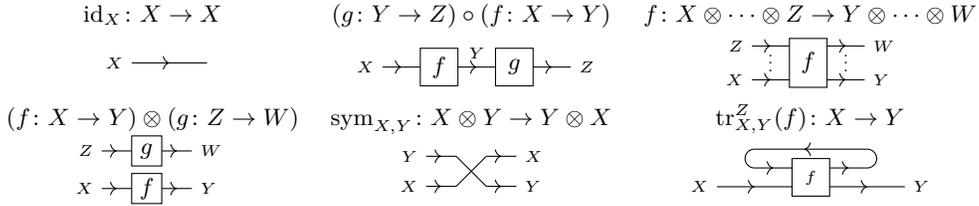

\begin{figure}[t]
  \centering
  \begin{tabular}{ccc}
    \rule{0pt}{10pt}
    $\psem{x:\tau \vdash x:\tau}$
    &
    $\psem{\Gamma \vdash \lambda
      x:\sigma.\,M : \sigma \multimap \tau}$
    &
    $\psem{\Gamma \# \Delta \vdash M\,N : \tau}$
    \\
    $\vcenter{\hbox{\begin{tikzpicture} %
          \begin{scope}[scale=0.4] %
            \draw[->-=0.2] (0,0) node [left]
            {$\scriptstyle \psem{\tau}_{+}$} %
            -- ++ (1,0) %
            -- ++ (1,1) %
            -- ++ (1,0); %
            \draw[->-=0.2] (0,1) node[left] %
            {$\scriptstyle \psem{\tau}_{-}$} %
            -- ++ (1,0) %
            -- ++ (1,-1) %
            -- ++ (1,0); %
          \end{scope} %
        \end{tikzpicture}}}$ %
    &
    $\vcenter{\hbox{\begin{tikzpicture} \draw
          (0,0) rectangle node {$M$} ++
          (0.5,1.25); %
          \draw[->-] (-1,0.125) node[left]
          {$\scriptstyle \psem{\Gamma}_{+}$} -- ++
          (1,0); %
          \draw[->-] (0.5,0.125) -- ++ (1,0)
          node[right]
          {$\scriptstyle \psem{\Gamma}_{-}$}; %
          \draw[->-] (-1,0.625) node[left]
          {$\scriptstyle \psem{\sigma}_{+}$} -- ++
          (1,0); %
          \draw[->-] (0.5,0.625) -- ++ (1,0)
          node[right]
          {$\scriptstyle \psem{\sigma}_{-}$}; %
          \draw[->-] (-1,1.125) node[left]
          {$\scriptstyle \psem{\tau}_{-}$} -- ++
          (1,0); %
          \draw[->-] (0.5,1.125) -- ++ (1,0)
          node[right]
          {$\scriptstyle \psem{\tau}_{+}$}; %
        \end{tikzpicture}}}$
    &
    $\vcenter{\hbox{\begin{tikzpicture}[scale=0.4]
          \begin{scope}[yshift=3cm]
            \draw (0,0) rectangle node {$M$} ++
            (1,1.5); %
            \draw[->-] (-1,0.25) -- ++ (1,0); %
            \draw[->-] (1,0.25) -- ++ (1,0); %
            \draw[->-] (-1,0.75) -- ++ (1,0); %
            \node at (-2.4,0.4)
            {$\scriptstyle \psem{\sigma}_{+}$};
            \draw[->-] (1,0.75) -- ++ (1,0);%
            \node at (3.4,0.4)
            {$\scriptstyle \psem{\sigma}_{-}$}; %
            \draw[->-] (-3,1.25) node[left]
            {$\scriptstyle \psem{\tau}_{-}$} -- ++
            (3,0); %
            \draw[->-] (1,1.25) -- ++ (3,0)
            node[right]
            {$\scriptstyle \psem{\tau}_{+}$}; %
            \draw (2,0.25) -- ++ (1,-0.75); %
            \draw[->-] (3,-0.5) -- ++ (1,0)
            node[right]
            {$\scriptstyle \psem{\Delta}_{-}$}; %
            \draw (-1,0.25) -- ++ (-1,-0.75); %
            \draw[-<-] (-2,-0.5) -- ++ (-1,0)
            node[left]
            {$\scriptstyle \psem{\Delta}_{+}$}; %
          \end{scope}
          \begin{scope}
            \draw (0,1) rectangle node {$N$} ++
            (1,1); %
            \draw[->-] (-3,1.25) node[left]
            {$\scriptstyle \psem{\Gamma}_{+}$} --
            ++ (3,0); %
            \draw[->-] (1,1.25) -- ++ (3,0)
            node[right]
            {$\scriptstyle \psem{\Gamma}_{-}$}; %
            \draw[->-] (-1,1.75) -- ++ (1,0);
            \node at (-2,1.8)
            {$\scriptstyle \psem{\sigma}_{-}$}; %
            \draw[->-] (1,1.75) -- ++ (1,0);
            \node at (3,1.8)
            {$\scriptstyle \psem{\sigma}_{+}$}; %
          \end{scope}
          \draw (-1,1.75) arc(270:90:0.25); %
          \draw[-<-] (-1,2.25) -- ++ (1,0); %
          \draw (0,2.25) -- ++ (1,0.5); %
          \draw[-<-] (1,2.75) -- ++ (1,0); %
          \draw (2,2.75) arc(-90:90:0.5); %
          \draw (2,1.75) arc(-90:90:0.25); %
          \draw[->-] (2,2.25) -- ++ (-1,0); %
          \draw (1,2.25) -- ++ (-1,0.5); %
          \draw[->-] (0,2.75) -- ++ (-1,0); %
          \draw (-1,2.75) arc(270:90:0.5); %
        \end{tikzpicture}}}$
    \\[25pt]
    \rule{0pt}{10pt}
    $\psem{\vdash \const{a} : \mathbf{R}}$
    &
    $\psem{\Gamma\#\Delta \vdash
      \const{f}(M,N):\mathbf{R}}$
    &
    $\psem{\Gamma \# \Delta \vdash M \otimes N :
      \tau \otimes \sigma}$ 
    \\
    $\vcenter{\hbox{
        \begin{tikzpicture}[scale=0.4]
          \draw (0,0) rectangle node {$a$}
          ++(1,1); %
          \draw[->-] (1,0.5) -- ++ (1,0)
          node[right] {$R$}; %
        \end{tikzpicture}
      }}$
    &
    $\vcenter{\hbox{\begin{tikzpicture}[scale=0.4]
          \begin{scope}[yshift=3cm]
            \draw (0,0) rectangle node {$N$} ++
            (1,1.5); %
            \draw[->-] (-1,0.25) -- ++ (1,0); %
            \draw[->-] (1,0.25) -- ++ (1,0); %
            \draw[->-] (1,1.25) -- ++ (2.5,0)
            node[above left] {$\scriptstyle R$}; %
            \draw (2,0.25) -- ++ (0.5,-0.75); %
            \draw[->-] (2.5,-0.5) -- ++ (3,0)
            node[right]
            {$\scriptstyle \psem{\Delta}_{-}$}; %
            \draw (-1,0.25) -- ++ (-0.5,-0.75); %
            \draw[-<-] (-1.5,-0.5) -- ++ (-1,0)
            node[left]
            {$\scriptstyle \psem{\Delta}_{+}$}; %
          \end{scope}
          \begin{scope}
            \draw (0,1.25) rectangle node {$M$} ++
            (1,1.5); %
            \draw[->-] (-2.5,1.5) node[left]
            {$\scriptstyle \psem{\Gamma}_{+}$} --
            ++ (2.5,0); %
            \draw[->-] (1,1.5) -- ++ (4.5,0)
            node[right]
            {$\scriptstyle \psem{\Gamma}_{-}$}; %
            \draw[->-] (1,2.5) -- ++ (1,0); %
            \draw (2,2.5) -- ++ (0.5,0.75); %
            \draw[->-] (2.5,3.25) -- ++ (1,0)
            node[above left] {$\scriptstyle R$}; %
          \end{scope}
          \begin{scope}[xshift=3.5cm,yshift=3cm]
            \draw (0,0) rectangle node {$f$} ++
            (1,1.5); %
            \draw[->-] (1,0.75) -- ++ (1,0) node
            [right] {$\scriptstyle R$}; %
          \end{scope}
        \end{tikzpicture}}}$
    &
    $\vcenter{\hbox{\begin{tikzpicture}[scale=0.4]
          \begin{scope}[yshift=3cm]
            \draw (0,0) rectangle node {$N$} ++
            (1,1.5); %
            \draw[->-] (-1,0.25) -- ++ (1,0); %
            \draw[->-] (1,0.25) -- ++ (1,0); %
            \draw[->-] (-3,1.25) node[left]
            {$\scriptstyle \psem{\sigma}_{-}$} --
            ++ (3,0); %
            \draw[->-] (1,1.25) -- ++ (3,0)
            node[right]
            {$\scriptstyle \psem{\sigma}_{+}$}; %
            \draw (2,0.25) -- ++ (1,-0.75); %
            \draw[->-] (3,-0.5) -- ++ (1,0)
            node[right]
            {$\scriptstyle \psem{\Delta}_{-}$}; %
            \draw (-1,0.25) -- ++ (-1,-0.75); %
            \draw[-<-] (-2,-0.5) -- ++ (-1,0)
            node[left]
            {$\scriptstyle \psem{\Delta}_{+}$}; %
          \end{scope}
          \begin{scope}
            \draw (0,1.25) rectangle node {$M$} ++
            (1,1.5); %
            \draw[->-] (-3,1.5) node[left]
            {$\scriptstyle \psem{\Gamma}_{+}$} --
            ++ (3,0); %
            \draw[->-] (1,1.5) -- ++ (3,0)
            node[right]
            {$\scriptstyle \psem{\Gamma}_{-}$}; %
            \draw[->-] (-1,2.5) -- ++ (1,0); %
            \draw[->-] (1,2.5) -- ++ (1,0); %
            \draw (2,2.5) -- ++ (1,0.75); %
            \draw[->-] (3,3.25) -- ++ (1,0)
            node[right]
            {$\scriptstyle \psem{\tau}_{+}$}; %
            \draw (-1,2.5) -- ++ (-1,0.75); %
            \draw[-<-] (-2,3.25) -- ++ (-1,0)
            node[left]
            {$\scriptstyle \psem{\tau}_{-}$}; %
          \end{scope}
        \end{tikzpicture}}}$
    \\[25pt]
    \rule{0pt}{10pt}
    &
    $\psem{\Gamma \# \Delta \vdash
      \letin{\ast}{M}{N} : \tau}$
    &
    $\psem{\Gamma \# \Delta \vdash \letin{x
        \otimes y}{M}{N} : \sigma}$
    \\
    & $\vcenter{\hbox{
        \begin{tikzpicture}[scale=0.4]
          \begin{scope}[yshift=3cm]
            \draw (0,0) rectangle node {$N$} ++
            (1,1.5); %
            \draw[->-] (-2,1.25) node[left]
            {$\scriptstyle \psem{\tau}_{-}$} -- ++
            (2,0); %
            \draw[->-] (1,1.25) -- ++ (2,0)
            node[right]
            {$\scriptstyle \psem{\tau}_{+}$}; %
            \draw[->-] (-2,0.25) node[left]
            {$\scriptstyle \psem{\Delta}_{-}$} --
            ++ (2,0); %
            \draw[->-] (1,0.25) -- ++ (2,0)
            node[right]
            {$\scriptstyle \psem{\Delta}_{+}$}; %
          \end{scope}
          \begin{scope}
            \draw (0,1.75) rectangle node {$M$} ++
            (1,1); %
            \draw[->-] (-2,2.25) node[left]
            {$\scriptstyle \psem{\Gamma}_{+}$} --
            ++ (2,0); %
            \draw[->-] (1,2.25) -- ++ (2,0)
            node[right]
            {$\scriptstyle \psem{\Gamma}_{-}$}; %
          \end{scope}
        \end{tikzpicture}}}$ &
    $\vcenter{\hbox{\begin{tikzpicture}[scale=0.4]
          \begin{scope}[yshift=3cm]
            \draw (0,0) rectangle node {$M$} ++
            (1,1.5); %
            \draw[->-] (-1,0.25) -- ++ (1,0); %
            \draw[->-] (1,0.25) -- ++ (1,0); %
            \draw[->-] (-1,0.75) -- ++ (1,0); %
            \node at (-3,0.4)
            {$\scriptstyle \psem{\tau_{1} \otimes
                \tau_{2}}_{+}$}; \draw[->-]
            (1,0.75) -- ++ (1,0);%
            \node at (4,0.4)
            {$\scriptstyle \psem{\tau_{1} \otimes
                \tau_{2}}_{-}$}; %
            \draw[->-] (-3,1.25) node[left]
            {$\scriptstyle \psem{\tau}_{-}$} -- ++
            (3,0); %
            \draw[->-] (1,1.25) -- ++ (3,0)
            node[right]
            {$\scriptstyle \psem{\tau}_{+}$}; %
            \draw (2,0.25) -- ++ (1,-0.75); %
            \draw[->-] (3,-0.5) -- ++ (1,0)
            node[right]
            {$\scriptstyle \psem{\Delta}_{-}$}; %
            \draw (-1,0.25) -- ++ (-1,-0.75); %
            \draw[-<-] (-2,-0.5) -- ++ (-1,0)
            node[left]
            {$\scriptstyle \psem{\Delta}_{+}$}; %
          \end{scope}
          \begin{scope}
            \draw (0,1) rectangle node {$N$} ++
            (1,1); %
            \draw[->-] (-3,1.25) node[left]
            {$\scriptstyle \psem{\Gamma}_{+}$} --
            ++ (3,0); %
            \draw[->-] (1,1.25) -- ++ (3,0)
            node[right]
            {$\scriptstyle \psem{\Gamma}_{-}$}; %
            \draw[->-] (-1,1.75) -- ++ (1,0); %
            \draw[->-] (1,1.75) -- ++ (1,0);
          \end{scope}
          \draw (-1,1.75) arc(270:90:0.25); %
          \draw[-<-] (-1,2.25) -- ++ (1,0); %
          \draw (0,2.25) -- ++ (1,0.5); %
          \draw[-<-] (1,2.75) -- ++ (1,0); %
          \draw (2,2.75) arc(-90:90:0.5); %
          \draw (2,1.75) arc(-90:90:0.25); %
          \draw[->-] (2,2.25) -- ++ (-1,0); %
          \draw (1,2.25) -- ++ (-1,0.5); %
          \draw[->-] (0,2.75) -- ++ (-1,0); %
          \draw (-1,2.75) arc(270:90:0.5); %
        \end{tikzpicture}}}$
  \end{tabular}
  \caption{The Interpretation of $\LL{S}$ in
    $\mathbf{Int}(\metcppo)$}
  \label{fig:goi}
\end{figure}

\section{Finding Your Way Around the Zoo}
\label{sec:finding-your-way}

We describe how admissible metrics on $\LL{S}$ in
this paper are related. Below, for metrics
$d =\{d_{\Gamma,\tau}\}_{\Gamma \in \env,\tau \in
  \type}$ and
$d' = \{d'_{\Gamma,\tau}\}_{\Gamma \in \env,\tau
  \in \type}$ on $\LL{S}$, we write $d \leq d'$
when for all terms $\Gamma \vdash M : \tau$ and
$\Gamma \vdash N : \tau$, we have
$d_{\Gamma,\tau}(M,N) \leq d'_{\Gamma,\tau}(M,N)$.
We write $d < d'$ when we have $d \leq d'$ and
$d \neq d'$. Our main results are about the
relationships between the various metrics on
$\LL{S}$ illustrated in Figure~\ref{fig1}.

\begin{theorem}\label{thm:main}
  The following inclusions hold.
  \begin{enumerate}
  \item For any admissible metric
    $d$ on $\LL{S}$, we have
    $\dext = \dobs \leq d \leq \deq$. \label{item:1}%
    \longversion{
    \item If a metric $d$ on $\LL{S}$ satisfies
      (A1) and $\dobs \leq d \leq \deq$, then $d$
      is admissible. \label{item:2}}
  \item $\dext = \dobs \leq \dden < \dint \leq \deq$.
    \label{item:3}
  \end{enumerate}
\end{theorem}
\longshortversion{
  \begin{proof}
    (Proof of (\ref{item:1})) We first show that
    $\dobs \leq d$. For any 
    $(n,\sigma,C[-]) \in \mathcal{K}(\Gamma,\tau)$, if
    $C[M] \hookrightarrow \const{a_{1}} \otimes
    \cdots \otimes \const{a_{n}} \otimes V$ and
    $C[N] \hookrightarrow \const{b_{1}} \otimes
    \cdots \otimes \const{b_{n}} \otimes U$, then
    \begin{equation*}
      \sum_{1 \leq i \leq n}
      |a_{i} - b_{i}|
      \overset{\text{(A3)}+\text{(A4)}}{\leq}
      d_{\varnothing,\mathbf{R}^{\otimes n}\otimes \sigma}
      (C[M],C[N])
      \overset{\text{(A1)}}{\leq}
      d_{\Gamma,\tau}(M,N).    
    \end{equation*}
    By the definition of $\dobs$, we obtain
    $\dobs \leq d$. We next show that
    $d \leq \deq$. We can inductively show that if
    $\Gamma \vdash M \approx_{r} N : \tau$, then
    $d_{\Gamma,\tau}(M,N) \leq r$. In the
    induction step for
    $\Gamma \vdash M \approx_{0} N:\tau$, we use
    (A4). In the induction step for
    $\vdash \const{a} \approx_{r} \const{b}:
    \mathbf{R}$, we use (A2). In the induction
    step for
    $\Delta \vdash C[M] \approx_{r} C[N]: \sigma$,
    we use (A1). By the definition of
    $\deq_{\Gamma,\tau}(M,N)$, we obtain
    $d_{\Gamma,\tau}(M,N) \leq
    \deq_{\Gamma,\tau}(M,N)$. (Proof of
    (\ref{item:2})) We check that $d$ satisfies
    (A2), (A3) and (A4). The condition (A2) holds
    because
    $|a-b|= \dobs_{\varnothing,\mathbf{R}}
    (\const{a},\const{b}) \leq
    d_{\varnothing,\mathbf{R}}(\const{a},\const{b})
    \leq
    \deq_{\varnothing,\mathbf{R}}(\const{a},\const{b})
    = |a-b|$. (A3) follows from
    $\dobs_{\Gamma,\tau} \leq d_{\Gamma,\tau}$.
    (A4) follows from
    $d_{\Gamma,\tau} \leq \deq_{\Gamma,\tau}$.
    (Proof of (\ref{item:3})) The inequalities
    $\dobs \leq \dden$ and $\dint \leq \deq$
    follow from \eqref{item:3}. The proof of the
    strict inequality $\dden < \dint$ is deferred
    to the next section.
  \end{proof}
}{The first claim in the main theorem states that
  the observational metric is the least admissible
  metric and the equational metric is the greatest
  admissible metric. In the proof, the conditions
  (A1), (A2), (A3) and (A4) in the definition of
  admissibility play different roles. While
  $\dobs \leq d$ follows from (A1), (A3) and (A4),
  $d \leq \deq$ follows from (A1), (A2) and (A4).
  In the long version \cite{longversion}, we also
  show the converse of this statement. Namely, if
  a metric $d$ on $\LL{S}$ satisfies (A1) and
  $\dobs \leq d \leq \deq$, then $d$ is
  admissible. This implies the notion of
  admissibility captures reasonable class of
  metrics on $\LL{S}$. The second claim in the
  main theorem is what illustrated in
  Figure~\ref{fig1}. The inequalities
  $\dobs \leq \dden$ and $\dint \leq \deq$ follow
  from the first claim; the proof of the strict
  inequality $\dden < \dint$ is deferred to the
  next section.}

Concrete metrics in-between $\dobs$ and $\deq$ are
useful to calculate $\dobs$ and $\deq$. For
example, it is not easy to \emph{directly} prove
$\deq_{(k:\mathbf{R} \multimap
  \mathbf{I}),\mathbf{I}}
(k\,\const{2},k\,\const{3}) \geq 1$ since we need
to know that \emph{whenever}
$k:\mathbf{R} \multimap \mathbf{I} \vdash
k\,\const{2} \approx_{r} k\,\const{3}:\mathbf{I}$
is derivable, we have $r \geq 1$. Let us give a
semantic proof for the inequality
$\deq_{(k:\mathbf{R} \multimap
  \mathbf{I}),\mathbf{I}}
(k\,\const{2},k\,\const{3}) \geq 1$. Here, we use
the interactive semantic model. The
interpretations of these terms in the interactive
semantic model are
\begin{equation*}
  \begin{tikzpicture}[scale=0.35]
    \draw[dotted,->-=0.15] (0,0) node[left] {$I$} %
    -- ++ (1,0) -- ++ (1,1.5) -- ++ (2,0) %
    node[right] {$I$};%
    \node[rect] at (2.5,0) (0) {$2$};%
    \draw[->-] (0.east) -- ++ (1,0) %
    node[right] {$R$};%
  \end{tikzpicture}, \qquad
  \begin{tikzpicture}[scale=0.35]
    \draw[dotted,->-=0.15] (0,0) node[left] {$I$} %
    -- ++ (1,0) -- ++ (1,1.5) -- ++ (2,0) %
    node[right] {$I$};%
    \node[rect] at (2.5,0) (0) {$3$};%
    \draw[->-] (0.east) -- ++ (1,0) %
    node[right] {$R$};%
  \end{tikzpicture}
\end{equation*}
where we can \emph{directly} see the values applied
to $k$. Hence, we obtain
$\dint_{(k:\mathbf{R} \multimap
  \mathbf{I}),\mathbf{I}}
(k\,\const{2},k\,\const{3}) = 1$. Then, the claim
follows from $\dint \leq \deq$.

\longversion{
  By applying Theorem~\ref{thm:main} to
  $\dden$, we can show admissibility of $\dobs$ and
  $\deq$.
  \begin{corollary}\label{cor:adm}
    The metrics $\dext=\dobs$ and $\deq$ on
    $\LL{S}$ are admissible.
  \end{corollary}
  \begin{proof}
    We first show admissibility of $\dobs$. As we
    mentioned in the proof of
    Proposition~\ref{prop:ext-adm}, it remains to
    check that $\dobs$ satisfies (A4). When
    $\Gamma \vdash M = N : \tau$, then by
    Theorem~\ref{thm:main}, we obtain
    $\dobs_{\Gamma,\tau}(M,N) \leq
    \dden_{\Gamma,\tau}(M,N) = 0$. Hence,
    $\dext_{\Gamma,\tau}(M,N) =
    \dobs_{\Gamma,\tau}(M,N) = 0$. We next show
    admissibility of $\deq$. As we mentioned in
    the proof of Proposition~\ref{prop:deq}, it
    remains to check that $\deq$ satisfies (A2)
    and (A3). Since for all
    $a,b \in \bbR$, we have
    \begin{equation*}
      |a-b| \leq \dden_{\varnothing,\mathbf{R}}
      (\const{a},\const{b})
      \leq \deq_{\varnothing,\mathbf{R}}(\const{a},\const{b}).
    \end{equation*}
    Hence, $\deq$ satisfies (A2). (A3) can be
    checked in the same way.
  \end{proof}

  As for semantic metrics, we have the following
  separation results.
  \begin{proposition}
    If $S = \emptyset$, then we have $\dobs < \dden$
    and $\dobs < \dint$.
  \end{proposition}
  \begin{proof}
    We only check the statement $\dobs < \dden$.
    The other strict inequality can be checked in
    the same way. Since $\dobs \leq \dden$, we
    only need to check they are different. Let
    $\Gamma$ be
    $(f : \mathbf{R}^{\otimes 2} \multimap
    \mathbf{R})$. Then, we have
    \begin{equation*}
      \dden_{\Gamma,\mathbf{R}^{\otimes 2}}(
      \const{0} \otimes (f\,\const{0}\,\const{0}),
      \const{1} \otimes (f\,\const{0}\,\const{0}))
      = 1.
    \end{equation*}
    On the other hand, as we observed in the proof
    of Proposition~\ref{aprop:two_to_one_arg},
    there is no closed term of type
    $\mathbf{R}^{\otimes 2} \multimap \mathbf{R}$.
    Hence,
    $\dobs_{\Gamma,\mathbf{R}^{\otimes 2}}(
    \const{0} \otimes (f\,\const{0}\,\const{0}),
    \const{1} \otimes (f\,\const{0}\,\const{0}))
    =
    \dext_{\Gamma,\mathbf{R}^{\otimes 2}}(
    \const{0} \otimes (f\,\const{0}\,\const{0}),
    \const{1} \otimes (f\,\const{0}\,\const{0}))
    = 0$.
  \end{proof}

  \begin{proposition}
    We have
    $\dden_{(k : \mathbf{R} \multimap
    \mathbf{I}),\mathbf{I}}(k\,\const{0},
    k\,\const{1}) = 0$ and
    $\dint_{(k : \mathbf{R} \multimap
    \mathbf{I}),\mathbf{I}}(k\,\const{0},
    k\,\const{1}) = 1$.
    In particular,
    $\dint \not\leq \dden$.
  \end{proposition}
  \begin{proof} We have
    $\dden_{(k : \mathbf{R} \multimap \mathbf{I}),
      \mathbf{I}}(k\,\const{0}, k\,\const{1}) =
    \sup_{k \colon \lfloor \mathbf{R} \rfloor \to
      I} d(k(0),k(1)) = 0$. On the other hand, for
    $a \in \bbR$, we have
    $\psem{k\,\const{a}} = a \colon I \to R$.
    Hence,
    $\dint_{(k : \mathbf{R} \multimap
      \mathbf{I}),\mathbf{I}}(k\,\const{0},
    k\,\const{1}) = 1$.
  \end{proof}}

\section{Comparing the Two Denotational
  Viewpoints}

In this section we show that, by passing from $\mathbf{MetCppo}$ to the interactive model via the Int-construction, one obtains a more discriminative metric.
In other words, our goal is to establish that $\dden< \dint$.

In this section, 
beyond the evaluation relation defined in Section 3, we will make reference to the standard $\beta$-reduction and $\beta$-equivalence relations on $\Lambda_{S}$. Indeed, the two semantics we are considering behave differently with respect to these relations: for $\beta$-equivalent terms $M,N$, while their interpretations in $\mathbf{MetCppo}$ coincide (and thus $\dden(M,N)=0$), this needs not be the case in the interactive model.

Let us start by making the interactive metric more explicit. 
Notably, in the case of $\beta$-normal terms, computing distances in $\mathbf{Int}(\mathbf{MetCppo})$ can be reduced to computing  distances in $\mathbf{MetCppo}$ as follows: a morphism from $\Gamma$ to $\sigma$ in $\mathbf{Int}(\mathbf{MetCppo})$ is a morphism in $\mathbf{MetCppo}$ from $\psem{\Gamma}_{+}\otimes \psem{\sigma}_{+}$ to $\psem{\Gamma}_{-}\otimes \psem{\sigma}_{+}$, where these two objects correspond to tensors of the form $\mathbf U\otimes \dots \otimes \mathbf U$, with $\mathbf U\in \{\mathbf I, \mathbf R\}$, 
More precisely, with any list of types $\Gamma$ one can associate two natural numbers $\Gamma^{+},\Gamma^{-}$ defined inductively by
$(\emptyset)^{+}= (\emptyset)^{-}  =0$, $(\mathbf U*\Gamma)^{+}=1+\Gamma^{+}$, $ (\mathbf U*\Gamma)^{-}=\Gamma^{-}$, 
$(\sigma\multimap \tau*\Gamma)^{+}= \sigma^{-}+ \tau^{+} + \Gamma^{+}$, $ 
(\sigma\multimap \tau* \Gamma)^{-}= \sigma^{+}+ \tau^{-}+ \Gamma^{-}$,
$(\sigma\otimes \tau*\Gamma)^{+}=\sigma^{+}+ \tau^{+}+ \Gamma^{+}$,   $
(\sigma\otimes \tau*\Gamma)^{-} = \sigma^{-}+ \tau^{-}+ \Gamma^{-}$.
Then one has the following:

%
\begin{proposition}[first-order int-terms]\label{prop:wires}
Let $M,N$ be $\beta$-normal terms such that $\Gamma\vdash M,N:\sigma$ and let $m=\Gamma^{+}+ \sigma^{-}$, $n=\Gamma^{-}+ \sigma^{+}$.
Then there exist \emph{first-order} linear terms $H^{M}_{1},\dots, H^{M}_{n}$, depending on variables $x_{1},\dots, x_{m}$, and a partition $I_{1},\dots, I_{m}$ of $\{1,\dots, m\}$ such that:
\begin{itemize}
\item $\Gamma_{j}\vdash H^{M}_{j}: \mathbf U$, for all $j=1,\dots,n$, where $\Gamma_{j}=\{ x_{l}:\mathbf U\mid  l\in I_{j}\}$, with $\mathbf U\in \{\mathbf I, \mathbf R\}$;
\item $\sem{M}^{\mathbf{Int}(\mathbf{MetCppo})}= \bigotimes_{j}\sem{H^{M}_{j}}^{\mathbf{MetCppo}}$.
\end{itemize}

\end{proposition}
\longversion{
\begin{proof}
\begin{itemize}
\item if $M=x$, then $\Gamma=\{x:\sigma\}$, so $m=\sigma^{+}+\sigma^{-}$ and $n=\sigma^{-}+\sigma^{+}$, hence the variables $\alpha_{1},\dots, \alpha_{m}$ can be split as 
$\beta_{1},\dots, \beta_{\sigma^{+}}, \gamma_{1},\dots, \gamma_{\sigma^{-}}$, and we let, for $i\leq \sigma^{-}$, $H^{M}_{i}=\gamma_{i}$, and for $i\geq \sigma^{+}$, $H^{M}_{\sigma^{-}+i}=\alpha_{i}$;

\item if $M=\star$, then $\Gamma=\emptyset$ and $n=1$, and we let $H^{M}_{1}=\star$;

\item if $M=\overline a$, then $\Gamma=\emptyset$ and $n=1$, and we let $H^{M}_{1}=\overline a$;

\item if  $M=\overline f(M_{1},\dots, M_{k})$, then 
there is a partition $J_{1},\dots, J_{k}$ of $1,\dots, m$, so that $\Gamma_{l}\vdash M_{l}:\mathbf U$, where $\Gamma_{l}$ only contains the variables $\alpha_{r}$ with $r\in J_{l}$. Moreover, we have that $m=\Gamma^{+}+\mathbf U^{-}=\Gamma^{+}= \sum_{l}(\Gamma_{l})^{+}$ and $n=\Gamma^{-}+\mathbf R^{+}=
\sum_{l}(\Gamma_{l})^{-}+1$. 
We thus define $H^{M}_{i}$ as follows:
\begin{itemize}
\item if $i=\sum_{l=1}^{m}\Gamma_{l}^{-}+j$, with $m< k$ and $j\leq \Gamma_{{m+1}}^{-}$, then $H^{M}_{i}=H^{M_{m+1}}_{j}$;

\item if $i=\sum_{l}^{j}\Gamma_{l}^{-}+1$, then $H^{M}_{i}=
\overline f(H^{M_{1}}_{\Gamma_{{1}}^{-}+1},\dots, H^{M_{k}}_{\Gamma_{{k}}^{-}+1})$.

\end{itemize}

\item if $M=\lambda x.M'$, then the $H^{M}_{i}$ are defined like the $H^{M'}_{i}$.

\item if $M=xM_{1}\dots M_{k}$, then there is a partition $J_{1},\dots, J_{k}$ of $\rho^{+}+1,\dots, m$ such that $
\Gamma= \{x: \rho\}+\sum_{l=1}\Gamma_{l}$,
with $\Gamma_{l}$ containing only the variables $\alpha_{s}$, for $s\in J_{l}$, and
 where 
$\rho=\sigma_{1}\multimap \dots \multimap \sigma_{k}\multimap \sigma$ and
$\Gamma_{l}\vdash M_{l}:\sigma_{l}$. Then $m=\rho^{+}+\sum_{l}\Gamma_{l}^{+}+\sigma^{-}=
\sum_{l}\sigma_{l}^{-}+ \sigma^{+}+\sum_{l}\Gamma_{l}^{+}+\sigma^{-}$, so the variables $\alpha_{1},\dots, \alpha_{m}$ can be identified with the variables occurring in all the terms $H^{M_{l}}_{l}$ plus new variables $\beta_{s}$ for any negative occurrence in $\sigma$ and $\gamma_{r}$ for any positive occurrence in $\sigma$; moreover, $n=\rho^{-}+\sum_{l}\Gamma_{l}^{-}+\sigma^{+}=
\sum_{l}\sigma_{l}^{+} + \sigma^{-} + \sum_{l}\Gamma_{l}^{-}+\sigma^{+}$. So we define the terms $H^{M}_{i}$ as follows:
	\begin{itemize}
	\item for $i=\sum_{l=1}^{m}\sigma_{l}^{+} + j$ for $m<k$ and $j\leq \sigma_{m+1}^{+}$, $H^{M}_{i}= H^{M_{m+1}}_{\Gamma_{m+1}^{-}+j}$;
	\item for $i= \sum_{l}\sigma_{l}^{+} + s$, for $s\leq \sigma^{-}$, $H^{M}_{i}=\beta_{s}$;
	\item for $i=\sum_{l}\sigma^{+}_{l}+\sigma^{-}+
	\sum_{l=1}^{m}\Gamma_{l}^{-} + j$, for $m<k$ and $j\leq \Gamma_{m+1}^{-}$, $H^{M}_{i}=H^{M_{m-1}}_{j}$;
	\item for $i=\sum_{l}\sigma_{l}^{+}+\sigma^{-}+\sum_{l}\Gamma_{l}^{-}+ r$, for $r\leq \sigma^{+}$, $H^{M}_{i}=\gamma_{r}$.

	\end{itemize}

\item if $M=M_{1}\otimes M_{2}$, then $\sigma=\sigma_{1}\otimes \sigma_{2}$ and $\Gamma$ splits as $\Gamma_{1}+\Gamma_{2}$, with $\Gamma_{1}\vdash M_{1}:\sigma_{1}$ and $\Gamma_{2}\vdash M_{2}:\sigma_{2}$.
Then $m=\Gamma_{1}^{+}+\Gamma_{2}^{+}+\sigma_{1}^{-}+\sigma_{2}^{-}$ and $n=\Gamma_{1}^{-}+\Gamma_{2}^{-}+\sigma_{1}^{+}+\sigma_{2}^{+}$, so we define $H^{M}_{i}$ as follows:
	\begin{itemize}
	\item if $i\leq \Gamma_{1}^{-}$, then $H^{M}_{i}= H^{M_{1}}_{i}$;
	\item if $i= \Gamma_{1}^{-}+j$, with $j\leq \Gamma_{2}^{-}$, then $H^{M}_{i}=H^{M_{2}}_{j}$;
	\item if $i=\Gamma^{-}+j$, with $j\leq \sigma_{1}^{+}$, then $H^{M}_{i}=H^{M_{1}}_{\Gamma^{-}_{1}+j}$;
		\item if $i=\Gamma^{-}+\sigma_{1}^{+}+j$, with $j\leq \sigma_{2}^{+}$, then $H^{M}_{i}=H^{M_{2}}_{\Gamma^{-}_{2}+j}$.

	\end{itemize}

\item if $M=\mathbf{let} \ \star \ \mathbf{be} \ M  \ \mathbf{in} \ N$, then the definition goes as for $(\lambda x.N)M$;

\item if $M=\mathbf{let} \ x\otimes y \ \mathbf{be} \ M  \ \mathbf{in} \ N$, then the definition goes as for
$(\lambda x.N)M$.

\end{itemize}

That $\sem{M}^{\mathbf{Int}(\mathbf{MetCppo})}= \bigotimes_{j}\sem{H^{M}_{j}}^{\mathbf{MetCppo}}$ can easily be checked by induction on $M$.
\end{proof}
}
Intuitively, the variables occurring in the left-hand of $\Gamma_{j}\vdash H^{M}_{j}:\mathbf U$ correspond to the left-hand ``wires'' of the string diagram representation of $\sem{M}^{\mathbf{Int}(\mathbf{MetCppo})}$, and the first-order 
 term $H^{M}_{j}$ describes what exits from  $i$-th right-hand ``wire'' of $\sem{M}^{\mathbf{Int}(\mathbf{MetCppo})}$. 
 
\begin{example}\label{ex:goiterms}
Let $M=\overline f(x(y\overline 0),z\overline 2)$ and $N=\overline g(x(z\overline 1),y\overline 3)$, so that $\Gamma \vdash M,N : \mathbf R$, where 
 $\Gamma= \{x:\mathbf R\multimap \mathbf R, y:\mathbf R\multimap \mathbf R,z:\mathbf R\multimap \mathbf R\}$. 
 The string diagram representations of $M$ and $N$, with the associated int-terms,
 are illustrated in Fig.~\ref{fig:examplewires}. \end{example}
 
  \begin{figure}
 \begin{center}
 \begin{tabular}{c c  c}
 $\vcenter{\hbox{\begin{tikzpicture}[scale=0.5]
 
 \node[draw, rectangle](f) at (0,0.5) {$   f$};
 
 \draw[->-] (-1.5,0.5) node[left=2mm] {$ x$} to 
 (-1.1,0.7) to (-0.4,0.7);
  \draw[->-] (-1.5,-0.5) node[left=2mm] {$ y$} to node[right=7mm] {$ y$} (1.5,-0.5);
   \draw[->-] (-1.5,-1.5) node[left=2mm] {$ z$} to (-1.1,-1.5) to (-0.4,0.3);
   
  \draw[->-] (0.4, 0.5) to node[right=2mm] {$  f(x,z)$} (1.5,0.5); 
  
  \node[draw, rectangle](zero) at (0,-1.5) {$  0$};
    \node[draw, rectangle](two) at (0,-2.5) {$  2$};

    \draw[->-] (zero) to node[right=2.5mm] {$ \overline 0$} (1.5,-1.5);  
        \draw[->-] (two) to node[right=2.5mm] {$ \overline 2$} (1.5,-2.5);
   
 \end{tikzpicture}}}$
 & \ \ \ \ \ \  & 
$\vcenter{\hbox{\begin{tikzpicture}[scale=0.5]
 
 \node[draw, rectangle](f) at (0,0.5) {$  g$};
 
 \draw[->-] (-1.5,0.5) node[left=2mm] {$ x$} to 
 (-1.1,0.7) to (-0.4,0.7);
  \draw[->-] (-1.5,-0.5) node[left=2mm] {$ y$} to (-0.37,0.3);
   \draw[->-] (-1.5,-1.5) node[left=2mm] {$ z$} to
   (-0.5,-0.5) 
   to node[right=4mm] {$ z$} (1.5,-0.5);

  \draw[->-] (0.4, 0.5) to node[right=2mm] {$ \overline g(x,y)$} (1.5,0.5); 
  
  \node[draw, rectangle](zero) at (0,-1.5) {$  3$};
    \node[draw, rectangle](two) at (0,-2.5) {$  1$};

    \draw[->-] (zero) to node[right=2mm] {$ \overline 3$} (1.5,-1.5);  
        \draw[->-] (two) to node[right=2mm] {$ \overline 1$} (1.5,-2.5);
   
 \end{tikzpicture}
 }}$\\
 \end{tabular}
 \end{center}
 
 \caption{String diagrams with int-terms for  
 $M=\overline f(x(y\overline 0),z\overline 2)$ and $N=\overline g(x(z\overline 1),y\overline 3)$.}
 \label{fig:examplewires}
 \end{figure}

 From Proposition \ref{prop:wires} we can now deduce the following:

\begin{corollary}\label{cor:goiterms}
For all $\beta$-normal terms $M,N$, 
$\dint(M,N)= \sum_{j=1}^{n}\dden(H^{M}_{j}, H^{N}_{j})$.
\end{corollary}

For instance, in the case of Example \ref{ex:goiterms}, the distance $\dint(M,N)$ coincides with the sum of the 
 distances, computed in $\mathbf{MetCppo}$, between the int-terms illustrated in Fig.~\ref{fig:examplewires}.

%
%

We can use Corollary \ref{cor:goiterms} to show that 
the equality $\dint= \dden$ cannot hold.
For instance, while $\dden_{(k:\mathbf R\multimap \mathbf I),\mathbf I}(k\overline 2, k\overline 3)=0$, 
 by computing the int-terms 
$H^{k\overline 2}_{1}(x)=H^{k\overline 3}_{1}(x)=x$, $H^{k\overline 3}_{2}=\overline 2$, $H^{k\overline 3}_{2}=\overline 3$ we deduce $\dint_{(k:\mathbf R\multimap \mathbf I),\mathbf I}(k\overline 2, k\overline 3)=0+1=1$.

It remains to prove then that $\dden\leq \dint$.  

\begin{theorem}\label{thm:metgoi}
For all $M,N$ such that $\Gamma \vdash M,N:\sigma$ holds, 
$\dden(M,N)\leq \dint(M,N)$.
\end{theorem} 
\shortversion{
\begin{proof}[Proof sketch]
It suffices to prove the claim for $M,N$ $\beta$-normal, using the fact that, if $M^{*}$ and $N^{*}$ are the $\beta$-normal forms of $M,N$, then $\dden(M,M^{*})=\dden(N,N^{*})=0$, and moreover $\dint(M^{*},N^{*})\leq \dint(M,N)$, as a consequence of the non-expansiveness of the trace operator. Recall that 
\begin{align*}
\dden(M,N)& =\sup\{\dden_{\sigma}(
\sem{M}^{\mathbf{MetCppo}}(\vec a),\sem{N}^{\mathbf{MetCppo}}(\vec a))\mid \vec a\in \sem{\Gamma}^{\mathbf{MetCppo}}\},\\ 
\dint(M,N)&=\sup\left\{\sum_{i=1}^{n}
\dden_{\mathbf R}(H^{M}_{i}[\vec r],H^{N}_{i}[\vec r])\ \Big \vert \  \vec r\in \mathbb R^{m}\right\}.
\end{align*}
For fixed $\vec a\in \sem{\Gamma}^{\mathbf{MetCppo}}$ we will construct reals $\vec r\in \mathbb R^{m}$, 
a sequence of terms $M=M_{0},\dots, M_{k}=N$, 
where $k=\Gamma^{-}+\sigma^{+}$, and a 
bijection $\rho: \{1,\dots, k\}\to \{1,\dots k\}$ such that the distance between $M_{i}[\vec a]$ and $M_{i+1}[\vec a]$ is bounded by the distance between the int-terms $H^{M}_{\rho(i+1)}[\vec r]$ and $H^{N}_{\rho(i+1)}[\vec r]$. 
In this way we can conclude by a finite number of applications of the triangular law that 
\begin{align*}
\dden_{\sigma}(M[\vec a],N[\vec a])  & \leq 
\dden_{\sigma}(M_{0}[\vec a], M_{1}[\vec a]) + \dots+
\dden_{\sigma}(M_{k-1}[\vec a], M_{k}[\vec a]) \\
  & \leq
\dden_{\mathbf R}(H^{M}_{\rho(1)}[\vec r],H^{N}_{\rho({1})}[\vec r])
+ \dots + 
\dden_{\mathbf R}(H^{M}_{\rho(k)}[\vec r],H^{N}_{\rho({k})}[\vec r])
\leq
\dint(M,N).
\end{align*}
We observe (see \cite{longversion} for more details) that:
\begin{itemize}
\item bound or free variables $x_{i}$ of $M$ are bijectively associated with subterms $\phi_{i}M$ of $M$ of the form $x_{i}\vec Q$ and with first-order variables $\alpha_{i}$;

 \item int-terms $H^{M}_{i}[\alpha_{i_{1}},\dots, \alpha_{i_{s}}]$ are bijectively associated with subterms $\psi_{i}M$ of $M$ of the form $H_{i}^{M}[\phi_{i_{1}}M,\dots, \phi_{i_{s}}M]$.  
\end{itemize}
Similar observations hold for $N$, and $\rho$ is defined so that, whenever $\psi_{i}N$ is a subterm of $\psi_{j}N$, $\phi(j)\leq \phi(i)$. 
%
%
%
%
Let us set $r_{\alpha_{j}}:= \phi_{j}M[\vec a]$.
 The desired sequence is defined by letting $M_{0}=M$ and $M_{i+1}$ be obtained from $M_{i}$ by replacing the subterm $\psi_{\rho(i)}M=H^{M}_{\rho(i)}[\phi_{i_{1}}M,\dots, \phi_{i_{s}}M]$ by 
$H^{N}_{\rho(i)}[\phi_{i_{1}}M,\dots, \phi_{i_{s}}M]$. Using the properties of $\rho$, one can check that this replacement is well-defined at each step, and that $M_{k}$ actually coincides with $N$. Moreover, at each step the passage from $M_{i}$ to $M_{i+1}$ is bounded in distance by
\begin{align*} 
  \dden(H^{M}_{\rho(i)}[\dots \phi_{j}M[\vec a]\dots],H^{N}_{\rho(i)}[\dots  \phi_{j}M[\vec a]\dots])  & 
 =  
\dint(H^{M}_{\rho(i)}[\dots r_{a_{j}}\dots],H^{N}_{\rho(i)}[\dots r_{a_{j}}\dots])\\ & \leq 
\dint(H^{M}_{\rho(i)},H^{N}_{\rho(i)}).
\end{align*}
\end{proof}
}

\begin{example}
For the terms $M$ and $N$ from Example \ref{ex:goiterms}, the procedure just sketched defines the sequence:
$M=\overline f(x(y\overline 0),z\overline 2)   \stackrel{\overline f(x,z)\mapsto \overline g(x,y)}{\to}
\overline g(x(y\overline 0),y\overline 0) \stackrel{y\mapsto z}{\to} \overline g(x(z\overline 2),y\overline 0)
\stackrel{\overline 0\mapsto \overline 3}{\to} \overline g(x(z\overline 2),y\overline 3)
\stackrel{\overline 2\mapsto\overline  1}{\to}\overline  g(x(z\overline 1),y\overline 3)= N$,
where at each step the replacement is of the form $H^{M}_{i}[\dots \varphi_{j}M\dots]\mapsto H^{N}_{i}[\dots \varphi_{j}M\dots]$.
\end{example}

While the argument above holds in the linear case, it does not seem to scale to graded exponentials, and in this last case we are not even sure if a result like Theorem \ref{thm:metgoi} may actually hold (see also the discussion in the next section).

\longversion{
The rest of this section is devoted to prove Theorem \ref{thm:metgoi}. For simplicity, we will restrict ourselves to a linear language without unit and tensor types $\mathbf I,\sigma\otimes \tau$, and without the associated term constructors. However, the argument developed below can be easily adapted to include such constructors. 
Given our restriction, we can suppose w.l.o.g.~that in Theorem \ref{thm:metgoi} the right-hand type $\sigma$ is $\mathbf R$.

Moreover, 
it suffices to prove the claim for $M,N$ $\beta$-normal, using the fact that, if $M^{*}$ and $N^{*}$ are the $\beta$-normal forms of $M,N$, then $\dden(M,M^{*})=\dden(N,N^{*})=0$, and moreover $\dint(M^{*},N^{*})\leq \dint(M,N)$, as a consequence of the non-expansiveness of the trace operator. 

Recall that 
\begin{align*}
\dden(M,N)& =\sup\{\dden_{\sigma}(
\sem{M}^{\mathbf{MetCppo}}(\vec a),\sem{N}^{\mathbf{MetCppo}}(\vec a))\mid \vec a\in \sem{\Gamma}^{\mathbf{MetCppo}}\}\\ 
\dint(M,N)&=\sup\left\{\sum_{i=1}^{n}
\dden_{\mathbf R}(H^{M}_{i}[\vec r],H^{N}_{i}[\vec r])\ \Big \vert \  \vec r\in \mathbb R^{m}\right\}
\end{align*}
For fixed $\vec a\in \sem{\Gamma}^{\mathbf{MetCppo}}$ we will construct reals $\vec r\in \mathbb R^{m}$, 
a sequence of terms $M=M_{0},\dots, M_{k}=N$, 
where $k=\Gamma^{-}+\sigma^{+}$, and a 
bijection $\rho: \{1,\dots, k\}\to \{1,\dots k\}$ such that the distance between $M_{i}[\vec a]$ and $M_{i+1}[\vec a]$ is bounded by the distance between the int-terms $H^{M}_{\rho(i+1)}[\vec r]$ and $H^{N}_{\rho(i+1)}[\vec r]$. 
In this way we can conclude by a finite number of applications of the triangular law that 
\begin{align*}
\dden_{\sigma}(M[\vec a],N[\vec a])  & \leq 
\dden_{\sigma}(M_{0}[\vec a], M_{1}[\vec a]) + \dots+
\dden_{\sigma}(M_{k-1}[\vec a], M_{k}[\vec a]) \\
  & \leq
\dden_{\mathbf R}(H^{M}_{\rho(1)}[\vec r],H^{N}_{\rho({1})}[\vec r])
+ \dots + 
\dden_{\mathbf R}(H^{M}_{\rho(k)}[\vec r],H^{N}_{\rho({k})}[\vec r])
\leq
\dint(M,N)
\end{align*}

To construct the sequence $M_{0},\dots, M_{k}$, we need a few preliminary results.

For any type $\sigma$ (or list of types $\Gamma$), we indicate as $\{\sigma^{+}\}$ (resp.~$\{\sigma\}^{-}$) the list, read from left to right, of positive (resp.~negative) atomic occurrences in $\sigma$. Observe that $\sigma^{+}$ (resp.~$\sigma^{-}$) coincides with the length of the list $\{\sigma^{+}\}$ (resp.~$\{\sigma^{-}\}$).

We will establish a few bijections, more precisely:
\begin{itemize}

\item between the elements of the list $\{\Gamma^{-}\}*\mathbf R$ and the \emph{positive subterms} of $M$ (resp.~of $N$), cf.~Def.~\ref{def:positivesubterms} below; this will allow us to associate each first-order term $H_{i}^{M}$ with a positive subterm of $M$;
\item between the elements of the list $\{\Gamma^{+}\}$ and the free and bound variables of $M$ (resp.~of $N$); this will allow us to associate each variable $x_{i}$ in $M$ with a first-order variable $\XX_{i}$ appearing in the int-terms of $M$.

\item finally, between $\{\Gamma^{-}\}*\mathbf R$ and a certain quotient over the set of variables of $M$.

\end{itemize}

\begin{notation}
In the following we use $F(\XX_{1},\dots, \XX_{n})$ and $G(\XX_{1},\dots, \XX_{n})$ to indicate linear first-order terms with free variables included in $\XX_{1},\dots, \XX_{n}$. Moreover, given terms $M_{1},\dots, M_{n}$ of type $R$, we indicate with $F(M_{1},\dots, M_{n})$ the (non necessarily first-order) term obtained by substituting $\XX_{1},\dots, \XX_{n}$ with $M_{1},\dots, M_{n}$ in $F$.
%
\end{notation}

\begin{definition}\label{def:positivesubterms}
A subterm of $M$ of the form $N=F(N_{1},\dots, N_{k})$ is called a \emph{positive subterm of $M$} if 
for no other first-order function $\overline g$, $N$ occurs in $M$ in a term of the form $\overline g(P_{1},\dots, P_{r-1}, N,P_{r+1},\dots, P_{k})$.
We let $\mathsf{PS}(M)$ indicate the set of positive subterms of $M$.
Moreover, we let $\mathsf{PS}_{0}(M)\subseteq \mathsf{PS}(M)$ indicate the set of positive subterms of $M$ containing no free or bound variable. 
\end{definition}

\begin{notation}
In the following, when indicating positive subterms as 
$F(M_{1},\dots, M_{n})$ we make w.l.o.g.~the assumption that the terms $M_{1},\dots, M_{n}$ do not start with a function symbol, i.e.~are of the form $x Q_{1}\dots Q_{s}$.
\end{notation}

\begin{lemma}
There exists a bijection $\iota_{M}: \{\Gamma^{-}\}*\mathbf R \to \mathsf{PS}(M)$ (and similarly for $N$).
\end{lemma}
\begin{proof}
By induction on $M$:
\begin{itemize}
\item if $M=F(x_{1},\dots, x_{n})$ is a first-order term, then $\{\Gamma^{+}\}=\mathbf R*\dots *\mathbf  R$ so $\{\Gamma^{-}\}*\mathbf R=\{\mathbf R\}$, and the bijection is $\iota_{M}(1)= t$;
\item if $M= F(M_{1},\dots, M_{n})$, where 
$M_{i}=x_{i}M_{i1}\dots M_{iq_{i}}$, then let $M_{ij}= \lambda z_{1}.\dots.\lambda z_{n_{ij}}.t'_{ij}$, where for some context $\Delta_{ij}=\{z_{1}:\sigma_{ij1},\dots, z_{n_{ij}}:\sigma_{ijr_{ij}}\}$, $\Gamma_{ij}, \Delta_{ij} \vdash t'_{ij}:R$, with $\Gamma_{ij}$ being a partition of $\Gamma-\{x_{1}:\sigma_{1},\dots, x_{n}:\sigma_{n}\}$, with $\sigma_{i}= \sigma_{i1}\multimap \dots \multimap \sigma_{iq_{i}}\multimap \mathbf R$, with $\sigma_{ij}= \sigma_{ij1}\multimap \dots \multimap \sigma_{ijn_{ij}}\multimap R$; then by the I.H.~there exist bijections $\iota_{M_{ij}}$ between $\{(\Gamma_{ij},\Delta_{ij})^{-}\}$ and $\mathsf{PS}(M_{ij})$. Notice that $\mathsf{PS}(M)=\{t\}\cup \bigcup_{i,j}\mathsf{PS}(M_{ij})$. 

Now, observe that an element of $\{\Gamma^{-}\}*\mathbf R$ is either (1) the last element $\mathbf R$, (2) an element of $\{\Gamma^{-}_{ij}\}$, (3) the last element of some $\{\sigma_{ij}^{+}\}$, or (4) an element of some $\{\sigma_{ijm}^{-}\}$.
We thus obtain then a bijection $\iota_{M}:\{ \Gamma^{-}\}*\mathbf R \to \mathrm{PS}(M)$ by letting:
\begin{itemize}
\item if $l$ is the last element of $\{\Gamma^{-}\}*\mathbf R$, then 
$\iota_{M}(l)=M$;

\item if $l$ is in $\{\Gamma^{-}_{ij}\}$, $\iota_{M}(l)=\iota_{M_{ij}}(l)$;
\item if $l$ is the last element of $\{\sigma_{ij}^{+}\}$, then $\iota_{M}(l)=M'_{ij}$;

\item if $l$ is in $\{\sigma_{ijm}^{-}\}$, $\iota_{M}(l)=\iota_{M_{ij}}(l^{*})$, where $l^{*}$ is the corresponding element in $\sigma_{i}$.

\end{itemize}
\end{itemize}
\end{proof}

\begin{remark}
The lemma above actually defines a bijection between the positive subterms of $t=F(N_{1},\dots, N_{k})$ and the terms $H_{i}^{M}$ (which, as described in more detail below, are indeed of the form $F(\XX_{1},\dots, \XX_{k})$).
\end{remark}

Let $\sigma$ be a type; for any occurrence $l\in\{ \sigma^{+}\}$, let $\sigma_{l}$ indicate the unique sub-type of $\sigma$ having $l$ as its rightmost occurrence.
Intuitively, $\iota_{M}(l)$ is the positive subterm that receives type $\sigma_{l}$ in the typing of $M$.


Let $\Gamma=\{x_{1}:\sigma_{1},\dots, x_{n}:\sigma_{n}\}$ and let $\mathrm{V}(M)$ be the set of free and bound variables of $t$.
\begin{lemma}\label{lemma:delta}
There exists a bijection $\delta_{M}:\{\Gamma^{+}\}\to \mathrm{V}(M)$ (and similarly for $\mathrm{V}(N)$).
\end{lemma}
\begin{proof}
By induction on $M$:
\begin{itemize}
\item if $M=F(x_{1},\dots, x_{n})$, then $\Gamma^{+}=\underbrace{\mathbf R*\dots *\mathbf R}_{n\text{ times}}$, and we let $\delta_{M}(i)=x_{i}$;

\item if $M= F(M_{1},\dots, M_{n})$, where 
$M_{i}=x_{i}M_{i1}\dots M_{iq_{i}}$, then let $M_{ij}= \lambda z_{1}.\dots.\lambda z_{n_{ij}}.t'_{ij}$, where for some context $\Delta_{ij}=\{z_{1}:\sigma_{ij1},\dots, z_{n_{ij}}:\sigma_{ijr_{ij}}\}$, $\Gamma_{ij}, \Delta_{ij} \vdash t'_{ij}:\mathbf R$, with $\Gamma_{ij}$ being a partition of $\Gamma-\{x_{1}:\sigma_{1},\dots, x_{n}:\sigma_{n}\}$, with $\sigma_{i}= \sigma_{i1}\multimap \dots \multimap \sigma_{iq_{i}}\multimap \mathbf R$, with $\sigma_{ij}= \sigma_{ij1}\multimap \dots \multimap \sigma_{ijn_{ij}}\multimap \mathbf R$; then by the I.H.~there exist bijections $\delta_{M_{ij}}$ between $(\Gamma_{ij}*\Delta_{ij})^{+}$ and $\mathrm{V}(M_{ij})$. Notice that $\mathrm{V}(M)=\{x_{1},\dots, x_{n}\}\cup \bigcup_{i,j}\mathrm{V}(M_{ij})$. 

Now, observe that an element of $\Gamma^{+}$ is either (1) an element of $\Gamma^{+}_{ij}$, (2) the last element of some $\sigma_{i}^{+}$, or (3) an element of some $\sigma_{ijm}^{+}$.
We obtain then a bijection $\iota_{M}: \Gamma^{+} \to \mathrm{V}(M)$ by letting:
\begin{itemize}

\item if $l$ is in $\Gamma^{+}_{ij}$, $\delta_{M}(l)=\delta_{M_{ij}}(l)$;
\item if $l$ is the last element of $\sigma_{i}^{+}$, then $\delta_{M}(l)=x_{i}$;

\item if $l$ is in $\sigma_{ijm}^{+}$, $\delta_{M}(l)=\delta_{M_{ij}}(l^{*})$, where $l^{*}$ is the corresponding element in $\sigma_{i}$.

\end{itemize}
\end{itemize}
\end{proof}


\begin{remark}
The lemma above actually defines a bijection between the variables of $M$ and the first-order variables appearing in the int-terms of $M$ (which are indeed enumerated by the list $\{\Gamma^{+}\}$).

\end{remark}

\begin{notation}
Using the lemma above $\mathrm{V}(M)\simeq \mathrm{V}(N)\simeq \Gamma^{+}$. 
We use $X$ to indicate any of these equivalent sets. As
modulo renaming we can suppose $\delta_{M}=\delta_{N}$, from now on, for all $i\in X$, we indicate as $x_{i}$ the variable $\delta_{M}(i)=\delta_{N}(i)$.
\end{notation}
 
\begin{definition}For any $i\in X$, let $\phi_{i}M$ be the unique subterm of $M$ of the form $x_{i}M_{1}\dots M_{k}$, and 
$\psi_{i}M$ the unique positive subterm of $M$ of the form 
$F(N_{1},\dots, N_{m-1}, \phi_{i}M, N_{m+1},\dots, N_{n})$. 
For all $i,j\in X$, let $i \sqsubset_{M} j$ if 
$\phi_{i}M$ is a subterm of $\phi_{j}M$, and 
$i\sim_{M}j$ if $\psi_{i}M=\psi_{j}M$.

\end{definition}

%
%

\begin{lemma}
$\sqsubseteq_{M}$ is a well-founded order on $X$. \\
 $\sim_{M}$ is an equivalence relation on $X$.
 \end{lemma}

%
%
%

The relation $\sqsubseteq_{M}$ extends naturally to $\mathsf{PS}_{0}(M)$, by letting 
$F \sqsubset_{M} i$, for $F\in\mathsf{PS}_{0}(M)$, if 
$F$ is a subterm of $\phi_{i}M$.

Let $X_{\sim_{M}}$ indicate the quotient of $X$ under $\sim_{M}$ and let $X^{M}:=X_{\sim_{M}}\cup \mathsf{PS}_{0}(M)$.
 In the following we will use $\xi,\chi,\dots$ to indicate elements of  $X^{M}$.
 
Let us extend the relation $\sqsubseteq_{M}$ from $X$ to $X^{M}$:

\begin{definition}
For all $\xi,\chi \in X^{M}$,
$\xi\sqsubseteq_{M}^{*} \chi$ holds if either $ 
\xi =\chi $ or 
$\exists j\in \chi \ \forall   i\in \xi\ i\sqsubset_{M} j
$.

Moreover, we write $\xi\sqsubseteq_{M}^{0}\chi$ if $\xi\neq \chi$, $\xi\sqsubseteq_{M}\chi$ and for all $\theta\in X^{M}$,
$\xi\sqsubseteq_{M}\theta$ and 
$\theta\sqsubseteq_{M}\chi$ implies $\theta\in \{\xi,\chi\}$. 
\end{definition}

Observe that  $\xi\sqsubseteq_{M}^{0}\chi$ holds precisely when there is $j\in \chi$ such that for all $i\in \xi$, $\phi_{j}M= x_{j}Q_{1}\dots Q_{l-1}(\psi_{i}M) Q_{l+1}\dots Q_{r}$. Moreover the following is easily proved:
\begin{lemma}
$\sqsubseteq_{M}^{*}$ is a well-founded preorder with a maximum element $\top_{M}=\{{i_{1}},\dots, {i_{r}}\}$, where $M=F ( \phi_{i_{1}}M, \dots \phi_{i_{r}}M)$.
\end{lemma}

We can define a bijection $\theta_{M}: \mathsf{PS}(M) \to X^{M}$ sending each positive subterm $P=F(N_{1},\dots, N_{n})$, where $N_{i}=x_{i}Q_{i1}\dots Q_{in_{i}}$, onto the equivalence class $\{x_{1},\dots, x_{n}\}$, if $n>0$, and onto the singleton $\{P\}$ otherwise (i.e.~if $P\in \mathsf{PS}_{0}$). 
The existence of bijections $\theta_{M}: \mathsf{PS}(M) \to X^{M}$, $\theta_{N}: \mathsf{PS}(N) \to X^{N}$ together with $\iota_{M}: \{\Gamma^{-}\}*\mathbf R$ and $\iota_{M}: \{\Gamma^{-}\}*\mathbf R$  implies that the two sets $X^{M}= X_{\sim_{M}}\cup \mathsf{PS}_{0}(M)$
and $X^{N}= X_{\sim_{N}}\cup \mathsf{PS}_{0}(N)$ are also in bijection (being both in bijection with $\{\Gamma^{-}\}*\mathbf R$).





Let $L= (\theta_{N}\circ \delta_{N})\circ (\theta_{M}\circ \delta_{M})^{-1}: X^{M} \to X^{N}$ be the bijection associating each element of $X^{M}$ with the unique element of $X^{N}$ corresponding to the same occurrence of $R$ in $\{\Gamma^{-}\}*\mathbf R$. Observe that the rightmost element of the list $\{\Gamma^{-}\}*\mathbf R$ is associated with $\top_{M}=M$ and $\top_{N}=N$, we deduce that $L(\top_{M})=\top_{N}$.

%
%
%
%

As a consequence of the bijections established above, we can enumerate the int-terms of $M$ and $N$ using, as index set, $X^{M}$ rather than $\{\Gamma^{-}\}*\mathbf R$. In particular, for all $\xi\in X^{M}$, we indicate the associated positive subterm of $t$ as
\begin{align*}
M_{\xi}& := F_{\xi}\big( \phi_{j}M\big)_{j\in \xi} \\
N_{\xi}& := G_{\xi}\big( \phi_{j}N\big)_{j\in L(\xi)}
\end{align*}
and the associated int-terms as
\begin{align*}
H^{M}_{\xi}& := F_{\xi}(\XX_{j})_{j\in \xi}\\
H^{N}_{\xi}& := G_{\xi}(\XX_{j})_{j\in L(\xi)}
\end{align*}

%
%
%

We now introduce a special class of terms:

\begin{definition}[terms with brackets]
The set of \emph{$\lambda$-terms with brackets} is defined by enriching the syntax of $\lambda$-terms with a new clause:
 if $t$ is a term, then $[{t}]$ is a term.
For any term with brackets $M$, we let $M^{\downarrow}$ indicate the term obtained by erasing all brackets from $M$.
%
%
%
%

\end{definition}


For all $i\in X$, let $\{{\phi_{i}M}\}= x_{i}(\lambda \vec z_{1}. [{M_{1}}])\dots (\lambda \vec z_{n}. [{M_{n}}])$, where $\phi_{i}M=x_{i}(\lambda \vec z.M_{1})\dots (\lambda \vec z.M_{n})$.

\begin{definition}
A set $U\subseteq X^{M}$ is \emph{upward closed} if $\alpha \in U$ and $\alpha\sqsubseteq_{M}\alpha'$ implies $\alpha'\in U$. 
For all upward closed sets $U$, the \emph{frontier of $U$}, noted $\partial U$, is the set
of $\xi\in X^{M}- U$ such that, for some $\chi \in U$ 
$\xi\sqsubseteq_{M}^{0}\chi$. We conventionally let $\partial\emptyset= \{\top_{M}\}$.
\end{definition}

For any upward closed $U\subseteq X^{M}$, we define a set $U(M)$ of $\lambda$-terms with brackets by  induction on $M$ as follows:
\begin{itemize}
\item if $U=\emptyset$ and $M=\lambda \vec z.M'$, then $U(M)=\{\lambda \vec z. [ Q]\mid Q \text{ is a $\lambda$-term}\}$;
\item if $U\neq \emptyset$ (which implies $\top_{M}\in U$) and $M=\lambda \vec z. F\big(\phi_{i}M\big)_{i\in \top_{M}}$, then 
\begin{align*}
U(M)=\left \{ \lambda \vec z.F\big( x_{i} {Q^{i}_{1}},\dots, {Q^{i}_{r_{i}}}\big)_{i\in\top_{M}}\mid Q^{i}_{j}\in U_{j}^{i}(M^{i}_{j})\right\}
\end{align*}
where for all $i\in\top_{M}$, $\phi_{i}M=x_{i} P^{i}_{1}\dots P^{i}_{r_{i}}$ and $U_{j}^{i}= U\cap X^{P_{j}^{i}}$ is an upward closed set of $X^{M_{j}^{i}}$.

\end{itemize}
Intuitively, $P\in U(M)$ if it is a term which is defined like $M$ at all positions corresponding to elements of $U$, 
and has a term in brackets at all positions of $\partial U$.

The following facts are easily established by induction on $M$:
\begin{lemma}\label{lemma:xt}
\begin{itemize}
\item[i.] $P\in X^{M}(M)$ iff $P=M$.
\item[ii.] if $P\in U(M)$ and $P$ is bracket-free, then $U=X^{M}$ and $P=M$.
\end{itemize}
\end{lemma}

We now have all ingredients to define, by induction, a sequence of terms $S_{0},\dots, S_{k}$ and a 
sequence of upward closed sets $U_{0}\subseteq \dots\subseteq U_{k}\subseteq X^{N}$, 
verifying at each step $i$ that:
\begin{itemize}
\item[a.] $S_{i}\in U_{i}(N)$;
\item[b.] for all $\xi \in \partial U_{i}$, $S_{i}$ contains the subterm  $[
M_{L^{-1}(\xi)}]$ at position $\xi$.
%

\end{itemize}

Let $S_{0}:= [M]= [{M_{\alpha_{\top}} }]$,  and $U_{0}=\emptyset$. Then $S_{0}\in U_{0}(M)$ certainly holds.
Moreover, 
 $ \partial U_{0}=\{\top_{N}\}$, and $S_{0}$ contains at its root the subterm 
$ [{M_{L^{-1}(\top_{N})}}]= [{M_{\top_{M}}}]=[ M]$.

Now, to define $S_{i+1}$, choose $\xi \in \partial U_{i}$, let $\chi=L^{-1}(\xi)$ and  
\begin{align*}
S_{i+1}:=S_{i}\Big( [{F_{\chi}\big( \phi_{i}M\big)_{i\in \chi}} ]\mapsto
G_{\xi}\big(\{{\phi_{j}M}\}\big)_{j\in\xi}\Big)
\end{align*}
and finally let $U_{i+1}= U_{i}\cup\{\xi\}$.
To check that $S_{i+1}$ is well-defined let us observe that:
\begin{itemize}
\item by the induction hypothesis $S_{i}$ contains the subterm 
$[{M_{\chi}}]=[{F_{\chi}\big( \phi_{i}M\big)_{i\in \chi}}]$ at position $\xi$; 
\item if one of the newly introduced variables $j\in \xi$ is bound in $S_{i}$, it is never introduced outside the scope of its abstraction $\lambda x_{j}$. Indeed, by the induction hypothesis, $S_{i}$ coincides with $N$ at all positions $\theta\in U_{i}$; hence, since $N$ has at position $\xi$ the subterm $G_{\xi}(\phi_{j}N)_{j\in \xi}$, it follows that any of the variables $x_{j}$ is in the scope of an abstraction $\lambda x_{j}$ in $S_{i+1}$ iff it is in the scope of the same abstraction in $N$.

%
\end{itemize}

Now, from $S_{i}\in U_{i}(N)$, that $S_{i+1}\in U_{i+1}(N)$ follows from the definition of $U_{i+1}(N)$.
Moreover, if $\xi' \in \partial U_{i+1}$, then either $\xi' \in \partial U_{i}$, so by IH we deduce that 
$S_{i}$ contains $[{M}_{L^{-1}(\xi')}]$, and since $\xi'\neq \xi$, also $S_{i+1}$ does;
or $\xi' \sqsubset_{M}^{0}\xi$, i.e.~there is $j\in\xi$ such that $\xi' \sqsubset_{M}^{0} j$; now, position $\xi'$ in $S_{i+1}$ contains one of the terms $ Q_{i}$, for $i=1,\dots, s$, where $[{\phi_{j}M}]=x_{j} {W_{1}}\dots {W_{s}}$ and $W_{i}= \lambda \vec z.Q_{i}$, and indeed we have that $Q_{i}=[{M_{L^{-1}(\xi')}}]$. 

%
%
%
 
Let us define now $M_{1},\dots, M_{k}$ as $M_{i}:=S_{i}^{\downarrow}$. It is clear then that $M_{1}=M$.
Let us check that $N_{k}=N$: from $\sharp U_{0}=0$ and $\sharp U_{i+1}=\sharp U_{i}+1$, we deduce $\sharp U_{k}=k$ and thus $\sharp \partial U_{k}=0$, which implies $U_{k}=X^{N}$. 
From $S_{k}\in U_{k}(N)=X^{N}(N)$, using Lemma \ref{lemma:xt} we deduce then that $M_{k}=S_{k}^{\downarrow}=S_{k}=N$.
%
%
%
%
%
%
%

We can now establish the main result:
\begin{proposition}
For all $\vec a \in \sem{\Gamma}$ there exist $\vec r\in \sem{\Gamma}_{+}$ such that 
$$
| \sem{M}(\vec \varphi)- \sem{N}(\vec a)| \ \leq \ \sum_{i \in
\{\Gamma^-\}*R} \sem{H^{M}_{i}}(\vec r)-\sem{H^{N}_{i}}(\vec r)|.
$$

\end{proposition}
\begin{proof}
Using the bijection $\delta:\Gamma^{+}\to  V(M)$, let $r_{l}:=\sem{\phi_{\delta(l)}M}(\vec \varphi)$. 

Recall that the terms $H^{M}_{i}$ and $H^{N}_{i}$, for $i\in \{\Gamma^{-}\}*\mathbf R$, can be equivalently enumerated
as $H^{M}_{L^{-1}(\xi)}$, $H^{N}_{\xi}$, for $\xi\in X^{N}$.
Let $\xi_{0},\dots, \xi_{k}$ be the sequence in $X^{N}$ chosen in the construction of the sequence $S_{0},\dots, S_{k}$, and let 
$\chi_{0}=L^{-1}(\xi_{0}),\dots, \chi_{k}=L^{-1}(\xi_{k})$.
We will show that for all $i=1,\dots,k$, $| \sem{M_{i-1}}(\varphi)-\sem{M_{i}}(\varphi)| \leq | \sem{H^{M}_{\chi_{i}}}(\vec r)-\sem{H^{N}_{\xi_{i}}}(\vec r)|$. Indeed we have
\begin{align*}
\Big \vert \sem{M_{i-1}}(\vec \varphi) -\sem{M_{i}}(\vec a)\Big \vert & =
\Big\vert \sem{M_{i-1}}(\vec \varphi) - 
\sem{M_{i-1}\Big( F_{\chi_{i}}\big( r_{l}\big)_{l\in \chi_{i}}\mapsto
G_{\xi_{i}}\big({r_{m}}\big)_{m\in \xi_{i}}\Big)
}(\vec a) \Big\vert\\
&  = 
\Big\vert \sem{M_{i-1}}(\vec \varphi) - 
\sem{M_{i-1}\Big( H^{M}_{\chi_{i}}(\vec r)\mapsto
H^{N}_{\xi_{i}}(\vec r)\Big)}(\vec a)
\Big\vert\\
& \leq 
\big \vert \sem{H^{M}_{\chi_{i}}}(\vec r) - \sem{H^{N}_{\xi_{i}}}(\vec r)\big \vert
\end{align*}
Using the fact that $M_{0}=M$ and $M_{k}=N$, as well as the triangular law, we deduce then
\begin{align*}
|\sem{M}(\vec a)-\sem{N}(\vec a)| & \leq 
|\sem{M_{0}}(\vec a)-\sem{M_{1}}(\vec a)| + \dots + |\sem{M_{k-1}}(\vec a)- \sem{M_{k}}(\vec a)| \\
& \leq \sum_{i=1}^{k}\big \vert \sem{H^{M}_{\chi_{i}}}(\vec r)- \sem{H^{N}_{\xi_{i}}}(\vec r) \big \vert.
\end{align*}
\end{proof}

%
%

}

%
%
%
%
%
%

\shortversion{
  \section{Graded Exponentials}
  \label{sec:grade}


  We give some hints about which ones of our
  results can be extended to Fuzz
  \cite{10.1145/1932681.1863568}. For results in
  Section~\ref{sec:dext_dobs}, we can show that
  the metric given in
  \cite{10.1145/1932681.1863568} in terms of
  metric logical relations is equal to the
  observational metric $\dobs_{\Gamma,\tau}(M,N)$
  given by
  \begin{equation*}
    \sup_{C[-] \colon (\Gamma,\tau) \to
        (\varnothing,\mathbf{R})}
    \inf
    \left\{
      r \in \bbRR \,\middle|\,
      \begin{array}{l}
        \text{if} \
        C[M] \hookrightarrow a,
        \ \text{then} \
        C[N] \hookrightarrow b
        \ \text{and} \ |a-b| \leq r, \\
        \text{and if} \
        C[N] \hookrightarrow a,
        \ \text{then} \
        C[M] \hookrightarrow b
        \ \text{and} \ |a-b| \leq r
      \end{array}
    \right\}
  \end{equation*}
  as long as we equip Fuzz with unary
  multiplications
  $r \times (-) \colon \oc_{r}\mathbf{R} \to
  \mathbf{R}$ for all $r \in \bbR_{\geq 0}$. For
  Fuzz \emph{without} unary multiplications, we
  need to define observational metric by
  observations at types
  $\oc_{r_{1}} \mathbf{R} \otimes \cdots \otimes
  \oc_{r_{n}} \mathbf{R}$ in order to prove that
  the observational metric coincides with the
  logical metric. As is shown in
  \cite{10.1145/3009837.3009890}, $\metcppo$
  provides an adequate semantics for Fuzz, and hence,
  the denotational semantics provides an
  over-approximation of the observational metric
  for Fuzz. Unfortunately, extending the interactive semantics to Fuzz
  is problematic. In order to do that, we
  need to model the graded exponentials
  $\oc_{r}(-)$. A natural way to do that in
  $\mathbf{Int}(\metcppo)$ would be to lift the
  denotation of $\oc_{r}(-)$ in $\metcppo$ to
  $\mathbf{Int}(\metcppo)$. However, this approach
  does not work because we can not lift
  contraction:
  $\oc_{r + s}\tau \multimap \oc_{r}\tau \otimes
  \oc_{s}\tau$. This failure stems from the nature
  of the Int-construction. Namely, because
  morphisms in $\mathbf{Int}(\metcppo)$ model
  bidirectional computation, lifting contraction
  to $\mathbf{Int}(\metcppo)$ requires
  co-contraction
  $\oc_{r}\tau \otimes \oc_{s}\tau \multimap\oc_{r
    + s}\tau$ in $\metcppo$, and there is no
  natural co-contraction in $\metcppo$. In the
  long version of this paper, we provide an interactive
  semantic model for Fuzz whose gradings are
  restricted to non-negative possibly infinite
  integers. We believe that this restriction is not so
  strong because, for example, closed terms of
  type
  $\oc_{k/n}\mathbf{R} \multimap
  \oc_{h/m}\mathbf{R}$ in Fuzz are
  ``definable'' as closed terms of type
  $\oc_{km}\mathbf{R} \multimap
  \oc_{hn}\mathbf{R}$. }

\longversion{\section{A Linear Programming Language with Graded
  Exponentials}
\label{sec:grade}
In the following part of this paper, we generalize
some of our arguments to a restriction of Fuzz,
namely, Fuzz without additive (co)products and
recursive types where gradings are non-negative
possibly infinite integers rather than real
numbers. We note that while we do not have
recursive types, we have recursion. Our
generalization goes as follows.
\begin{itemize}
\item We describe our target language, which we
  call $\LLL{S}$.
\item We extend the logical metric and the
  observational metric to $\LLL{S}$, and we show
  that these extensions coincide.
\item We extend the denotational metric and the
  interactive metric to $\LLL{S}$, and we show
  that the observational metric is bounded by
  these metrics.
\end{itemize}

\subsection{Syntax}
\label{sec:fuzz-syntax}
Let us give our extended target language, which we
call $\LLL{S}$, and its operational semantics. For
types, we have graded exponentials.
\begin{align*}
  \mathrm{Types}
  &&
  \tau,\sigma := \cdots \mid
  \oc_{n} \tau &&
  \mathrm{Environments}
  &&
  \Gamma,\Delta,\Xi := \varnothing \mid
  \Gamma,x:_{n}\tau
\end{align*}
In the definition of types and environments, $n$
varies over the set
$\mathbb{N}_{> 0}^{\infty} = \{n \in \mathbb{N}
\mid n > 0\} \cup \{\infty\}$ consisting of
positive possibly infinite integers.

For an environment $\Gamma$, we write
$\abs{\Gamma}$ for the syntactic object obtained
by removing all gradings from $x :_{n} \tau$ in
$\Gamma$. For environments $\Gamma$ and $\Delta$
such that $\abs{\Gamma} = \abs{\Delta}$, we
inductively define
an environment $\Gamma + \Delta$ by
\begin{equation*}
  \varnothing + \varnothing = \varnothing,
  \qquad
  (\Gamma, x :_{n} \tau) + (\Delta, x :_{m} \tau)
  = (\Gamma+\Delta), x :_{n + m} \tau.
\end{equation*}
When we write $\Gamma + \Delta$, we always suppose
that $\abs{\Gamma}$ is equal to $\abs{\Delta}$. We
write $\Gamma \geq \Delta$ when
$\abs{\Gamma}=\abs{\Delta}$ and for all
$x:_{n} \tau \in \Gamma$ and
$x:_{m} \tau \in \Delta$, we have $n \leq m$. For
an environment $\Gamma$ and
$n \in \mathbb{N}^{\infty}_{> 0}$, we define
$n \cdot \Gamma$ to be the environment obtained by
the componentwise multiplication of gradings in
$\Gamma$ by $n$.

Terms, values contexts are given by the following
BNF.
\begin{align*}
  \mathrm{Terms} &&
  M,N := \ &
  \cdots \mid
  \oc M \mid
  \letin{\oc x}{M}{N} \mid
  \fix{\tau,\sigma}{f,x,M} \mid
  \const{a} \cdot M \mid
  M + N \\
  \mathrm{Values} &&
  V,U := \ & \cdots
  \mid \fix{\tau,\sigma}{f,x,M}
  \mid \oc V
  \\
  \text{Contexts}
  &&
  C[-] :=
  &
  \cdots \mid
  \oc C[-] \mid
  \letin{\oc x}{C[-]}{M} \mid
  \letin{\oc x}{M}{C[-]}
\end{align*}
Namely, we have graded exponentials $\oc (-)$,
let-bindings for the graded comonad, recursion,
unary multiplications and addition. We add these
term constructors so as to simplify the definition
of the observational metric on $\LLL{S}$. Typing
rules are given in Figure~\ref{fig:fuzz-typing},
and evaluation rules are given in
Figure~\ref{fig:fuzz-evaluation_rules}. We
naturally extend the definition of
$\fterm(\Gamma,\tau)$, $\fval(\tau)$ and
$\fval(\Gamma)$ to $\LLL{S}$. We write
$C[-] \colon (\Gamma,\tau) \to (\Delta,\sigma)$
when $C[-]$ satisfies the following conditions.
\begin{itemize}
\item For all terms $\Gamma \vdash M:\tau$, we
  have $\Delta \vdash C[M]:\sigma$.
\item For a fresh variable $y$,
  we have $y:_{1} \oc_{k_{1}}\tau_{1} \multimap \cdots
  \oc_{k_{n}}\tau_{n} \multimap \tau,\Delta \vdash
  C[y\,\oc x_{1} \cdots \, \oc x_{n}] : \sigma$
  where $\Gamma = (x_{1}:_{k_{1}} \tau_{1},
  \ldots,x_{n}:_{k_{n}}\tau_{n})$.
\end{itemize}
We do not have $\mathbf{fix}$ in the definition of
contexts because when a hole of a context is under
$\mathbf{fix}$, then the second condition never
holds. Intuitively, the second condition means
that the hole $[-]$ appears linearly in the
context $C[]$.


\begin{figure}[t]
  \centering
  \begin{equation*}
    \prftree{
      a \in \bbR
    }{
      \Gamma \vdash  \const{a} : \mathbf{R}
    }
    \qquad
    \prftree{
    }{
      \Gamma \vdash \ast: \mathbf{I}
    }
    \qquad
    \prftree{
      f \in S
    }{
      \Gamma_{1} \vdash M_{1}:\mathbf{R}
    }{
      \ldots
    }{
      \Gamma_{\mathrm{ar}(f)} \vdash
      M_{\mathrm{ar}(f)}
      :\mathbf{R}
    }{
      \Gamma_{1} + \cdots + \Gamma_{\mathrm{ar}(f)}
      \vdash \const{f}
      (M_{1},\ldots,M_{\mathrm{ar}(f)}):\mathbf{R}
    }
  \end{equation*}
  \begin{equation*}
    \prftree{
      x :_{n} \tau \in \Gamma
    }{
      n \geq 1
    }{
      \Gamma
      \vdash x: \tau
    }
    \qquad
    \prftree{
      \Gamma
      \vdash M :\tau
    }{
      n \cdot \Gamma \leq \Delta
    }{
      \Delta \vdash \oc M : \oc_{n} \tau
    }
    \qquad
    \prftree{
      \Gamma,f:_{\infty} \tau \multimap \sigma,
      x:_{1}\tau
      \vdash M:\sigma
    }{
      \infty \cdot \Gamma \vdash
      \fix{\tau,\sigma}{f,x,M}
      : \tau \multimap \sigma
    }
  \end{equation*}
  \begin{equation*}
    \prftree{
      \Gamma \vdash M : \mathbf{R}
    }{
      \Delta \vdash N : \mathbf{R}
    }{
      \Gamma + \Delta \vdash M + N : \mathbf{R}
    }
    \qquad
    \prftree{
      \Gamma \vdash M : \oc_{n}\mathbf{R}
    }{
      |a| \leq n
    }{
      \Gamma \vdash \const{a} \cdot M : \mathbf{R}
    }
  \end{equation*}
  \begin{equation*}
    \prftree{
      \Gamma , x :_{1} \sigma  \vdash M:\tau
    }{
      \Gamma
      \vdash
      \lambda x:\sigma.\,M : \sigma \multimap \tau
    }
    \qquad
    \prftree{
      \Gamma \vdash M : \sigma \multimap \tau
    }{
      \Delta \vdash N : \sigma
    }{
      \Gamma + \Delta \vdash
      M\,N : \tau
    }
    \qquad
    \prftree{
      \Gamma \vdash M : \tau
    }{
      \Delta \vdash N : \sigma
    }{
      \Gamma + \Delta \vdash
      M \otimes N : \tau \otimes \sigma
    }
  \end{equation*}
  \begin{equation*}
    \prftree{
      \Gamma \vdash M : \mathbf{I}
    }{
      \Xi \vdash N : \tau
    }{
      n \cdot \Gamma \leq \Delta 
    }{
      \Delta + \Xi
      \vdash
      \letin{\ast}{M}{N} : \tau
    }
    \qquad
    \prftree{
      \Gamma \vdash M : \sigma_{1} \otimes \sigma_{2}
    }{
      \Xi,x:_{n}\sigma_{1},y:_{n}\sigma_{2} \vdash N:\tau
    }{
      n \cdot \Gamma \leq \Delta
    }{
      \Delta + \Xi
      \vdash
      \letin{x \otimes y}{M}{N}:\tau
    }
  \end{equation*}
  \begin{equation*}
    \prftree{
      \Gamma \vdash M:\oc_{m} \sigma
    }{
      \Xi, x :_{n \cdot m} \sigma \vdash
      N :\tau
    }{
      n \cdot \Gamma \leq \Delta
    }{
      \Delta +
      \Xi
      \vdash
      \letin{\oc x}{M}{N}
      :\tau
    }
  \end{equation*}
  \caption{Typing Rules}
  \label{fig:fuzz-typing}
\end{figure}

\begin{figure}[t]
  \centering
  \begin{equation*}
    \prftree{
      V
      \hookrightarrow
      V
    }
    \qquad
    \prftree{
      M_{1}
      \hookrightarrow
      \const{a_{1}}
    }{
      \ldots
    }{
      M_{n}
      \hookrightarrow
      \const{a_{n}}
    }{
      \const{f}(M_{1},\ldots,M_{n})
      \hookrightarrow
      \const{f(a_{1},\ldots,a_{n})}
    }
    \qquad
    \prftree{
      M \hookrightarrow \const{a}
    }{
      N \hookrightarrow \const{b}
    }{
      c = a + b
    }{
      M + N \hookrightarrow \const{c}
    }
  \end{equation*}
  \begin{equation*}
    \prftree{
      M \hookrightarrow \oc\const{b}
    }{
      ab = c
    }{
      \const{a} \cdot M \hookrightarrow \const{c}
    }
    \qquad
    \prftree{
      M
      \hookrightarrow
      \lambda x:\tau.\,M'
    }{
      N
      \hookrightarrow
      V
    }{
      M'[V/x]
      \hookrightarrow U
    }{
      M\,N
      \hookrightarrow
      U
    }
  \end{equation*}
  \begin{equation*}
    \prftree{
      M \hookrightarrow V
    }{
      N \hookrightarrow U
    }{
      M \otimes N
      \hookrightarrow
      V \otimes U
    }
    \qquad
    \prftree{
      M \hookrightarrow \ast
    }{
      N \hookrightarrow V
    }{
      \letin{\ast}{M}{N}
      \hookrightarrow
      V
    }
    \qquad
    \prftree{
      M
      \hookrightarrow
      V \otimes U
    }{
      N[V/x,U/y]
      \hookrightarrow
      W
    }{
      \letin{x \otimes y}{M}{N}
      \hookrightarrow
      W
    }
  \end{equation*}
  \begin{equation*}
    \prftree{
      M \hookrightarrow
      \fix{\tau,\sigma}{f,x,M'}
    }{
      N \hookrightarrow V
    }{
      M'[\fix{\tau,\sigma}{f,x,M'}/f,
      V/x
      ]
      \hookrightarrow U
    }{
      M\,N \hookrightarrow U
    }
  \end{equation*}
  \begin{equation*}
    \prftree{
      M \hookrightarrow V
    }{
      \oc M \hookrightarrow \oc V
    }
    \qquad
    \prftree{
      M \hookrightarrow \oc V
    }{
      N[V/x] \hookrightarrow U
    }{
      \letin{\oc x}{M}{N}
      \hookrightarrow
      U
    }
  \end{equation*}
  \caption{Evaluation Rules}
  \label{fig:fuzz-evaluation_rules}
\end{figure}
The following propositions can be shown by
induction on derivations of type judgements.
\begin{proposition}[Substitution]
  If $\Gamma \vdash M:\tau$ and
  $\gamma \in \fval(\Gamma)$, then
  $\vdash M\gamma:\tau$.
\end{proposition}
\begin{proposition}[Preservation]
  If $\vdash M:\tau$ and $M \hookrightarrow V$,
  then $\vdash V : \tau$.
\end{proposition}

In general, a type judgement
$\Gamma \vdash M : \tau$ may have different
derivations. For example,
$x:_{3}\mathbf{R} \vdash x + x : \mathbf{R}$ has
the following derivations.
\begin{equation*}
  \prftree{
    x :_{1} \mathbf{R} \vdash x:\mathbf{R}
  }{
    x :_{2} \mathbf{R} \vdash x:\mathbf{R}
  }{
    x :_{3} \mathbf{R} \vdash x + x : \mathbf{R}
  },
  \qquad
  \prftree{
    x :_{2} \mathbf{R} \vdash x:\mathbf{R}
  }{
    x :_{1} \mathbf{R} \vdash x:\mathbf{R}
  }{
    x :_{3} \mathbf{R} \vdash x + x : \mathbf{R}
  }.
\end{equation*}
We can show that grading is the only source of
non-uniqueness of derivations. This observation is
useful to define denotational semantics for $\LLL{S}$
in Section~\ref{sec:fuzz-denotational-metric}.
\begin{proposition}\label{prop:uniqueness_grading}
  For any environment $\Gamma$ and any term $M$,
  if $D_{1}$ is a derivation of
  $\Gamma \vdash M : \tau$ and $D_{2}$ is a
  derivation of $\Gamma \vdash M : \sigma$, then
  $\tau$ can be obtained by changing gradings in
  $\sigma$, and $D_{1}$ can also be obtained by
  changing gradings in $D_{2}$.  
\end{proposition}

\section{Logical Metric and Observational Metric}
\label{sec:fuzz-logic-metr-observ}

\subsection{Metric Logical Relation}
\label{sec:fuzz-metr-logic-relat}

We define metric logical relations
\begin{equation*}
  \{
  (-) \preceq_{r}^{\tau} (-)
  \subseteq \fterm(\tau) \times \fterm(\tau)\}_{
    \tau \in \type,\
    r \in \bbRR
  }
\end{equation*}
for $\LLL{S}$ by induction on $\tau$ as follows.
\begin{align*}
  M \preceq_{r}^{\mathbf{R}} N
  &
  \iff
  \text{if} \
  M \hookrightarrow a,
  \ \text{then} \
  N \hookrightarrow b
  \ \text{and} \
  |a - b| \leq r \\
  M \preceq_{r}^{\mathbf{I}} N
  &\iff
  \text{if} \
  M \hookrightarrow \ast,
  \ \text{then} \
  N \hookrightarrow \ast
  \\
  M \preceq_{r}^{\tau \otimes \sigma} N
  &\iff
  \text{if} \
  M \hookrightarrow V \otimes V',
  \ \text{then} \
  N \hookrightarrow U \otimes U'
  \\
  &\mathrel{\phantom{\iff}}
  \text{and} \
  \exists s,s' \in \bbRR,\,
  V \preceq_{s}^{\tau} U
  \text{ and }
  V' \preceq_{s'}^{\sigma} U'
  \text{ and }
  s + s' \leq r
  \\
  M \preceq_{r}^{\tau \multimap \sigma} N
  &\iff
  \text{if} \
  M \hookrightarrow V,
  \ \text{then} \
  N \hookrightarrow V'
  \\
  &\mathrel{\phantom{\iff}}
  \text{and} \
  \forall U,U' \in \fval(\tau),\,
  \text{if} \
  U \preceq_{s}^{\tau} U',
  \ \text{then} \
  V\,U \preceq_{r + s}^{\sigma} V'\,U'
  \\
  M \preceq_{r}^{\oc_{n}\tau} N
  &\iff
  \text{if} \
  M \hookrightarrow \oc V,
  \ \text{then} \
  N \hookrightarrow \oc U
  \\
  &\mathrel{\phantom{\iff}}
  \text{and} \
  \exists s \in \bbRR,\,
  V \precsim_{s}^{\tau} U
  \ \text{and} \
  r \geq
  n \, s.
\end{align*}
Let
$\Gamma=(x_{1}:_{k_{1}}\sigma_{1},\ldots,
x_{n}:_{k_{n}}\sigma_{n})$ be an environment. For
$\gamma=(V_{1},\ldots,V_{n})$ and
$\gamma'=(V'_{1},\ldots,V'_{n})$ in
$\fval(\Gamma)$, and for
$\epsilon =(r_{1},\ldots,r_{n})\in (\bbRR)^{n}$,
we write
$\gamma \preceq_{\epsilon}^{\Gamma} \gamma'$ when
we have
$V_{1} \preceq_{r_{1}}^{\sigma_{1}} V'_{1},
\ldots,V_{n} \preceq_{r_{n}}^{\sigma_{n}} V'_{n}$.
We define $\epsilon \cdot \Gamma$ to be
$r_{1}k_{1} + \cdots + r_{n}k_{n}$. Here, we
define $0 \infty$ to be $0$. Then, for terms
$\Gamma \vdash M :\tau$ and
$\Gamma \vdash N:\tau$, we define
$\dext_{\Gamma,\tau}(M,N) \in \bbRR$ by
\begin{equation*}
  \dext_{\Gamma,\tau}(M,N)
  =
  \inf \left\{r \in \bbRR \,\middle|\,
    \begin{array}{l}
      \forall \gamma,\gamma' \in \fval(\Gamma),\,
      \text{if} \
      \gamma \preceq_{\epsilon}^{\Gamma} \gamma',
      \ \text{then} \\
      M\gamma
      \preceq_{r + \epsilon \cdot \Gamma}^{\tau}
      N\gamma'
      \ \text{and} \
      N\gamma
      \preceq_{r + \epsilon \cdot \Gamma}^{\tau}
      M\gamma'
    \end{array}
  \right\}.
\end{equation*}

We call $\dext$ the \emph{logical metric} on
$\LLL{S}$. For later use, we prove the fundamental
lemma.
\begin{lemma}\label{lem:fuzz-basic-lemma}
  Let
  $\Gamma =
  (x_{1}:_{k_{1}}\tau_{1},\ldots,x_{n}:_{k_{n}}\tau_{n})$
  be an environment, and let
  $\Gamma \vdash M : \tau$ be a term. Given
  $\gamma, \gamma' \in \fval(\Gamma)$ such that
  $\gamma \preceq_{\epsilon}^{\Gamma} \gamma'$,
  then we have
  $M\gamma \preceq_{\epsilon \cdot \Gamma}^{\tau}
  M\gamma'$.
\end{lemma}
\begin{proof}
  The proof is essentially the same with
  \cite{10.1145/1932681.1863568} using step
  indexed logical relations counting the number of
  $\mathbf{fix}$-reductions in
  $M \hookrightarrow V$.
\end{proof}

\section{Observational Metric}
\label{sec:fuzz-observational-metric}

For terms $\Gamma \vdash M : \tau$ and
$\Gamma \vdash N : \tau$, we define
$\dobs_{\Gamma,\tau}(M,N) \in \bbRR$ by
\begin{equation*}
  \dobs_{\Gamma,\tau}(M,N)
  =
  \sup_{C[-] \colon (\Gamma,\tau)
    \to (\varnothing,\mathbf{R})
  }
  \inf
  \left\{r \in \bbRR \, \middle|\,
    C[M] \sqsubseteq_{r} C[N]
    \ \text{and} \
    C[N] \sqsubseteq_{r} C[M]
  \right\}
\end{equation*}
where for $\vdash L,L':\mathbf{R}$,
\begin{equation*}
  L \sqsubseteq_{r} L'
  \iff
  \text{if} \ L \hookrightarrow a,
  \ \text{then} \
  L' \hookrightarrow b
  \ \text{and} \
  |a-b| \leq r.
\end{equation*}

\begin{theorem}\label{thm:fuzz-dobs=dext}
  For terms $\Gamma \vdash M : \tau$ and
  $\Gamma \vdash N : \tau$, we have
  $\dobs_{\Gamma,\tau}(M,N) =
  \dext_{\Gamma,\tau}(M,N)$.
\end{theorem}
\begin{proof}
  The statement follows from
  Lemma~\ref{lem:fuzz-dext<=dobs}
  and
  Lemma~\ref{lem:fuzz-dobs<=dext}
  shown below.
\end{proof}

\begin{lemma}\label{lem:fuzz-dext<=dobs}
  For any environment
  $\Gamma= (x_{1}:_{k_{1}}\tau_{1},\ldots,x_{n}:_{k_{n}}\tau_{n})$
  and any pair of terms $\Gamma \vdash M : \tau$
  and $\Gamma \vdash N : \tau$, if
  $\gamma \in \fval(\Gamma)$ and
  $\dobs_{\Gamma,\tau}(M,N) < \infty$, then
  $M\gamma
  \preceq_{\dobs_{\Gamma,\tau}(M,N)}^{\tau}
  N\gamma$.
\end{lemma}
\begin{proof}
  By induction on $\tau$, we show that for all
  $\Gamma \vdash M : \tau$ and
  $\Gamma \vdash N : \tau$, we have
  $M\gamma
  \preceq^{\tau}_{\dobs_{\Gamma,\tau}(M,N)}
  N\gamma$. We only give a proof for
  $\tau = \sigma \otimes \rho$ and
  $\tau = \sigma \multimap \rho$. (The case of
  $\tau = \sigma \otimes \rho$) Given
  $\gamma \in \fval(\Gamma)$, if $M\gamma$
  diverges, then by the definition of $\preceq$,
  we obtain
  $M\gamma \preceq_{\dobs_{\Gamma,\tau}(M,N)}
  N\gamma$. If we have
  $M\gamma \hookrightarrow V_{1} \otimes V_{2}$,
  then since $\dobs_{\Gamma,\tau}(M,N) <\infty$,
  there are $U_{1} \in \fval(\sigma)$
  and $U_{2} \in \fval(\rho)$ such that
  $N\gamma \hookrightarrow U_{1} \otimes U_{2}$.
  By the induction hypothesis on $\sigma$ and
  $\rho$, we have
  $V_{1}
  \preceq_{\dobs_{\varnothing,\sigma}(V_{1},U_{1})}^{\sigma}
  U_{1}$ and
  $V_{2}
  \preceq_{\dobs_{\varnothing,\rho}(V_{2},U_{2})}^{\rho}
  U_{2}$. It remains to check that
  $\dobs_{\varnothing,\sigma}(V_{1},U_{1}) +
  \dobs_{\varnothing,\rho}(V_{2},U_{2}) \leq
  \dobs_{\Gamma,\sigma \otimes \rho} (M,N)$.
  Below, we suppose that
  $\gamma=(W_{1},\ldots,W_{n})$. We write
  $\widetilde{M}$ and $\widetilde{N}$ for
  \begin{align*}
    (\lambda x_{1}:\tau_{1}.\,\cdots
    \lambda x_{n}:\tau_{n}.\,M)
    \,W_{1}\,\cdots\,W_{n} \\
    (\lambda x_{1}:\tau_{1}.\,\cdots
    \lambda x_{n}:\tau_{n}.\,N)
    \,W_{1}\,\cdots\,W_{n}
  \end{align*}
  respectively. Then,
  \begin{align*}
    &\dobs_{\varnothing,\sigma}(V_{1},U_{1}) +
    \dobs_{\varnothing,\rho}(V_{2},U_{2})
    \\
    &\leq
    \sup_{
      \begin{array}{c}
        \scriptstyle
        C[-] \colon (\varnothing,\sigma)
        \to (\varnothing,\mathbf{R}) \\
        \scriptstyle
        D[-] \colon (\varnothing,\rho)
        \to (\varnothing,\mathbf{R}) \\
      \end{array}
    }
    \inf
    \left\{
      r \in \bbR_{\geq 0}
      \,\middle|\,
      \begin{array}{l}
        C[U_{1}] + D[U_{2}]
        \sqsubseteq_{r} C[V_{1}]
        + D[V_{2}] \\
        \ \text{and} \\
        C[V_{1}] + D[V_{2}]
        \sqsubseteq_{r}
        C[U_{1}] + D[U_{2}]
      \end{array}
    \right\}
    \\
    &\leq
    \sup_{
      \begin{array}{c}
        \scriptstyle
        C[-] \colon (\varnothing,\sigma)
        \to (\varnothing,\mathbf{R}) \\
        \scriptstyle
        D[-] \colon (\varnothing,\rho)
        \to (\varnothing,\mathbf{R}) \\
      \end{array}
    }
    \inf
    \left\{
      r \in \bbR_{\geq 0}
      \, \middle|\,
      \begin{array}{l}
        \letin{x \otimes y}{
        \widetilde{M}}{C[x]+D[y]} \\
        \qquad \sqsubseteq_{r} \letin{x \otimes y}{
        \widetilde{N}
        }{C[x]+D[y]} \\
        \ \text{and} \ \\
        \letin{x \otimes y}{
        \widetilde{N}
        }{C[x]+D[y]} \\
        \qquad \sqsubseteq_{r} \letin{x \otimes y}{
        \widetilde{M}
        }{C[x]+D[y]}      
      \end{array}
    \right\}
    \\
    &\leq
    \dobs_{\Gamma,\sigma\otimes\rho}
    (M,N).
  \end{align*}
  We note that we use the addition to construct
  contexts. We need unary multiplications to prove
  the case where $\tau = \oc_{n}\sigma$. (The case
  of $\tau = \sigma \multimap \rho$) Given
  $\gamma \in \fval(\Gamma)$, if $M\gamma$
  diverges, then by the definition of $\preceq, $
  we have
  $M\gamma
  \preceq^{\tau}_{\dobs_{\Gamma,\tau}(M,N)}
  N\gamma$. Let us assume that we have
  $M \gamma \hookrightarrow V$, and we show that
  $M\gamma
  \preceq^{\tau}_{\dobs_{\Gamma,\tau}(M,N)}
  N\gamma$. From the assumption, since
  $\dobs_{\Gamma,\tau}(M,N)< \infty$, we see that
  we have $N \gamma \hookrightarrow V'$ for some
  $V' \in \fval(\tau)$. In order to prove
  $M\gamma
  \preceq^{\tau}_{\dobs_{\Gamma,\tau}(M,N)}
  N\gamma$, we show that for all
  $U \preceq_{r}^{\sigma} U'$, we have
  $V\,U \preceq_{r +
    \dobs_{\Gamma,\tau}(M,N)}^{\rho} V'\,U'$. By
  Lemma~\ref{lem:fuzz-basic-lemma}, we obtain
  $V'\,U \preceq_{r}^{\rho} V'\,U'$. Hence, by the
  triangle inequality, it remains to check
  $V\,U \preceq_{\dobs_{\Gamma,\tau}(M,N)}^{\rho}
  V'\,U$. By the definition of $\preceq$, this is
  equivalent to
  $M\gamma\,U
  \preceq_{\dobs_{\Gamma,\tau}(M,N)}^{\rho}
  N\gamma\,U$. It follows from the induction
  hypothesis on $\rho$ that we have
  \begin{equation*}
    M\gamma\,U \preceq_{\dobs_{\Gamma,\rho}(M\,U,N\,U)}^{\rho}
    N\gamma\,U.
  \end{equation*}
  Since
  $\dobs_{\Gamma,\rho}(M\,U,N\,U) \leq
  \dobs_{\Gamma,\tau}(M,N)$, we obtain the claim.
\end{proof}
\begin{lemma}\label{lem:fuzz-dobs<=dext}
  For terms $\Gamma \vdash M : \tau$
  and $\Gamma \vdash N : \tau$,
  we have
  $\dobs_{\Gamma,\tau}(M,N)
  \leq \dext_{\Gamma,\tau}(M,N)$.
\end{lemma}
\begin{proof}
  For simplicity, we suppose that
  $\Gamma=(x:_{k}\sigma)$. We show that if
  \begin{align*}
    \lambda y:\oc_{k}\sigma.\,
    \letin{\oc x}{y}{M}
    &\preceq_{r}^{\oc_{k}\sigma \multimap \tau}
    \lambda y:\oc_{k}\sigma.\,
    \letin{\oc x}{y}{N} \\
    \lambda y:\oc_{k}\sigma.\,
    \letin{\oc x}{y}{N}
    &\preceq_{r}^{\oc_{k}\sigma \multimap \tau}
    \lambda y:\oc_{k}\sigma.\,
    \letin{\oc x}{y}{M}
  \end{align*}
  for some $r \in \bbRR$, then
  $\dobs_{\Gamma,\tau}(M,N) \leq r$. Given a
  context
  $C[-] \colon (\Gamma,\tau) \to
  (\varnothing,\mathbf{R})$, it follows from
  adequacy of denotational model
  (Theorem~\ref{thm:fuzz-den-adequacy}) that
  \begin{align*}
    C[M] \hookrightarrow a
    &\iff
    (\lambda z:\oc_{k}\sigma \multimap \tau.\,
    C[z\,\oc x])(\lambda y:\oc_{k}\sigma.\,
    \letin{\oc x}{y}{M})
    \hookrightarrow a, \\
    C[N]\hookrightarrow b
    &\iff
    (\lambda z:\oc_{k}\sigma \multimap \tau.\,
    C[z\,\oc x])(\lambda y:\oc_{k}\sigma.\,
    \letin{\oc x}{y}{N})
    \hookrightarrow b.
  \end{align*}
  Hence, it follows from
  Lemma~\ref{lem:fuzz-basic-lemma} that $C[M]$
  converges if and only if $C[N]$ converges. If
  $C[M] \hookrightarrow a$ and
  $C[N] \hookrightarrow b$, then we have
  $|a-b| \leq r$. Since this holds for any context
  $C[-] \colon (\Gamma,\tau) \to
  (\varnothing,\mathbf{R})$, we obtain
  $\dobs_{\Gamma,\tau}(M,N) \leq r$.
\end{proof}

\section{Denotational Metric}
\label{sec:fuzz-denotational-metric}

Let $\mathbf{MetCpo}_{\bot}$ be the category of
pointed metric cpos and strict continuous and
non-expansive functions. Concretely, objects in
$\mathbf{MetCpo}_{\bot}$ are metric cpos with
least elements $\bot$ such that the distances
$d(\bot,x)$ are $\infty$ when $x \neq \bot$, and
morphisms from $X$ to $Y$ are bottom-preserving.
As is shown in \cite{10.1145/3009837.3009890},
$\mathbf{MetCpo}_{\bot}$ provides an adequate
semantics for Fuzz. By restricting their result to
our language, we obtain adequacy for $\LLL{S}$.
For a term $\Gamma \vdash M : \tau$, let us write
$\sem{M}^{\mathrm{den}} \colon
\sem{\Gamma}^{\mathrm{den}} \to
\sem{\tau}^{\mathrm{den}}$ for the interpretation
of $M$ in $\mathbf{MetCpo}_{\bot}$. We note that
$\sem{M}^{\mathrm{den}}$ is defined with respect to
the type judgement rather than type derivations
of $\Gamma \vdash M : \tau$. This can be checked by
observing that the underlying continuous function
of $\sem{M}^{\mathrm{den}}$
is obtained by first transforming $\LLL{S}$
into the $\lambda_{c}$-calculus
\cite{10.5555/77350.77353}
and then interpreting the transformed term
in $\mathbf{Cpo}_{\bot}$. In the transformation
of $\LLL{S}$ into the $\lambda_{c}$-calculus,
gradings are dropped, and therefore, all
derivations of a type judgement $\Gamma \vdash M :\tau$
are transformed into the same derivation in
the $\lambda_{c}$-calculus.

\begin{theorem}[\cite{10.1145/3009837.3009890}]
  \label{thm:fuzz-den-adequacy}
  Let $\vdash M : \tau$ be a term in $\LLL{S}$.
  \begin{itemize}
  \item If $M \hookrightarrow V$, then
    $\sem{M}^{\mathrm{den}} =
    \sem{V}^{\mathrm{den}}$.
  \item If $\sem{M}^{\mathrm{den}} \neq \bot$,
    then there is a value $V \in \fval(\tau)$ such
    that $M \hookrightarrow V$.
  \end{itemize}
\end{theorem}

For terms $\Gamma \vdash M :\tau$ and
$\Gamma \vdash N : \tau$, we define
$\dden_{\Gamma,\tau}(M,N) \in \bbRR$ by
\begin{equation*}
  \dden_{\Gamma,\tau}(M,N) =
  d(\sem{M}^{\mathrm{den}},\sem{N}^{\mathrm{den}}).
\end{equation*}
It is easy to see that $\dden$ is a metric on
$\LLL{S}$. We call $\dden$ the \emph{denotational
  metric} on $\LLL{S}$.

It follows from adequacy of
$\mathbf{MetCpo}_{\bot}$ that $\dobs$ is bounded
by $\dden$.
\begin{theorem}\label{thm:fuzz-dctx<=dden}
  $\dobs \leq \dden$.
\end{theorem}
\begin{proof}
  If there is a context
  $C[-] \colon (\Gamma,\tau) \to
  (\varnothing,\mathbf{R})$ such that $C[M]$
  converges and $C[N]$ diverges, then, by
  adequacy, we have
  $\sem{C[M]} = \sem{\const{a}}$.
  $\sem{C[N]} = \bot$ for some $a \in \mathbb{R}$.
  Hence,
  $\dden_{\Gamma,\tau}(M,N) \geq
  \dden_{\varnothing,\mathbf{R}} (C[M],C[N]) =
  \infty \geq \dobs_{\Gamma,\tau}(M,N)$.
  Similarly, when $C[M]$ diverges and $C[N]$
  converges, then
  $\dden_{\Gamma,\tau}(M,N) \geq
  \dobs_{\Gamma,\tau}(M,N)$. Below, we suppose
  that for any context
  $C[-] \colon
  (\Gamma,\tau)\to(\varnothing,\mathbf{R})$,
  $C[M]$ diverges if and only if $C[N]$ diverges.
  In this case, it follows from
  Theorem~\ref{thm:fuzz-den-adequacy} that if
  $C[M] \hookrightarrow a$ and
  $C[N] \hookrightarrow b$, then
  $|a-b| \leq \dden_{\Gamma,\tau}(M,N)$. Hence, we
  obtain the statement.
\end{proof}

\section{Interactive Semantic Model}
\label{sec:fuzz-inter-semant-model}

\subsection{Preparation}
\label{sec:fuzz-preparation}

\subsubsection{Structures for Interpreting Graded Exponentials}
\label{sec:fuzz-struct-interpr-grad}

We prepare structures on the category
$\mathbf{Int}(\metcppo)$ that we use to interpret
graded exponentials in $\LLL{S}$. For
$X \in \metcppo$ and
$n \in \mathbb{N}^{\infty}_{>0}$, we define
$n \cdot X$ to be the countably infinite product
of the underlying cpo of $X$ equipped with the
following metric:
\begin{equation*}
  d((x_{i})_{i \in \mathbb{N}},(y_{i})_{i \in \mathbb{N}})
  = \sum_{i < n} d(x_{i},y_{i}).
\end{equation*}
It is not difficult to check that $n \cdot (-)$ is
a traced symmetric monoidal functor on
$\metcppo$. Hence, we can lift the
functors $n \cdot (-)$ to symmetric monoidal
functors on $\mathbf{Int}(\metcppo)$.
Abusing notation, we also denote the functors on
$\mathbf{Int}(\metcppo)$ by $n \cdot (-)$.
To be concrete, on objects $X = (X_{+},X_{-})$ in
$\mathbf{Int}(\metcppo)$, we have
$n \cdot X = (n \cdot X_{+},n \cdot X_{-})$.

In order to interpret dereliction,
digging and contraction of $\LLL{S}$, for
$n,m \in \mathbb{N}^{\infty}_{>0}$, we choose bijections
$u_{n,m} \colon \mathbb{N} \times \mathbb{N} \to
\mathbb{N}$ and
$v_{n,m} \colon \{0,1\} \times \mathbb{N}\to
\mathbb{N}$ such that
\begin{itemize}
\item $u_{n,m}$ embeds
  $\{(i,j) \in \mathbb{N} \times \mathbb{N} \mid
  i < n \ \text{and} \ j < m\}$ into
  $\{i \in \mathbb{N} \mid i < nm\}$; and
\item $v_{n,m}$ embeds
  $\{(0,i) \mid i < n\}
  \cup
  \{(1,i) \mid i < m\}$ into
  $\{i \in \mathbb{N} \mid i < n + m\}$.
\end{itemize}
Then, we define the following morphisms
\begin{align*}
  d_{n,X} \colon n \cdot X &\to
  X, \\
  \delta_{n,m,X} \colon nm \cdot X &\cong
  n \cdot (m \cdot X), \\
  c_{n,m,X} \colon (n + m) \cdot X &\cong
  (n \cdot X) \otimes (m \cdot X)
\end{align*}
for $n,m \in \mathbb{N}^{\infty}_{>0}$ by
\begin{align*}
  d_{n,X}((x_{i})_{i \in \mathbb{N}},y)
  &= ((y,\bot,\bot,\ldots),x_{0}), \\
  \delta_{n,m,X}((x_{i})_{i \in \mathbb{N}},
  ((y_{i,j})_{j \in \mathbb{N}})_{i \in \mathbb{N}}) 
  &=
  ((y_{u^{-1}_{n,m}(i)})_{i \in \mathbb{N}},
  ((x_{u_{n,m}(i,j)})_{j \in \mathbb{N}}
  )_{i \in \mathbb{N}}),
  \\
  c_{X}((x_{i})_{i \in \mathbb{N}},
  ((y_{0,i})_{i \in \mathbb{N}},
  (y_{1,i})_{i \in \mathbb{N}}))
  &=
  ((y_{v_{n,m}(i)})_{i \in \mathbb{N}},
  ((x_{v^{-1}_{n,m}(0,i)})_{i \in \mathbb{N}},
  (x_{v^{-1}_{n,m}(1,i)})_{i \in \mathbb{N}})
  ).
\end{align*}
We also give a bit more general dereliction
$\tilde{d}_{n,m,X} \colon (n + m) \cdot X
\to n \cdot X$ by
\begin{multline*}
  \tilde{d}_{n,m,X}((x_{i})_{i \in \mathbb{N}},
  (y_{i})_{i \in \mathbb{N}})
  = ((y_{0},y_{1},\ldots,y_{n-1},
  \overbrace{\bot,\ldots,\bot}^{m},y_{n},y_{n+1},\ldots),
  \\
  (x_{0},x_{1},\ldots,x_{n-1},x_{n+m-1},x_{n+m},\ldots)).
\end{multline*}
We note that morphisms $d_{n,X}$,
$\delta_{n,m,X}$, $c_{n,m,X}$, $\tilde{d}_{n,m,X}$
and $w_{X}$ are not natural with respect to $X$.
Still, we can show that they are \emph{pointwise}
natural transformation, that is, these morphisms
satisfy naturality conditions for global elements.
For example, for all $x \colon I \to X$, we have
$d_{n,X} \circ (n \cdot x) = x \circ d_{n,I} = x$. For
more details on pointwise naturality, see
\cite{abramsky_haghverdi_scott_2002}.


\subsubsection{Structures for Interpreting
  Weakening and Cbv Evaluation}
\label{sec:struct-interpr-weak}

We also need structures on
$\mathbf{Int}(\metcppo)$ to interpret
weakening and call-by-value evaluation. For the
former, for an object
$X \in \mathbf{Int}(\metcppo)$, we define
$w_{X} \colon X \to I$ in
$\mathbf{Int}(\metcppo)$ to be
$\bot \colon X_{+} \to X_{-}$ in
$\metcppo$. For the latter, we
use the Kleisli category of a continuation monad
\begin{equation*}
  TX =
  X \otimes (K,K)
  \cong (X \multimap (K,I))
  \multimap (K,I)
\end{equation*}
on $\mathbf{Int}(\metcppo)$ where
$K \in \metcppo$ is the Sierpi\'nski
space $\{\bot \leq \top\}$ equipped with
$d(\bot,\top) = \infty$.
We denote the unit, the multiplication and
the strength of the monad $T$ by
\begin{align*}
  \eta_{X} &\colon X \to TX, \\
  \mu_{X} &\colon TTX \to TX, \\
  \mathrm{str}_{X,Y} &\colon TX \otimes Y \to T(X \otimes Y),
\end{align*}
and we write
$\mathrm{dstr}_{X,Y} \colon TX \otimes TY \to T(X
\otimes Y)$ for the double strength
\begin{equation*}
  TX \otimes TY
  \longrightarrow
  T(X \otimes TY)
  \longrightarrow
  TT(X \otimes Y)
  \longrightarrow
  T(X \otimes Y).
\end{equation*}
For $n \in \mathbb{N}^{\infty}_{> 0}$,
we define a morphism
\begin{equation*}
  \xi_{n,X} \colon n \cdot TX
  \to T(n \cdot X)
\end{equation*}
in $\mathbf{Int}(\metcppo)$ to be
\begin{equation*}
  n \cdot TX
  = n \cdot (X \otimes (K,K))
  \cong n \cdot X \otimes n \cdot (K,K)
  \xrightarrow{(n \cdot X) \otimes d}
  (n \cdot X) \otimes (K,K)
  =
  T(n \cdot X).
\end{equation*}
We use this distributivity of $n \cdot(-)$ over
$T$ to model the action of graded exponentials on
terms.

\subsubsection{Structures for Interpreting Constants}
\label{sec:struct-interpr-const}

For interpretation of constants $\const{a}$ and
$\const{f}(M_{1},\ldots,M_{\mathrm{ar}(f)})$, we
follow the interpretation of terms in $\LL{S}$: we
interpret the base type $\mathbf{R}$ by $(R,I)$,
and we use $\lfloor a \rfloor \colon I \to R$ and
$\lfloor f \rfloor \colon R^{\otimes
  \mathrm{ar}(f)} \to R$ to interpret real numbers
and first order functions. For interpretation
of unary multiplications, we use
\begin{equation*}
  \mathrm{mult}_{a,n} \colon
  n \cdot (R,I) \to (R,I)
\end{equation*}
for $a \in \mathbb{R}$ and $n \in \mathbb{N}^{\infty}_{>0}$
such that $|a| \leq n$ given by
\begin{equation*}
  \mathrm{mult}_{a,n}((x_{i})_{i \in \mathbb{N}},\ast)
  = (\ast,a(x_{0} + \cdots + x_{n-1})/n).
\end{equation*}

\subsection{Interactive Semantic Model and its
  Associated Metrics}
\label{sec:inter-model-its}

Based on preparations in the previous sections, we
give interpretation of $\LLL{S}$ in the Kleisli
category $\mathbf{Int}(\metcppo)_{T}$.

Types in $\LLL{S}$ are interpreted as follows:
\begin{align*}
  \sem{\mathbf{R}}^{\mathrm{int}}
  &= (R,I), \\
  \sem{\mathbf{I}}^{\mathrm{int}}
  &= (I,I), \\
  \sem{\tau \otimes \sigma}^{\mathrm{int}}
  &= \sem{\tau}^{\mathrm{int}}
  \otimes \sem{\sigma}^{\mathrm{int}} \\
  \sem{\tau \multimap \sigma}^{\mathrm{int}}
  &= \sem{\tau}^{\mathrm{int}}
  \multimap T\sem{\sigma}^{\mathrm{int}}
  = (\sem{\tau}^{\mathrm{int}})^{\ast}
  \otimes \sem{\sigma}^{\mathrm{int}}
  \otimes (K,K) \\
  \sem{\oc_{n} \tau}^{\mathrm{int}}
  &= n \cdot \sem{\tau}^{\mathrm{int}}
\end{align*}
where $(X_{+},X_{-})^{\ast}$ is defined to be
$(X_{-},X_{+})$. As usual, we interpret
environments as follows:
\begin{equation*}
  \sem{(x:_{n}\tau,\ldots,y:_{m}\sigma)}^{\mathrm{int}}
  = n \cdot \sem{\tau}^{\mathrm{int}} \otimes \cdots \otimes
  m \cdot \sem{\sigma}^{\mathrm{int}}.
\end{equation*}
We next define interpretation of type judgements
in $\LLL{S}$ in
Figure~\ref{fig:interpretation_fuzz} where we
simply write $\sem{-}$ for
$\sem{-}^{\mathrm{int}}$ for the sake of
legibility. We note that the interpretation is
given with respect to type derivations rather than
type judgements. 

We call this model the \emph{interactive semantic
  model} for $\LLL{S}$. The interactive semantic
model gives rise to another semantically obtained
family of metrics. For terms
$\Gamma\vdash M :\tau$ and
$\Gamma \vdash N : \tau$, we define
$\dint_{\Gamma,\tau}(M,N) \in \bbRR$ by
\begin{equation*}
  \dint_{\Gamma,\tau}(M,N) =
  d(\sem{M}^{\mathrm{int}},\sem{N}^{\mathrm{int}}).
\end{equation*}

We prove adequacy of the interactive semantic
model, which will be used to prove
$\dobs \leq \dint$. Below, for $\vdash M : \tau$,
we write $\sem{M} \Downarrow$ when there is no
$f \colon I \to \sem{\tau}$ such that
$\sem{M} = f \otimes \bot_{(K,K)}$ where
$\bot_{X}$ denotes the least element of
$\mathbf{Int}(\metcppo)(I, X)$.
\begin{theorem}
  Let $\vdash M : \tau$ be a term in $\LLL{S}$.
  \begin{itemize}
  \item If $M \hookrightarrow V$, then there is a
    derivation of $\vdash V:\tau$ such that
    $\sem{M} = \sem{V}$.
  \item If $\sem{M} \Downarrow$, then there is a
    value $V$ such that $M \hookrightarrow V$.
  \end{itemize}
\end{theorem}
\begin{proof}
  We can show the first claim by induction on the
  derivation of $M \hookrightarrow V$ using
  pointwise naturality of morphisms in
  Section~\ref{sec:fuzz-struct-interpr-grad}, and we
  omit the detail. We prove the second claim by
  means of logical relations. For a type $\tau$,
  we define a binary relation
  \begin{equation*}
    P_{\tau}
    \subseteq
    \mathbf{Int}(\metcppo)
    (I,\sem{\tau})
    \times
    \fval(\tau)
  \end{equation*}
  by
  \begin{align*}
    P_{\mathbf{R}}
    &= \{(\sem{\const{a}},
    \const{a}) \mid a \in \bbR\}, \\
    P_{\mathbf{I}}
    &= \{(\sem{\ast},\ast)\}, \\
    P_{\tau \otimes \sigma}
    &= \{(f \otimes g , V \otimes U) \mid
    (f,V) \in P_{\tau} \ \text{and} \
    (g,U) \in P_{\sigma}
    \}, \\
    P_{\tau \multimap \sigma}
    &=
    \{(f,V) \mid
    \forall (g,U) \in P_{\tau},\,
    (f \bullet g, V\,U) \in \overline{P}_{\sigma}
    \}, \\
    P_{\oc_{n} \tau}
    &=
    \{(n \cdot f,\oc V) \mid
    (f,V) \in P_{\tau}
    \}
  \end{align*}
  where
  \begin{multline*}
    \overline{P}_{\tau}
    =
    \{(\eta \circ f , M) \mid
    M \hookrightarrow V
    \ \text{and} \
    (f,V) \in P_{\tau}\}
    \\ {}
    \cup {}
    \{(f \otimes \bot_{(K,K)},M) \mid
    f \colon I \to \sem{\tau}
    \ \text{and} \
    M \in \fterm(\tau)\}
  \end{multline*}
  and $f \bullet g$ is given by
  \begin{equation*}
    I \xrightarrow{f \otimes g}
    (\sem{\tau} \multimap T\sem{\sigma})
    \otimes
    \sem{\tau}
    \xrightarrow{\text{eval}}
    T\sem{\sigma}.
  \end{equation*}
  By the definition of $P_{\tau}$, we can show
  that $P_{\tau}$ and $\overline{P}_{\tau}$ are
  closed under taking least upper bounds of the
  first component: for all
  $(x_{1},M),(x_{2},M),\ldots \in P_{\tau}$, if
  $x_{1} \leq x_{2} \leq \cdots$, then we have
  $\left(\bigvee_{n \in \mathbb{N}}x_{n},M\right)
  \in P_{\tau}$. We show basic lemma for
  $P_{\tau}$: for any
  $\Gamma = (x:_{n_{1}}\sigma,\ldots,y:_{n_{k}}\rho)$, any
  $\Gamma \vdash M : \tau$ and any
  $(v,V) \in P_{\sigma},\ldots, (u,U) \in
  P_{\rho}$, we have
  \begin{equation*}
    (\sem{M} \circ ((n_{1} \cdot v) \otimes \cdots \otimes
    (n_{k} \cdot u)),
    M[V/x,\ldots,U/y]) \in \overline{P}_{\tau} .
  \end{equation*}
  We only check the case for $\mathbf{fix}$. The
  other cases are not difficult to check. What we
  check is: for environments
  $\Gamma = (x_{1}:_{k_{1}}\rho_{1},
  \ldots,x_{n}:_{k_{n}} \rho_{n})$ and
  $\Delta = (x_{1}:_{k'_{1}}\rho_{1},
  \ldots,x_{n}:_{k'_{n}} \rho_{n})$ such that
  $\infty \cdot \Gamma \leq \Delta$ and for a term
  $\Delta \vdash \fix{\tau,\sigma}{f,x,M} : \tau
  \multimap \sigma$, given
  $(v_{1},V_{1}) \in P_{\rho_{1}},\ldots,
  (v_{n},V_{n}) \in P_{\rho_{n}}$, we have
  \begin{equation*}
    (\sem{\fix{\tau,\sigma}{f,x,M}} \circ ((k_{1}
    \cdot v_{n}) \otimes \cdots \otimes (k_{n} \cdot
    v_{n})),
    \fix{\tau,\sigma}{f,x,M}[V_{1}/x_{1},\ldots,V_{n}/x_{n}])
    \in P_{\tau \multimap \sigma}.  
  \end{equation*}
  Let
  $\varphi \colon \mathbf{Int}(\metcppo) (\infty
  \cdot \sem{\Gamma},\tau \multimap \sigma) \to
  \mathbf{Int}(\metcppo) (\infty \cdot
  \sem{\Gamma},\tau \multimap \sigma)$ be the
  function given in
  Figure~\ref{fig:interpretation_fuzz}. Then, by
  induction on $m$, we can show that
  \begin{equation*}
    \left(
      I
      \xrightarrow{(k'_{1} \cdot v_{1}) \otimes
        \cdots (k'_{n} \cdot v_{n})}
      \sem{\Delta}
      \xrightarrow{d}
      \sem{\infty \cdot \Gamma}
      \xrightarrow{\varphi^{m}}
      \sem{\tau \multimap \sigma},
      \fix{\tau,\sigma}{f,x,M}[V_{1}/x_{1},\ldots,V_{n}/x_{n}]
    \right)
  \end{equation*}
  is an element of
  $P_{\tau \multimap \sigma}$. Since
  $P_{\tau \multimap \tau'}$ is closed under least
  upper bounds on the first component, we obtain
  the claim.
\end{proof}
\begin{theorem}\label{thm:fuzz-dctx<=dint}
  For any pair of terms $\Gamma \vdash M :\tau$
  and $\Gamma \vdash N : \tau$, we have
  $\dobs_{\Gamma,\tau}(M,N) \leq
  \dint(M,N)$.
\end{theorem}
\begin{proof}
  We can prove the statement in the same way with
  Theorem~\ref{thm:fuzz-dctx<=dden}.
\end{proof}

\begin{figure}
  \centering
  \begin{align*}
    &\sem{\Gamma \vdash \const{a} : \mathbf{R}}
    =
    \sem{\Gamma} \xrightarrow{w}
    I
    \xrightarrow{\lfloor a \rfloor}
    (R,I)
    \xrightarrow{\eta}
    T(R,I)
    \\
    &\sem{\Gamma \vdash \ast : \mathbf{I}}
    =
    \sem{\Gamma} \xrightarrow{w}
    I
    \xrightarrow{\eta}
    TI \\
    &\sem{\Gamma_{1}+\Gamma_{2} \vdash
      \const{f}(M_{1},M_{2}) : \mathbf{R}}
    = \\
    &\quad
    \sem{\Gamma_{1}+\Gamma_{2}}
    \xrightarrow{c}
    \sem{\Gamma_{1}} \otimes
    \sem{\Gamma_{2}}
    \xrightarrow{\sem{M_{1}} \otimes \sem{M_{2}}}
    T\sem{\mathbf{R}} \otimes
    T\sem{\mathbf{R}}
    \xrightarrow{\mathrm{dstr}}
    T(\sem{\mathbf{R}} \otimes \sem{\mathbf{R}})
    \xrightarrow{T\lfloor f \rfloor}
    T\sem{\mathbf{R}}
    \\
    &\sem{x:_{n}\tau \vdash x:\tau}
    = n \cdot \sem{\tau}
    \xrightarrow{d}
    \sem{\tau}
    \xrightarrow{\eta}
    T\sem{\tau} \\
    &\sem{\Delta \vdash \oc M : \oc_{n}\tau}
    =
    \sem{\Delta}
    \xrightarrow{\tilde{d}}
    n \cdot \sem{\Gamma}
    \xrightarrow{n \cdot \sem{M}}
    n \cdot T\sem{\tau}
    \xrightarrow{\xi}
    T\sem{\oc_{n} \tau}
    \\
    &
    \sem{\Delta \vdash \fix{\tau,\sigma}{
        f,x,M
      } : \tau \multimap \sigma}
    =
    \sem{\Delta}
    \xrightarrow{\tilde{d}}
    \infty \cdot \sem{\Gamma}
    \xrightarrow{\text{the least fixed point of $\varphi$}}
    \tau \multimap \sigma
    \\
    &\sem{\Gamma+\Delta \vdash
      M + N : \mathbf{R}}
    = \\
    &\qquad
    \sem{\Gamma+\Delta}
    \xrightarrow{c}
    \sem{\Gamma} \otimes
    \sem{\Delta}
    \xrightarrow{\sem{M} \otimes \sem{N}}
    T\sem{\mathbf{R}} \otimes
    T\sem{\mathbf{R}}
    \xrightarrow{\mathrm{dstr}}
    T(\sem{\mathbf{R}} \otimes \sem{\mathbf{R}})
    \xrightarrow{T(+)}
    T\sem{\mathbf{R}}
    \\
    &
    \sem{\Delta \vdash \const{a} \cdot M :\mathbf{R}}
    =
    \sem{\Delta}
    \xrightarrow{\tilde{d}}
    \sem{\Gamma}
    \xrightarrow{\sem{M}}
    n \cdot \sem{\mathbf{R}}
    \xrightarrow{\mathrm{mult}_{n,a}}
    \sem{\mathbf{R}}
    \xrightarrow{\eta}
    T\sem{\mathbf{R}}
    \\
    &\sem{\Gamma \vdash \lambda x:\sigma.\,
      M : \sigma \multimap \tau}
    = (\text{the currying of} \
    \sem{\Gamma} \otimes \sem{\sigma}
    \xrightarrow{\sem{M}}
    \sem{\tau})
    \xrightarrow{\eta}
    T\sem{\sigma \multimap \tau}
    \\
    &\sem{\Gamma + \Delta \vdash
      M \, N : \tau}
    = \sem{\Gamma + \Delta}
    \xrightarrow{c}
    \sem{\Gamma}
    \otimes
    \sem{\Delta}
    \xrightarrow{\sem{M} \otimes \sem{N}}
    \sem{\tau \multimap \sigma}
    \otimes
    \sem{\tau}
    \xrightarrow{\text{eval}}
    T\sem{\sigma}
    \\
    &\sem{\Gamma + \Delta \vdash
      M \otimes N : \tau \otimes \sigma}
    = \sem{\Gamma + \Delta}
    \xrightarrow{c}
    \sem{\Gamma}
    \otimes
    \sem{\Delta}
    \xrightarrow{\sem{M} \otimes \sem{N}}
    T\sem{\tau} \otimes T\sem{\sigma}
    \xrightarrow{\mathrm{dstr}}
    T\sem{\tau \otimes \sigma}
    \\
    &\sem{\Delta + \Xi \vdash
      \letin{\ast}{M}{N} : \tau}
    =
    \sem{\Delta + \Xi}
    \xrightarrow{c}
    \sem{\Xi} \otimes \sem{\Delta}
    \xrightarrow{\tilde{d}}
    \sem{\Xi} \otimes n \cdot \sem{\Gamma} 
    \xrightarrow{\sem{N} \otimes ((n \cdot \sem{M});\xi)}
    \\
    &\quad
    T\sem{\tau} \otimes TI
    \xrightarrow{\mathrm{dstr}}
    T\sem{\tau}
    \\
    &\sem{\Delta + \Xi \vdash
      \letin{x \otimes y}{M}{N} : \tau}
    =
    \sem{\Delta + \Xi}
    \xrightarrow{c}
    \sem{\Xi} \otimes \sem{\Delta}
    \xrightarrow{\tilde{d}}
    \sem{\Xi} \otimes n \cdot \sem{\Gamma}
    \xrightarrow{\sem{\Xi} \otimes ((n \cdot \sem{M});\xi)}
    \\
    &\quad
    \sem{\Xi} \otimes
    T(n \cdot (\sem{\sigma_{1}} \otimes
    \sem{\sigma_{2}}))
    \xrightarrow{\cong}
    \sem{\Xi} \otimes 
    T(n \cdot \sem{\sigma_{1}} \otimes
    n \cdot \sem{\sigma_{2}})
    \xrightarrow{\mathrm{str};\sem{N};\mu}
    T\sem{\tau}
    \\
    &\sem{\Delta + \Xi \vdash
      \letin{\oc x}{M}{N} : \tau}
    =
    \sem{\Delta + \Xi}
    \xrightarrow{c}
    \sem{\Xi} \otimes \sem{\Delta}
    \xrightarrow{\tilde{d}}
    \sem{\Xi} \otimes n \cdot \sem{\Gamma}
    \xrightarrow{\sem{\Xi} \otimes ((n \cdot \sem{M});\xi)}
    \\
    &\quad
    \sem{\Xi} \otimes T(n \cdot m \cdot \sem{\sigma})
    \xrightarrow{\cong}
    \sem{\Xi} \otimes T(nm \cdot \sem{\sigma})
    \xrightarrow{\mathrm{str};\sem{N};\mu}
    T\sem{\tau}
    \\
  \end{align*}
  where
  $\varphi \colon \mathbf{Int}(\metcppo)
  (\infty \cdot \sem{\Gamma},\tau \multimap
  \sigma) \to \mathbf{Int}(\metcppo)
  (\infty \cdot \sem{\Gamma},\tau \multimap
  \sigma)$ is given by
  \begin{multline*}
    \varphi(f) =
    \infty \cdot \sem{\Gamma}
    \xrightarrow{c}
    \infty \cdot \sem{\Gamma}
    \otimes
    \infty \cdot \sem{\Gamma}
    \xrightarrow{d \otimes \delta}
    \sem{\Gamma}
    \otimes
    \infty \cdot \infty \cdot \sem{\Gamma}
    \\
    \xrightarrow{\sem{\Gamma}
      \otimes \infty \cdot f}
    \sem{\Gamma}
    \otimes
    \infty \cdot (\tau \multimap \sigma)
    \xrightarrow{\text{the currying of $\sem{M}$}}
    \tau \multimap \sigma
  \end{multline*}
  \caption{Interpretation of Terms}
  \label{fig:interpretation_fuzz}
\end{figure}

}

\section{Conclusion}

In this paper we study quantitative reasoning
about linearly typed higher-order programs. We
introduce a notion of admissibility for families
of metrics on a purely linear programming language
$\LL{S}$, and among them, we investigate five
notions of program metrics and how these are
related, namely the logical metric, observational
metric, equational metric, denotational metric and
interactive metric. Some of our results can be
seen as quantitative analogues of well-known
results about program equivalences: the
observational metric is less discriminating than
or equal to semantic metrics, and non-definable
functionals in the semantics are the source of
inclusions. Existence of fully abstract semantics
for $\LL{S}$ with respect to the observational
metric is left open. Our study reveals the
intrinsic difficulty of comparing denotational
models with interactive semantic models obtained
by applying the Int-construction. Indeed, their
relationship is not trivial already at the level
of program equivalences. It follows from
\cite{hasegawa_2009} that there is a symmetric
monoidal coreflection between
$\mathbf{Int}(\mathbf{MetCppo})$ and
$\mathbf{MetCppo}$. This is a strong connection
between these models. However, we do not know
whether this categorical structure sheds light on
their relationship at the level of higher-order
programs.


Some of our results can be extended to a fragment
of Fuzz where grading is restricted to extended
natural numbers. Providing a quantitative
equational theory and an interactive metric for
full Fuzz is another very interesting topic for
 future work. There are some notions of metric
that we have not taken into account in this paper.
In \cite{10.1145/3209108.3209149}, Gavazzo gives
coinductively defined metrics for an extension
of Fuzz with algebraic effects and recursive types, 
which we do not consider here. 
The so-called observational quotient \cite{Hyland:2000aa} is a way to
construct less discriminating program metrics from
fine-grained ones. A thorough comparison of these notions
of program distance with the ones we introduce here
is another intriguing problem on which we plan to
work in the future.


\bibliographystyle{splncs04}
\bibliography{mybibliography}

\end{document}